\documentclass[a4paper]{article}

\usepackage{authblk}

\usepackage{hyperref}

\usepackage{leftidx}
\usepackage{mathtools} 
\usepackage{paralist}
\usepackage{todonotes}
\usepackage{tikz} 
\usetikzlibrary{arrows}
\usetikzlibrary{arrows,automata}
\usetikzlibrary{automata,positioning,arrows}
\usetikzlibrary{shapes}
\usetikzlibrary{shapes,backgrounds}
\usetikzlibrary{matrix}
\usetikzlibrary{chains,fit,shapes}
\usepackage{amsmath}
\usepackage{amssymb}
\usepackage{amsfonts}
\usepackage{amsthm}
\usepackage{wrapfig}
\usepackage{url}
\usepackage{algorithm2e}

\usepackage{stmaryrd}

\usepackage{tensor}

\usepackage{subfig}

\newcommand{\open}[1]{\tensor[^{#1}]{\triangleright}{_{}}}
\newcommand{\close}[1]{\tensor[_{}]{\triangleleft}{^{#1}}}

\newcommand{\DBfont}[1]{\mathcal{#1}}

\DeclareMathOperator{\bigOstar}{O^*}
\newcommand{\poly}{\textsf{poly}}

\DeclareMathOperator{\reflang}{\mathcal{L}_{\refmark}}
\newcommand{\reflangimagebound}[1]{\mathcal{L}^{\leq #1}_{\refmark}}
\newcommand{\reflangvarmap}[1]{\mathcal{L}^{#1}_{\refmark}}
\DeclareMathOperator{\refmark}{\textsf{ref}}

\DeclareMathOperator{\equrel}{\mathsf{er}}

\DeclareMathOperator{\NFAintQuery}{\alpha_{\textsf{ni}}}
\DeclareMathOperator{\NFAintQueryk}{\alpha^k_{\textsf{ni}}}

\DeclareMathOperator{\RE}{\mathsf{RE}}
\DeclareMathOperator{\var}{\textsf{var}}

\newcommand{\varmap}[1]{\mathsf{vmap}_{#1}}
\newcommand{\varmapX}[2]{\mathsf{vmap}_{#1, #2}}
\newcommand{\varprefix}[1]{\triangledown_{#1}}
\newcommand{\interxregex}[1]{\langle#1\rangle_{\mathsf{int}}}

\DeclareMathOperator{\DBD}{\DBfont{D}}

\DeclareMathOperator{\npclass}{\mathsf{NP}}

\DeclareMathOperator{\bigO}{O}

\DeclareMathOperator{\pspaceclass}{\mathsf{PSpace}}

\DeclareMathOperator{\expspaceclass}{\mathsf{ExpSpace}}

\DeclareMathOperator{\nlclass}{\mathsf{NL}}

\DeclareMathOperator{\emptyword}{\varepsilon}

\DeclareMathOperator{\booleProb}{\textsc{Bool-Eval}}

\DeclareMathOperator{\checkProb}{\textsc{Check}}

\newcommand{\CLPQ}{\Re\text{-}\mathsf{CPQ}}

\newcommand{\unionsof}[1]{\cup\text{-}#1}

\newcommand{\RPQ}{\mathsf{RPQ}}
\newcommand{\CRPQ}{\mathsf{CRPQ}}
\newcommand{\ECRPQ}{\mathsf{ECRPQ}}
\newcommand{\eqECRPQ}{\mathsf{ECRPQ}^{\equrel}}

\newcommand{\bisCXRPQ}[1]{\mathsf{CXRPQ}^{\leq #1}}
\newcommand{\logCXRPQ}{\mathsf{CXRPQ}^{\log}}

\newcommand{\CXRPQ}{\mathsf{CXRPQ}}
\newcommand{\vsfCXRPQ}{\mathsf{CXRPQ}^{\mathsf{vsf}}}
\newcommand{\vsfpbCXRPQ}{\mathsf{CXRPQ}^{\mathsf{vsf}, \mathsf{fl}}}

\DeclareMathOperator{\NFA}{\mathsf{NFA}}

\DeclareMathOperator{\eword}{\varepsilon}

\DeclareMathOperator{\lang}{\mathcal{L}}

\newcommand{\langimagebound}[1]{\mathcal{L}^{\leq #1}}
\newcommand{\langvarmap}[1]{\mathcal{L}^{#1}}

\DeclareMathOperator{\altop}{\vee}

\DeclareMathOperator{\xregex}{\mathsf{XRE}}
\DeclareMathOperator{\conxregex}{\mathsf{CXRE}}
\newcommand{\dimconxregex}[1]{#1$-$\mathsf{CXRE}}
\newcommand{\deref}{\mathsf{deref}}

\newcommand{\varset}{\ensuremath{\mathcal{X}_s}}

\newcommand{\varsx}{\ensuremath{\mathsf{x}}}
\newcommand{\varsy}{\ensuremath{\mathsf{y}}}
\newcommand{\varsz}{\ensuremath{\mathsf{z}}}
\newcommand{\varsu}{\ensuremath{\mathsf{u}}}

\newcommand{\nodevarset}{\ensuremath{\mathcal{X}_n}}

\newcommand{\varnx}{\ensuremath{\mathit{x}}}
\newcommand{\varny}{\ensuremath{\mathit{y}}}
\newcommand{\varnz}{\ensuremath{\mathit{z}}}

\newcommand{\ta}{\ensuremath{\mathtt{a}}}
\newcommand{\tb}{\ensuremath{\mathtt{b}}}
\newcommand{\tc}{\ensuremath{\mathtt{c}}}
\newcommand{\td}{\ensuremath{\mathtt{d}}}

\newcommand{\tone}{\ensuremath{\mathtt{1}}}
\newcommand{\tzero}{\ensuremath{\mathtt{0}}}

\newcommand{\tp}{\ensuremath{\mathtt{p}}}
\newcommand{\ts}{\ensuremath{\mathtt{s}}}

\newtheorem{lemma}{Lemma}
\newtheorem{theorem}{Theorem}
\newtheorem{proposition}{Proposition}
\newtheorem{definition}{Definition}
\newtheorem{example}{Example}
\newtheorem{corollary}{Corollary}

\begin{document}

\title{Conjunctive Regular Path Queries with String Variables}

\author[1]{Markus L. Schmid}

\affil[1]{Humboldt-Universit\"at zu Berlin, Unter den Linden 6, D-10099, Berlin, Germany, \texttt{MLSchmid@MLSchmid.de}}

\maketitle

\begin{abstract}
We introduce the class $\CXRPQ$ of conjunctive xregex path queries, which are obtained from conjunctive regular path queries ($\CRPQ$s) by adding string variables (also called backreferences) as found in practical implementations of regular expressions. $\CXRPQ$s can be considered user-friendly, since they combine two concepts that are well-established in practice: pattern-based graph queries and regular expressions with backreferences. Due to the string variables, $\CXRPQ$s can express inter-path dependencies, which are not expressible by $\CRPQ$s. The evaluation complexity of $\CXRPQ$s, if not further restricted, is $\pspaceclass$-hard in data-complexity. We identify three natural fragments with more acceptable evaluation complexity: their data-complexity is in $\nlclass$, while their combined complexity varies between $\expspaceclass$, $\pspaceclass$ and $\npclass$. In terms of expressive power, we compare the $\CXRPQ$-fragments with $\CRPQ$s and unions of $\CRPQ$s, and with extended conjunctive regular path queries ($\ECRPQ$s) and unions of $\ECRPQ$s.
\end{abstract}

\section{Introduction}

The popularity of graph databases (commonly abstracted as directed, edge-labelled graphs) is due to their applicability in a variety of settings where the underlying data is naturally represented as graphs, e.\,g., Semantic Web and social networks, biological data, chemical structure analysis, pattern recognition, network traffic, crime detection, object oriented data. The problem of querying graph-structured data has been 
studied over the last three decades and still receives a lot of attention. 
%
%
%
%
For more background information, we refer to the introductions of the recent papers~\cite{LibkinEtAl2016, CzerwinskiEtAl2018, BarceloEtAl2012, BarceloEtAl2014}, and to the survey papers~\cite{AnglesEtAl2017, Barcelo2013, Wood2012, AnglesEtAl2019}.\par
Many query languages for graph databases (for practical systems as well as those studied in academia) follow an elegant and natural declarative approach: a query is described by a \emph{graph pattern}, i.\,e., a graph $G = (V, E)$ with edge labels that represent some path-specifications. The evaluation of such a query consists in \emph{matching} it to the graph database $\DBD = (V_{\DBD}, E_{\DBD})$, i.\,e., finding a mapping $h : V \to V_{\DBD}$, such that, for every $(x, s, y) \in E$, in $\DBD$ there is a path from $h(x)$ to $h(y)$ whose edge labels satisfy the path-specification $s$. In the literature, such query languages are also called \emph{pattern-based}.
Let us now briefly summarise where this concept can be found in theory and practice. \par
The most simple graph-patterns (called \emph{basic} in~\cite{AnglesEtAl2017}) have just fixed relations (i.\,e., edge-labels) from the graph database as their edge labels. A natural extension are \emph{wildcards}, which can match any edge-label of the database (e.\,g., as described in~\cite{CzerwinskiEtAl2018}), or \emph{label variables}, which are like wildcards, but different occurrences of the same variable must match the same label (see, e.\,g.,~\cite{BarceloEtAl2014}). It is common to extend such basic graph patterns with relational features like, e.\,g., projection, union, and difference (see~\cite{AnglesEtAl2017}).
In order to implement \emph{navigational features} that can describe more complex connectivities between nodes via longer paths instead of only single arcs, we need more complicated path specifications.\par
Navigational features are popular, since they allow to query the \emph{topology} of the data and, if transitivity can be described, exceed the power of the basic relational query languages. 
Using regular expressions as path specifications is the most common way of implementing navigational features. The \emph{regular path queries} ($\RPQ$s) given by \emph{single-edge} graph patterns $(\{x, y\}, \{(x, s ,y)\})$, where $s$ is a regular expression, can be considered the simplest navigational graph patterns. General graph patterns labelled by regular expressions are called \emph{conjunctive regular path queries} ($\CRPQ$s). \par

\begin{figure}[h]
\begin{center}
\scalebox{1.4}{\includegraphics{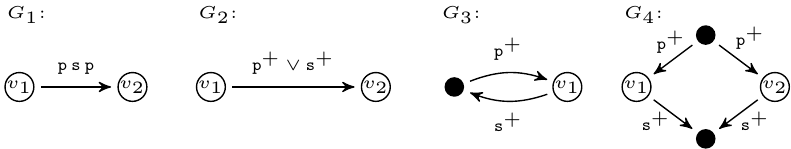}}
\end{center}
\caption{Simple graph patterns.}
\label{GraphPatternsExampleFigureOne}
\end{figure}

For example, consider a graph database with nodes representing persons, arcs $(u, \tp, v)$ meaning ``$u$ is a (biological) parent of v'' and arcs $(u, \ts, v)$ meaning ``$v$ is $u$'s PhD-supervisor''. We consider the graph patterns from Figure~\ref{GraphPatternsExampleFigureOne} (labelled nodes are considered as free variables of the query). Then $G_1$ describes pairs $(v_1, v_2)$, where $v_1$'s child has been supervised by $v_2$'s parent; $G_2$ describes pairs $(v_1, v_2)$, where $v_1$ is a biological ancestor or an academical descendant of $v_2$; $G_3$ describes $v_1$ that have a biological ancestor that is also their academical ancestor; $G_4$ describes pairs $(v_1, v_2)$, where $v_1$ and $v_2$ are biologically related as well as academically. Note that $G_1, G_2$ represent $\RPQ$s, while $G_3, G_4$ represent $\CRPQ$s. \par


The classes of $\RPQ$s and $\CRPQ$s (and modifications of them) have been intensively studied in the literature (see, e.\,g., \cite{CruzEtAl1987, CalvaneseEtAl2000, CalvaneseEtAl2003, ConsensMendelzon1990, BarceloEtAl2016, FlorescuEtAl1998, MendelzonWood1995, AbiteboulEtAl1997, AbiteboulEtAl1999, DeutschTannen2001}). The former can be evaluated efficiently (see, e.\,g.,~\cite{ConsensMendelzon1990}), while evaluation for the latter it is $\npclass$-complete in combined-complexity, but $\nlclass$-complete in data-complexity (see~\cite{BarceloEtAl2012}).\par
Despite their long-standing investigation, these  
basic classes still pose several challenges that are currently studied. For example,~\cite{LosemannMartens2013, MartensTrautner2018, MartensEtAl2020} provide an in-depth analysis of the complexity of $\RPQ$s for different path semantics. So far in this introduction, we implicitly assumed \emph{arbitrary path} semantics, but since there are potentially infinitely many such paths, query languages that also retrieve paths often restrict this by considering simple paths or trails. However, such semantics make the evaluation of $\RPQ$s much more difficult (see~\cite{LosemannMartens2013, MartensTrautner2018, MartensEtAl2020} for details). Much effort has also been spent on extending $\RPQ$s and $\CRPQ$s to the setting where the data-elements stored at nodes of the graph database can also be queried (see~\cite{LibkinEtAl2016, Kostylev2018}). In~\cite{BarceloEtAl2014}, the authors represent partially defined graph data by graph patterns and then query them with $\CRPQ$s (among others). In the very recent paper~\cite{BarceloEtAl2019}, the authors study the boundedness problem for unions of $\CRPQ$s (i.\,e., the problem of finding an equivalent union of (relational) conjunctive queries). \par
Also in the practical world, pattern-based query languages for graph databases play a central role. Most prominently, the \emph{W3C Recommendation} for \emph{SPARQL 1.1} ``is based around graph pattern matching'' (as stated in Section $5$ of~\cite{w3c:sparql11}), and \emph{Neo4J Cypher} also uses graph patterns as a core functionality (see~\cite{Cypher}). Moreover, the graph computing framework \emph{Apache TinkerPop$^{\text{TM}}$} contains the graph database query language \emph{Gremlin}~\cite{Gremlin}, which is more based on the navigational graph traversal aspect, but nevertheless supports pattern-based query mechanisms. Note that~\cite{AnglesEtAl2017} surveys the main features of these three languages. 

\subsection{Main Goal of this Work}

$\CRPQ$s are 
not expressive enough for many natural querying tasks (see, e.\,g., the introduction of~\cite{BarceloEtAl2012}). The most obvious shortcoming is that we cannot express any \emph{inter-path} dependencies, i.\,e., relations between the paths of the database that are matched by the arcs of the graph pattern, except that they must start or end with the same node. The main goal of this work is to extend $\CRPQ$s in order to properly increase their expressive power. In particular, we want to meet the following objectives:
\begin{enumerate}
\item\label{objone} The increased expressive power should be reasonable, i.\,e., it should cover natural and relevant querying tasks. 
\item\label{objtwo} The extensions should be user-friendly, i.\,e., the obtained query language should be intuitive.
\item\label{objthree} The evaluation complexity should still be acceptable. 
\end{enumerate}

The main idea is to allow \emph{string variables} in the edge labels of the graph patterns. For example, in $G_1$ of Figure~\ref{GraphPatternsExampleFigureTwo}, the $\varsx\{\ta \altop \tb\}$ label sets variable $\varsx$ to some word matched by regular expression $\ta \altop \tb$ and the occurrence of $\varsx$ on the other edge label then refers to the value of $\varsx$. Hence, $G_1$ describes all pairs $(v_1, v_2)$ such that $v_1$ has a direct $\ta$-predecessor that has $v_2$ as a transitive successor with respect to $\ta$ or $\tc$, or $v_1$ has a direct $\tb$-predecessor that has $v_2$ as a transitive successor with respect to $\tb$ or $\tc$.
Similarly, $G_2$ of Figure~\ref{GraphPatternsExampleFigureTwo} describes triangles $(v_1, v_2, v_3)$ with a complicated connectivity relation: $v_1$ reaches $v_2$ with $\ta \ta$ or $\tb$, $v_2$ reaches $v_3$ with some path labelled only with symbols different than $\ta$ and $\tb$ ($\Sigma$ is the set of edge labels), while $v_3$ reaches $v_1$ either in the same way as $v_1$ reaches $v_2$, or in the same way as $v_2$ reaches $v_3$. \par

\begin{figure}[h]
\begin{center}
\scalebox{1.4}{\includegraphics{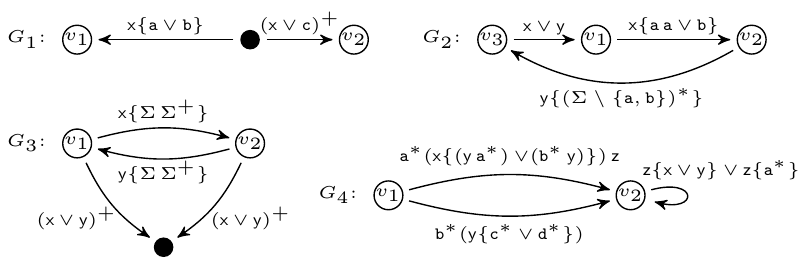}}
\end{center}
\caption{$\CRPQ$s with string variables.}
\label{GraphPatternsExampleFigureTwo}
\end{figure}

The graph patterns $G_1$ and $G_2$ of Figure~\ref{GraphPatternsExampleFigureTwo} could also be considered as $\CRPQ$s with some syntactic sugar. More precisely, both these graph patterns can be transformed into a union of $\CRPQ$s by simply ``spelling out'' all combinations that are possible for the variables $\varsx$ and $\varsy$. However, we can easily build examples that would translate into an exponential number of $\CRPQ$s, and, moreover, it can be argued that the necessity for explicitly listing \emph{all} possible combinations that can conveniently be stated in a concise way, is exactly what a user should not be bothered with. \par
We move on to an example, where a simple application of string variables adds substantial expressive power to a $\CRPQ$. Let us assume that the nodes in a graph database represent persons and arcs represent text messages sent by mobile phone (let $\Sigma$ be the set of messages). The idea is that some individuals try to hide their direct communication by encoding their messages by sequences of simple text messages that are send via intermediate senders and receivers. In particular, we want to discover individuals who are likely to be involved in such a hidden communication network. In this regard, $G_3$ of Figure~\ref{GraphPatternsExampleFigureTwo} describes pairs $(v_1, v_2)$ such that $v_1$ reaches $v_2$ (and $v_2$ reaches $v_1$) by a sequence $\varsx$ ($\varsy$, respectively) of at least $2$ messages, and there is some person that has been contacted by $v_1$ and $v_2$ by paths that are repetitions of these message-sequences. Note that, in this example, both the length of the message-sequences $\varsx$ and $\varsy$, as well as the number of their repetitions in order to reach the mutual friend of $v_1$ and $v_2$, are unbounded. \par
Finally, $G_4$ of Figure~\ref{GraphPatternsExampleFigureTwo} shows a feature that has not been used in the previous examples: references to variables could also occur in the definitions of other variables, e.\,g., $\varsy$ is defined on one edge, and used in the definition of both $\varsx$ and $\varsz$ on the other two edges. Moreover, note that the same variable $\varsz$ has two definitions, which are mutually exclusive and therefore do not cause ambiguities.\par
This formalism obviously extends $\CRPQ$; moreover, it is easy to see that it also covers wildcards and edge variables described above, as well as the fragment of \emph{extended conjunctive regular path queries} ($\ECRPQ$s)~\cite{BarceloEtAl2012} that only have equality-relations as non-unary relations ($\ECRPQ$s shall be discussed in more detail below). 

\subsection{Conjunctive Xregex Path Queries}

Regular expressions of the kind used in the graph patterns of Figure~\ref{GraphPatternsExampleFigureTwo} are actually a well-established concept, which, in the theoretical literature, is usually called \emph{regex} or \emph{xregex}, and string variables are often called \emph{backreferences}.
Xregex have been investigated in the formal language community~\cite{Schmid2016, Schmid2013, FreydenbergerSchmidJCSS, CampeanuEtAl2003, Freydenberger2013}, and, despite the fact that allowing them in regular expressions has many negative consequences (see~\cite{Aho1990, FernauSchmid2015, FernauEtAl2016, FernauEtAl2015, Freydenberger2013}), regular expression libraries of almost all modern programming languages (like, e.\,g., Java, PERL, Python and .NET) support backreferences (although they syntactically and even semantically slightly differ from each other (see the discussion in~\cite{FreydenbergerSchmidJCSS})), and they are even part of the POSIX standard~\cite{ieee:posix}. The syntax of xregex is quite intuitive: on top of normal regular expressions, we can \emph{define} variables by the construct $\varsx\{\ldots\}$ and \emph{reference} them by occurrences of $\varsx$ (see Figure~\ref{GraphPatternsExampleFigureTwo}). Formally defining their semantics is more tricky (mainly due to nested variable definitions and undefined variables), but for simple xregex their meaning is intuitively clear. \par
Obviously, we can define various query classes by replacing the regular expressions in $\CRPQ$ by some more powerful language descriptors. However, this will not remedy the lack of a means to describe inter-path dependencies, and, furthermore, regular expressions seem powerful enough to describe the desired navigational features (in fact, the analyses carried out in~\cite{BonifatiEtAl2017, BonifatiEtAl2019, MartensTrautner2019} suggest that the regular expressions used in practical queries are even rather simple).
What makes xregex interesting is that defining a variable on some edge and referencing it on another is a convenient way of describing inter-path dependencies, while, syntactically, staying in the realm of graph patterns with regular expressions. \par
To define such a query class, we first have to lift xregex to (multi-dimensional) \emph{conjunctive xregex}, i.\,e., tuples $\bar{\alpha} = (\alpha_1, \alpha_2, \ldots, \alpha_m)$ of xregex that can generate tuples $\bar{w} = (w_1, w_2, \ldots, w_m)$ of words, but such that the $w_i$ match the $\alpha_i$ in a \emph{conjunctive} way with respect to the string variables (i.\,e., occurrences of the same variable $\varsx$ in different $\alpha_i$ and $\alpha_j$ must refer to the same string). Labeling graph patterns with xregex and interpreting the edge labels as conjunctive xregex yields our class of \emph{conjunctive xregex path queries} ($\CXRPQ$s).\par
With respect to Objective~\ref{objtwo} from above, we note that $\CXRPQ$s are purely pattern-based, i.\,e., the queries are just graph patterns with edge labels (as pointed out above, such queries are widely adapted in practice and we can therefore assume their user-friendliness), and the xregex used as path-specifications are a well-known practical tool (in fact, the use of string variables (backreferences) in regular expressions is a topic covered by standard textbooks on practical application of regular expressions (see, e.\,g.,~\cite{FriedlBookThirdEdition})). Regarding Objective~\ref{objone}, string variables add some mechanism to describe inter-path relationships and we already saw some illustrating examples. We will further argue in favour of their expressive power in the remainder of the introduction (see also Figure~\ref{expPowerInclusionMap}).

\subsection{Barcelo et al.'s $\ECRPQ$}

An existing class of graph queries that is suitable for comparison with the class $\CXRPQ$ are the extended conjunctive regular path queries ($\ECRPQ$s) introduced in~\cite{BarceloEtAl2012}. While $\CRPQ$s can be seen as graph patterns with unary regular relations on the edges (i.\,e., the regular expressions), the class $\ECRPQ$ permits regular relations of arbitrary arity, which allows to formulate a wide range of inter-path dependencies. $\ECRPQ$s have acceptable evaluation complexity: $\nlclass$-complete in data-complexity and $\pspaceclass$-complete in combined-complexity (see~\cite{BarceloEtAl2012} for details). \par
In terms of expressive power, the class of $\CXRPQ$s completely cover the fragment $\eqECRPQ$ of those $\ECRPQ$s that have only unary relations or equality relations (i.\,e., relations requiring certain paths to be equal). 
On the other hand, we can show that there are $\CXRPQ$ that are not expressible by $\eqECRPQ$. From an intuitive point of view $\CXRPQ$ and $\ECRPQ$ are incomparable in the sense that $\ECRPQ$ can describe inter-path dependencies beyond simple equality, while $\CXRPQ$ can describe equality of arbitrarily many paths as, e.\,g., by $(\varsx \altop \varsy)^+$ in the query $G_3$ of Figure~\ref{GraphPatternsExampleFigureTwo}. \par
The class of $\ECRPQ$ is not purely pattern-based anymore, since also the relations must be given as regular expressions. In this regard, the authors of~\cite{BarceloEtAl2012} mention that ``[\dots] specifying regular relations with regular expressions is probably less natural than specifying regular languages (at least it is more cumbersome) [\dots]'' and suggest that any practical standard would rather provide some reasonable regular relations as built-in predicates. 

\subsection{Technical Contributions} 

In terms of Objective~\ref{objthree}, the best evaluation complexity that we can hope for is $\nlclass$ in data-complexity and $\npclass$ in combined-complexity (since $\CRPQ$s have these lower bounds).
Higher combined-complexity (e.\,g., $\pspaceclass$) is still acceptable, as long as the optimum of $\nlclass$ in data-complexity is reached (this makes sense if we can assume our queries to be rather small in comparison to the data.\par
Unfortunately, $\CXRPQ$s have a surprisingly high evaluation complexity: even for the fixed xregex $\NFAintQuery = \# \varsx\{(\ta \altop \tb)^*\} \, (\#\# \, \varsx)^* \#\#\#$, it is $\pspaceclass$-hard to decide whether a given graph database contains a path labelled by a word from $\lang(\NFAintQuery)$ (so Boolean evaluation is $\pspaceclass$-hard in data complexity).
This hardness result has nevertheless a silver lining: it directly points us to restrictions of $\CXRPQ$ that might lead to more tractable fragments. More precisely, for $\pspaceclass$-hardness it seems vital that references for variable $\varsx$ are subject to the star-operator, and that the variable $\varsx$ can store words of unbounded length.
%
%
Our main positive result will be that by restricting $\CXRPQ$s accordingly, we can tame their evaluation complexity and obtain more tractable fragments.\par
Let $\vsfCXRPQ$ be the class of \emph{variable-star free} $\CXRPQ$s, i.\,e., no variable can be used under a star or plus (but normal symbols still can). For example, $G_2, G_4$ of Figure~\ref{GraphPatternsExampleFigureTwo} are in $\vsfCXRPQ$. This restriction is enough to make the data-complexity drop from $\pspaceclass$-hardness to the optimum of $\nlclass$-completeness (although combined-complexity is $\expspaceclass$). 
The upper bound is obtained by showing that 
$q \in \vsfCXRPQ$ can be transformed into equivalent $q'$ in a certain \emph{normal form}, which can be evaluated in nondeterministic space $\bigO(|q'| \log(|\DBD|))$.
However, $|q'| = \bigO(2^{2^{|q|}})$ and we only get the single exponential space upper bound for combined-complexity by handling one exponential size blow-up with nondeterminism. 
A closer look at the normal form construction reveals that the exponential size blow-up is caused by chains of the following form: a reference of $\varsx$ occurs in the definition of $\varsy$, a reference of $\varsy$ occurs in the definition of $\varsz$, a reference of $\varsz$ occurs in the definition of $\varsu$ and so on (e.\,g., this happens with respect to $\varsx$, $\varsy$ and $\varsz$ in $G_4$ of Figure~\ref{GraphPatternsExampleFigureTwo}). If we require that every variable definition only contains references of variables that themselves have definitions without variables, then we obtain the fragment $\vsfpbCXRPQ$ which have normal forms of polynomial size and therefore the combined complexity drops to $\pspaceclass$ (i.\,e., the same complexity as $\ECRPQ$). \par
The second successful approach is to add a constant upper-bound $k$ on the \emph{image size} of $\CXRPQ$s, i.\,e., the length of the words stored in variables. Let $\bisCXRPQ{k}$ be the corresponding fragments. For every $\bisCXRPQ{k}$, the evaluation complexity drops to the optimum of $\nlclass$-completeness in data-complexity and $\npclass$-completeness in combined-complexity (i.\,e., the same as for $\CRPQ$s).
Unlike for $\vsfCXRPQ$, bounding the image size does not impose any syntactical restrictions; it is rather a restriction of how a $\CXRPQ$ can match a graph database. For example, we can treat $G_3$ of Figure~\ref{GraphPatternsExampleFigureTwo} as a $\bisCXRPQ{10}$, i.\,e., we only have a successful match if the paths between $v_1$ and $v_2$ are of length at most $10$ (note that the paths from $v_1$ and $v_2$ to their mutual friend can have unbounded size). For $G_1$ of Figure~\ref{GraphPatternsExampleFigureTwo}, on the other hand, the image size of variables is necessarily bounded by $1$ and therefore it does not matter whether we interpret it as $\CXRPQ$ or $\bisCXRPQ{k}$ with $k \geq 1$. We stress the fact that this bound only applies to strings stored in variables; we can still specify paths of arbitrary length with regular expressions, i.\,e., $\CRPQ \subseteq \bisCXRPQ{k}$.
%
%
Moreover, evaluating $\bisCXRPQ{k}$ is not as simple as just replacing all variables by fixed words of length at most $k$ and then evaluating a $\CRPQ$, since we also have to take care of dependencies between variable definitions. \par
Finally, there are two more noteworthy results about $\bisCXRPQ{k}$ (a negative and a positive one). While in combined-complexity $\CRPQ$ can be evaluated in polynomial-time if the underlying graph pattern is acyclic (see~\cite{BarceloEtAl2016, BarceloEtAl2012, Barcelo2013}), $\bisCXRPQ{k}$ remain $\npclass$-hard in combined-complexity even for single-edge graph patterns (and $k = 1$). This also demonstrates the general difference of $\bisCXRPQ{k}$ and $\CRPQ$. On the positive side, if instead of a constant upper bound, we allow images bounded logarithmically in the size of the database, then the $\npclass$ upper bound in combined-complexity remains, while the data-complexity slightly increases to $\bigO(\log^2(|\DBD|))$.\par
%
%
%
%
%
The question is whether these fragments are still interesting with respect to Objectives~\ref{objone}~and~\ref{objtwo}. We believe the answer is yes. First observe that the restrictions are quite natural: Not using the star-operator over variables is a rule not difficult for users to comply with, if they are familiar with regular expressions; it is also easily to be checked algorithmically, and the same holds for the additional restriction required by $\vsfpbCXRPQ$. 
%
%
The class $\bisCXRPQ{k}$ does not require any syntactical restriction; when interpreting the query result, the user only has to keep in mind that the paths corresponding to images of variables are bounded in length.\par
Regarding expressive power, all these fragments contain non-trivial examples of $\CXRPQ$s. With respect to the examples from Figure~\ref{GraphPatternsExampleFigureTwo}, $G_4 \in \vsfCXRPQ$ and $G_2 \in \vsfpbCXRPQ$; any $\CXRPQ$ can be interpreted as $\bisCXRPQ{k}$ for any $k$. Both $\vsfCXRPQ$ and $\vsfpbCXRPQ$ still cover the fragment $\eqECRPQ$. It is tempting to misinterpret queries from $\bisCXRPQ{k}$ as $\CRPQ$ with mere syntactic sugar (since string variables range over finite sets of words). However, it can be proven that even $\bisCXRPQ{1}$ contains queries that are not expressible by $\CRPQ$s. 

\section{Preliminaries}\label{sec:prelim}

Let $\mathbb{N} = \{1, 2, 3, \ldots\}$ and $[n] = \{1, 2, \ldots, n\}$ for $n \in \mathbb{N}$. 
$A^+$ denotes the set of non-empty words over an alphabet $A$ and $A^* = A^+ \cup \{\eword\}$ (where $\eword$ is the empty word). For a word $w \in A^*$, $|w|$ denotes its length, and for $k \in \mathbb{N}$, $A^{\leq k} = \{w \in A^* \mid |w| \leq k\}$.
For $w_1, w_2, \ldots, w_n \in A^*$, we set $\Pi^n_{i = 1} w_i = w_1 w_2 \ldots w_n$, and if $w = w_i$, for every $i \in [n]$, then we also write $w^n$ instead of $\Pi^n_{i = 1} w_i$. For some alphabet $A$, a word $w \in A^*$ and a $b \in A$, $|w|_b$ is the number of occurrences of symbol $b$ in $w$.\par
We fix a finite \emph{terminal} alphabet $\Sigma$ and an enumerable set $\varset$ of \emph{string variables}, where $\varset \cap \Sigma = \emptyset$. As a convention, we use symbols $\ta, \tb, \tc, \td, \ldots$ for elements from $\Sigma$, and $\varsx, \varsy, \varsz, \varsx_1, \varsx_2, \ldots, \varsy_1, \varsy_2, \ldots$ for variables from $\varset$. We consistently use sans-serif font for string variables to distinguish them from \emph{node variables} to be introduced later.
We use \emph{regular expressions} and (\emph{nondeterministic}) \emph{finite automata} ($\NFA$ for short) as commonly defined in the literature (see Section~\ref{xregexDefinition} and the remainder of this section for more details).

\subsection{Ref-Words} 

The following \emph{ref-words} (first introduced in~\cite{Schmid2016}) are convenient for defining the semantics of xregex (Section~\ref{sec:xregex}). They have also been used in~\cite{FreydenbergerSchmidJCSS} and for so-called document spanners in~\cite{FreydenbergerEtAl2018, Freydenberger2019, DoleschalEtAl2019}. Ref-words will be vital in our definition of conjunctive xregex (Section~\ref{sec:ConXregex}), which are the basis of the class $\CXRPQ$.\par
For every $\varsx \in \varset$, we use the pair of symbols $\open{\varsx}$ and $\close{\varsx}$, which are interpreted as opening and closing parentheses. 

\begin{definition}[Ref-Words]
A \emph{subword-marked} word (\emph{over terminal alphabet $\Sigma$ and variables $\varset$}) is a word $w \in (\Sigma \cup \{\open{\varsx}, \close{\varsx} \mid \varsx \in \varset\} \cup \varset)^*$ that, for every $\varsx \in \varset$, contains the parentheses $\open{\varsx}$ and $\close{\varsx}$ at most once and all these parentheses form a well-formed parenthesised expression. For every $\varsx \in \varset$, a subword $\open{\varsx} v \close{\varsx}$ in $w$ is called a \emph{definition} (of variable $\varsx$), and an occurrence of symbol $\varsx$ is called a \emph{reference} (of variable $\varsx$). For a subword-marked word $w$ over $\Sigma$ and $\varset$, the binary relation $\preceq_w$ over $\varset$ is defined by setting $\varsx \preceq_{w} \varsy$ if in $w$ there is a definition of $\varsy$ that contains a reference or a definition of $\varsx$. A \emph{ref-word} (\emph{over terminal alphabet $\Sigma$ and variables $\varset$}) is a subword-marked word over $\Sigma$ and $\varset$, such that the transitive closure of $\preceq_{w}$ is acyclic. 
\end{definition}

Ref-words are just words over alphabet $\Sigma \cup \varset$, in which some subwords are uniquely marked by means of the parentheses $\open{\varsx}$ and $\close{\varsx}$.
Moreover, the marked subwords are not allowed to overlap, i.\,e., $\open{\varsx} \open{\varsy} \close{\varsx} \close{\varsy}$ must not occur as subsequence. 
For every variable $\varsx \in \varset$, all occurrences of $\varsx$ in a ref-word are interpreted as references to the definition of $\varsx$. Since definitions may contain itself references or definitions of other variables, there are chains of references, e.\,g., the definition of $\varsx$ contains references of $\varsy$, but the definition of $\varsy$ contains references of $\varsz$ and so on. Therefore, in order to make this implicit referencing process terminate, we have to require that it is acyclic, which is done by requiring the transitive closure of $\preceq_w$ to be acyclic. 
For example, $\ta \varsx \tb \open{\varsx} \ta \tb \close{\varsx} \tc \open{\varsy} \varsx \ta \ta \close{\varsy} \varsy$ is a valid ref-word, while $\ta \varsx \tb \open{\varsx} \ta \tb \close{\varsx} \tc \open{\varsy} \varsx \ta \ta \varsy \close{\varsy} \varsy$, or $\ta \varsx \ta \open{\varsx} \ta \varsy \tb \close{\varsx} \tc \open{\varsy} \varsx \ta \close{\varsy}$ are not.\par
For ref-words over $\Sigma$ and $\varset$, it is therefore possible to successively substitute references by their definitions until we obtain a word over $\Sigma$. This can be formalised as follows.

\begin{definition}[Deref-Function]\label{derefDefinition}
For a ref-word $w$ over $\Sigma$ and $\varset$, $\deref(w) \in \Sigma^*$ is obtained from $w$ by the following procedure:
\begin{enumerate} 
\item\label{refStepOne} Delete all occurrences of $\varsx \in \varset$ without definition in $w$.
\item\label{refLoopHead} Repeat until we have obtained a word over $\Sigma$:
\begin{enumerate} 
\item\label{refStepTwo} Let $\open{\varsx} v_{\varsx} \close{\varsx}$ be a definition such that $v_{\varsx} \in \Sigma^*$.  
\item\label{refStepThree} Replace $\open{\varsx} v_{\varsx} \close{\varsx}$ and all occurrences of $\varsx$ in $w$ by $v_{\varsx}$.
\end{enumerate} 
\end{enumerate}
\end{definition}

\begin{proposition}\label{derefProposition}
The function $\deref(w)$ is well-defined.
\end{proposition}

\begin{proof}
If we reach Step~\ref{refLoopHead} without any definition for a variable, then there can also be no reference of a variable, since we deleted all variable references that have no definition in Step~\ref{refStepOne}, and for every variable definition deleted in Step~\ref{refStepThree}, we also deleted all corresponding references. Consequently, if we reach Step~\ref{refLoopHead} without variable definitions, then we have obtained a word over $\Sigma$ and terminate. If, on the other hand, there is at least one definition for some variable $\varsx$ when we reach Step~\ref{refLoopHead}, then, by the definition of a ref-word, there must be at least one definition that does not contain another definition. 
Hence, there must be at least one definition $\open{\varsx} v_{\varsx} \close{\varsx}$ such that $v_{\varsx} \in \Sigma^*$, as required by Step~\ref{refStepTwo}. In Step~\ref{refStepThree} this definition of $\varsx$ and all its references are then replaced by $v_{\varsx}$, which is a word over $\Sigma$. Moreover, it can be easily verified that $\deref(w)$ does not depend on the actual choices of the definitions that are replaced in the iterations of Step~\ref{refStepTwo}. 
\end{proof}

For a ref-word $w$, the procedure of Definition~\ref{derefDefinition} that computes $\deref(w)$ also uniquely allocates a subword $v_{\varsx}$ of $\deref(w)$ to each variable $\varsx$ that has a definition in $w$ (i.\,e., the subwords $v_{\varsx}$ defined in the iterations of Step~\ref{refStepTwo}). In this way, a ref-word $w$ over terminal alphabet $\Sigma$ and variables $\varset$ describes a \emph{variable mapping} $\varmapX{w}{\varset} : \varset \to \Sigma^*$, i.\,e., we set $\varmapX{w}{\varset}(x) = v_{\varsx}$ if $\varsx$ has a definition in $w$ and we set $\varmapX{w}{\varset}(x) = \eword$ otherwise. The elements $\varmapX{w}{\varset}(x)$ for $x \in \varset$ are called \emph{variable images}. 
\par
Note that even if $\varsx$ has a definition in $w$, $\varmapX{w}{\varset}(x) = \eword$ is possible due to a definition $\open{\varsx} \: \close{\varsx}$ or to a definition $\open{\varsx} v_{\varsx} \close{\varsx}$, where $v_{\varsx}$ is a non empty word with only references of variables $\varsy$ with $\varmapX{w}{\varset}(y) = \eword$. Also note that any ref-word $w$ over terminal alphabet $\Sigma$ and variables $\varset$ is also a ref-word over any terminal alphabet $\Sigma' \supset \Sigma$ and variables $\varset' \supset \varset$ ($\varmapX{w}{\varset'}$ then equals the extension of $\varmapX{w}{\varset}$ by $\varmapX{w}{\varset'}(x') = \emptyword$ for every $\varsx' \in \varset' \setminus \varset$). If the set of variables $\varset$ is clear from the context or negligible, we shall also denote the variable mapping by $\varmap{w}$ and if there is some obvious implicit order on $\varset$, e.\,g., given by indices as in the case $\varset = \{\varsx_1, \varsx_2, \ldots, \varsx_m\}$, then we also write $\varmap{w}$ as a tuple $(\varmap{w}(\varsx_1), \varmap{w}(\varsx_2), \ldots, \varmap{w}(\varsx_m))$.\par
A set $L$ of ref-words is called \emph{ref-language} and we extend the $\deref$-function to ref-languages in the obvious way, i.\,e., $\deref(L) = \{\deref(w) \mid w \in L\}$. Let us illustrate ref-words with an example.

\begin{example}
Let $\Sigma = \{\ta, \tb, \tc\}$ and let $\varsx_1, \varsx_2, \varsx_3, \varsx_4 \in \varset$.
\begin{equation*}
w = \ta \varsx_4 \ta \open{\varsx_1} \ta \tb \open{\varsx_2} \ta \tc \tc \close{\varsx_2} \ta \varsx_2 \varsx_4 \close{\varsx_1} \open{\varsx_3} \varsx_1 \ta \varsx_2 \close{\varsx_3} \varsx_3 \tb \varsx_1\,.
\end{equation*}
The procedure of Definition~\ref{derefDefinition} will first delete all occurrences of $\varsx_4$. Then $\open{\varsx_2} \ta \tc \tc \close{\varsx_2}$ and all references of $\varsx_2$ are replaced by $\ta \tc \tc$. After this step, the definition for variable $\varsx_1$ is $\open{\varsx_1} \ta \tb \ta \tc \tc\ta \ta \tc \tc \close{\varsx_1}$, so $\ta \tb \ta \tc \tc\ta \ta \tc \tc$ can be substituted for the definitions and references of $\varsx_1$. After replacing the last variable $\varsx_3$, we obtain $\deref(w)$. Note that $\varmap{w} = (\ta \tb \ta \tc \tc \ta \ta \tc \tc, \ta \tc \tc, \ta \tb \ta \tc \tc \ta \ta \tc \tc \ta \ta \tc \tc, \eword)$.
\end{example}

\subsection{Graph-Databases}

 A \emph{graph-database} (\emph{over $\Sigma$}) is a directed, edge labelled multigraph $\DBD = (V_{\DBD}, E_{\DBD})$, where $V_{\DBD}$ is the set of \emph{vertices} (or \emph{nodes}) and $E_{\DBD} \subseteq V_{\DBD} \times \Sigma \times V_{\DBD}$ is the set of \emph{edges} (or \emph{arcs}). A path from $u \in V_{\DBD}$ to $v \in V_{\DBD}$ of length $k \geq 0$ is a sequence $p = (w_0, a_1, w_1, a_2, w_2 \ldots, w_{k-1}, a_{k}, w_k)$ with $(w_{i-1}, a_{i}, w_{i}) \in E_{\DBD}$ for every $i \in [k]$. We say that $p$ is \emph{labelled} with the \emph{word} $a_1 a_2 \ldots a_k \in \Sigma^*$. According to this definition, for every $v \in V_{\DBD}$, $(v)$ is a path from $v$ to $v$ of length $0$ that is labelled by $\eword$. Hence, every node of every graph-database has an $\eword$-labelled path to itself (and these are the only $\eword$-labelled paths in $\DBD$).\par
Nondeterministic finite automata ($\NFA$s) are just graph databases, the nodes of which are called \emph{states}, and that have a specified \emph{start state} and a set of specified \emph{final states}. Moreover, we allow the empty word as edge label as well (which is not the case for graph databases). The language $\lang(M)$ of an $\NFA$ $M$ is the set of all labels from paths from the start state to some final state.\par
In the following, $\nodevarset$ is an enumerable set of \emph{node-variables}; we shall use symbols $\varnx, \varny, \varnz, \varnx_1, \varnx_2, \ldots, \varny_1, \varny_2, \ldots$ for node variables (in contrast to the string variables $\varset$ in sans-serif font).

\subsection{Conjunctive Path Queries} 

Let $\Re$ be a class of language descriptors, and, for every $r \in \Re$, let $\lang(r)$ denote the language represented by $r$. An \emph{$\Re$-graph pattern} is a directed, edge-labelled graph $G = (V, E)$ with $V \subseteq \nodevarset$ and $E \subseteq V \times \Re \times V$; it 
is an $\Re$-graph pattern \emph{over alphabet $\Sigma$}, if $\lang(\alpha) \subseteq \Sigma^*$ for every $(x, \alpha, y) \in E$. For an $\Re$-graph pattern $G = (V, E)$ over $\Sigma$ and a graph-database $\DBD = (V_{\DBD}, E_{\DBD})$ over $\Sigma$, a mapping $h : V \to V_{\DBD}$ is a \emph{matching morphism} for $G$ and $\DBD$ if, for every $e = (x, \alpha, y) \in E$, $\DBD$ contains a path from $h(x)$ to $h(y)$ that is labelled with a word $w_e \in \lang(\alpha)$. The tuple $(w_e)_{e \in E}$ is a tuple of \emph{matching words} (\emph{with respect to $h$}). In particular, a matching morphism can have several different tuples of matching words.\par
A \emph{conjunctive $\Re$-path query} ($\CLPQ$ for short) is a query $q = \bar{z} \gets G_q$, where $G_q = (V_q, E_q)$ is an $\Re$-graph pattern and $\bar{z} = (z_1, \ldots, z_\ell)$  with $\{z_1, z_2, \ldots, z_\ell\} \subseteq V_q$. We say that $q$ is an $\CLPQ$ \emph{over alphabet $\Sigma$} if $G_q$ is an $\Re$-graph pattern over $\Sigma$. The query $q$ is a \emph{single-edge} query, if $|E_q| = 1$. \par
For an $\CLPQ$ $q = \bar{z} \gets G_q$ over $\Sigma$ with $\bar{z} = (z_1, z_2, \ldots, z_\ell)$, a graph-database $\DBD = (V_{\DBD}, E_{\DBD})$ over $\Sigma$ and a matching morphism $h$ for $G_q$ and $\DBD$ 
(we also call $h$ a \emph{matching morphism for $q$ and $\DBD$}), we define $q_h(\DBD) = (h(z_1), h(z_2), \ldots, h(z_\ell))$ and we set $q(\DBD) = \{q_h(\DBD) \mid h \text{ is a matching morphism for $q$ and $\DBD$}\}$. 
The mapping $\DBD \mapsto q(\DBD)$ from the set of graph-databases to the set of relations over $\Sigma$ of arity $\ell$ that is defined by $q$ shall be denoted by $\llbracket q \rrbracket$, and for any class $A$ of conjunctive path queries, we set $\llbracket A \rrbracket = \{\llbracket q \rrbracket \mid q \in A\}$. \par
A \emph{Boolean} $\CLPQ$ has the form $\bar{z} \gets G_q$ where $\bar{z}$ is the empty tuple. In this case, we also denote $q$ just by $G_q$ instead of $() \gets G_q$. For a Boolean $\CLPQ$ $q$ and a graph-database $\DBD$, we either have $q(\DBD) = \{()\}$ or $q(\DBD) = \emptyset$, which we shall also denote by $\DBD \models q$ and $\DBD \not \models q$, respectively. For Boolean queries $q$, the set $\llbracket q \rrbracket$ can also be interpreted as $\{\DBD \mid \DBD \models q\}$. Two $\CLPQ$s $q$ and $q'$ are \emph{equivalent}, denoted by $q \equiv q'$, if $\llbracket q \rrbracket = \llbracket q' \rrbracket$, i.\,e., $q(\DBD) = q'(\DBD)$ for every graph-database $\DBD$ (or, in the Boolean case, $\DBD \models q \Leftrightarrow \DBD \models q'$).\par
For a class $Q$ of conjunctive path queries, $Q$-$\booleProb$ is the problem to decide, for a given Boolean $q \in Q$ and a graph database $\DBD$, whether $\DBD \models q$. By $Q$-$\checkProb$, we denote the problem to check $\bar{t} \in q(\DBD)$ for a given $q \in Q$, a graph database $\DBD$ and a tuple $\bar{t}$. \par
As common in database theory, the \emph{combined-complexity} for an algorithm solving $Q$-$\booleProb$ or $Q$-$\checkProb$ is the time or space needed by the algorithm measured in both $|q|$ and $|\DBD|$, while for the \emph{data-complexity} the query $q$ is considered constant. For simplicity, we assume $|q| = \bigO(|\DBD|)$ throughout the paper.\par
\emph{Conjunctive regular path queries} ($\CRPQ$) are $\CLPQ$ where $\Re$ is the class of regular expressions (which are defined in Section~\ref{sec:xregex}). For some of our results, we need the following result about $\CRPQ$s.

\begin{lemma}[\cite{BarceloEtAl2012}]\label{evalCRPQLemma}
$\CRPQ$-$\booleProb$ is $\npclass$-complete in combined complexity and $\nlclass$-complete in data-complexity.
\end{lemma}

\section{Xregex}\label{sec:xregex}

We next define the underlying language descriptors for our class of conjunctive path queries. Complete formalisations of the class of xregex can also be found elsewhere in the literature (e.\,g.,~\cite{Schmid2016, FreydenbergerSchmidJCSS, Freydenberger2013}). However, since we will extend xregex to conjunctive xregex (the main building block for $\CXRPQ$), we give a full definition.



\begin{definition}[Xregex]\label{xregexDefinition}
The set $\xregex_{\Sigma, \varset}$ of \emph{regular expressions with backreferences} (\emph{over $\Sigma$ and $\varset$}), also denoted by \emph{xregex}, for short, is recursively defined as follows:
\begin{enumerate}
\item\label{regexDefPointOne} $a \in \xregex_{\Sigma, \varset}$ and $\var(a) = \emptyset$, for every $a \in \Sigma \cup \{\eword\}$,
\item $\varsx \in \xregex_{\Sigma, \varset}$ and $\var(\varsx) = \{\varsx\}$, for every $\varsx \in \varset$, 
\item\label{regexDefPointTwo} $(\alpha \cdot \beta) \in \xregex_{\Sigma, \varset}$, $(\alpha \altop \beta) \in \xregex_{\Sigma, \varset}$, and $(\alpha)^+ \in \xregex_{\Sigma, \varset}$, for every $\alpha,\beta\in \xregex_{\Sigma, \varset}$;\\ furthermore, $\var((\alpha \cdot \beta)) = \var((\alpha \altop \beta)) = \var(\alpha) \cup \var(\beta)$ and $\var((\alpha)^+) = \var(\alpha)$, 
\item $\varsx\{\alpha\} \in \xregex_{\Sigma, \varset}$ and $\var(\varsx\{\alpha\}) = \var(\alpha) \cup \{\varsx\}$, for every $\alpha \in \xregex_{\Sigma, \varset}$ and $\varsx \in \varset \setminus \var(\alpha)$.
\end{enumerate}
\end{definition}

For technical reasons, we also add $\varnothing$ to $\xregex_{\Sigma, \varset}$. For $\alpha \in \xregex_{\Sigma, \varset}$, we use $r^*$ as a shorthand form for $r^+ \altop \eword$, and we usually omit the operator `$\cdot$', i.\,e., we use juxtaposition. If this does not cause ambiguities, we often omit parenthesis. For example, $\varsx$, $\varsx\{\varsy \ta\}$ and $\varsx\{(\varsy\{\varsz\{ \ta^* \altop \tb \tc\} \ta\} \varsy)^+ \tb\} \varsx$ are all valid xregex, while neither $\varsx\{\ta \varsx\} \tb$ nor $\varsx\{\ta \varsx\{\tb^*\} \ta \} \tb$ is a valid xregex.\par
We call an occurrence of $\varsx \in \varset$ a \emph{reference of variable $\varsx$} and a subexpression $\varsx\{\alpha\}$ a \emph{definition of variable $\varsx$}. The set $\xregex_{\Sigma, \emptyset}$ 
is exactly the set of regular expressions over $\Sigma$, which shall be denoted by $\RE_{\Sigma}$ in the following. We also use the term \emph{classical} regular expressions for a clearer distinction from xregex.
If the underlying alphabet $\Sigma$ or set $\varset$ of variables is negligible or clear form the context, we also drop these and simply write $\xregex$ and $\RE$. \par
We next define the semantics of xregex, for which we heavily rely on the concept of ref-words (see Section~\ref{sec:prelim}). First, for any $\alpha \in \RE$, let $\lang(\alpha)$ be the regular language described by the (classical) regular expression $\alpha$ defined as usual: $\lang(a) = \{a\}$, $\lang(\alpha \cdot \beta) = \lang(\alpha) \cdot \lang(\beta)$, $\lang(\alpha \altop \beta) = \lang(\alpha) \cup \lang(\beta)$, $\lang(\alpha^+) = \lang(\alpha)^+$. 
Now in order to define the language described by an xregex $\alpha \in \xregex_{\Sigma, \varset}$, we first define a classical regular expression $\alpha_{\mathsf{ref}}$ over the alphabet $\Sigma \cup \varset \cup \{\open{\varsx}, \close{\varsx} \mid \varsx \in \varset\}$ that is obtained from $\alpha$ by iteratively replacing all variable definitions $\varsx\{ \beta \}$ by $\open{\varsx} \beta \close{\varsx}$. For example,
\begin{align*}
\alpha =\:&\varsx\{(\varsy\{\varsz\{ \ta^* \altop \tb \tc\} \ta\} \varsy)^+ \tb\} \varsx\,,\\
\alpha_{\mathsf{ref}} =\:&\open{\varsx}(\open{\varsy} \open{\varsz} (\ta^* \altop \tb \tc) \close{\varsz} \ta\close{\varsy} \varsy)^+ \tb \close{\varsx} \varsx\,.
\end{align*}
 We say that $\alpha$ is \emph{sequential} if every $w \in \lang(\alpha_{\mathsf{ref}})$ contains for every $\varsx \in \varset$ at most one occurrence of $\open{\varsx}$ (if not explicitly stated otherwise, all our xregex are sequential). If an xregex $\alpha$ is sequential, then every $w \in \lang(\alpha_{\mathsf{ref}})$ must be a ref-word. Indeed, this directly follows from the fact that the definitions $\varsx\{\ldots\}$ are always subexpressions. Hence, for sequential xregex, $\lang(\alpha_{\mathsf{ref}})$ is a ref-language, which we shall denote by $\reflang(\alpha)$. 
If a ref-word $v \in \reflang(\alpha)$ contains a definition $\open{\varsx} v_{\varsx} \close{\varsx}$, then we say that the corresponding definition $\varsx\{ \gamma_x \}$ in $\alpha$ is \emph{instantiated} (by $v$). In particular, we observe that sequential xregex can nevertheless have several definitions for the same variable $\varsx$, but at most one of them is instantiated by any ref-word.
Finally, the language described by $\alpha$ is defined as $\lang(\alpha) = \deref(\reflang(\alpha))$. As a special case, we also define $\lang(\varnothing) = \emptyset$. For $\alpha \in \xregex_{\Sigma, \varset}$ and $w \in \Sigma^*$, we say that \emph{$w$ matches $\alpha$ with witness $u \in \reflang(\alpha)$ and variable mapping $\varmap{u}$}, if $\deref(u) = w$. \par

\begin{example}
Let $\alpha \in \xregex_{\{\ta, \tb\}, \varset}$ with $\varsx_1, \varsx_2 \in \varset$:
\begin{align*}
\alpha &= \ta^* \varsx_1\{\ta^*\varsx_2\{(\ta \altop \tb)^*\}\tb^*\ta^*\} \varsx_2^* (\ta \altop \tb)^* \varsx_1\,,\\
\alpha_{\mathsf{ref}} &= \ta^* \open{\varsx_1} \ta^* \open{\varsx_2}(\ta \altop \tb)^* \close{\varsx_2} \tb^* \ta^* \close{\varsx_1} \varsx_2^* (\ta \altop \tb)^* \varsx_1\,,\\
u_1 &= \ta \ta \ta \ta \open{\varsx_1} \open{\varsx_2} \tb \ta \close{\varsx_2} \tb \ta \ta \close{\varsx_1} \varsx_2 \varsx_2 \tb \tb \ta \varsx_1\in \reflang(\alpha)\,,\\
u_2 &= \ta \ta \ta \open{\varsx_1} \ta \open{\varsx_2} \tb \ta \tb \close{\varsx_2} \ta \ta \close{\varsx_1} \varsx_2 \ta \tb \tb \varsx_1\in \reflang(\alpha)\,,\\
w_{\alpha} &= \deref(u_1) = \deref(u_2) = \ta^4 (\tb \ta)^2 (\ta \tb)^3 (\tb \ta)^3 \ta\in \lang(\alpha)\,,\\
\varmap{u_1} &= (\tb \ta \tb \ta \ta, \tb \ta), \varmap{u_2} = (\ta \tb \ta \tb \ta \ta, \tb \ta \tb)\,.
\end{align*}
For $\gamma = \varsx_1\{\tc^* (\varsx_2\{\ta^*\} \altop \varsx_3\{\tb^*\})\} \tc \varsx_2 \tc \varsx_3 \tb \varsx_1$, $\tc^2 \ta^2 \tc \ta^2 \tc \tb \tc^2 \ta^2 \in \lang(\gamma)$ is witnessed by ref-word $u_5 = \open{\varsx_1} \tc^2 \open{\varsx_2} \ta^2 \close{\varsx_2} \close{\varsx_1} \tc \varsx_2 \tc \varsx_3 \tb \varsx_1$, which has the variable mapping $\varmap{u_5} = (\tc^2 \ta^2, \ta^2, \eword)$. Note that empty variable images can result from variables without definitions in the ref-word (as in this example), but also from definitions of variables that define the empty word, as, e.\,g., in the ref-word $\open{\varsx_1} \close{\varsx_1} \tc \varsx_1 \in \reflang(\varsx_1\{(\ta \altop \tb)^*\} \tc \varsx_1)$. 
\end{example}

For $\alpha \in \xregex_{\Sigma, \varset}$, we define the relation $\preceq_{\alpha}$ analogously how it is done for ref-words, i.\,e., $\varsx \preceq_{\alpha} \varsy$ if in $\alpha$ there is a definition of $\varsy$ that contains a reference or a definition of $\varsx$. Even though the transitive closure of $\preceq_v$ is acyclic for every $v \in \reflang(\alpha)$, the transitive closure of $\preceq_{\alpha}$ is not necessarily acyclic. For example, $\alpha = \varsx\{\ta^*\} \varsy\{\varsx\} \altop \varsy\{\ta^*\} \varsx\{\varsy\}$ is an xregex, but the transitive closure of $\preceq_{\alpha}$ is not acyclic. We call xregex \emph{acyclic} if $\preceq_{\alpha}$ is acyclic.\par

\subsection{Conjunctive Xregex}\label{sec:ConXregex}

Syntactically, conjunctive xregex are tuples of xregex. Their semantics, however, is more difficult and we spend some more time with intuitive explanations before giving a formal definition.

\begin{definition}[Conjunctive Xregex]\label{conjunctiveXregexDefinition}
A tuple $\bar{\alpha} = (\alpha_1, \ldots, \alpha_m) \in (\xregex_{\Sigma, \varset})^m$ is a \emph{conjunctive xregex of dimension $m$}, if $\alpha_1 \alpha_2 \ldots \alpha_m$ is an acyclic xregex.
\end{definition}

By $\dimconxregex{m}_{\Sigma, \varset}$, we denote the set of conjunctive xregex of dimension $m$ (over $\Sigma$ and $\varset$) and we set $\conxregex_{\Sigma, \varset} = \bigcup_{m \geq 1} \dimconxregex{m}_{\Sigma, \varset}$. Note that $\dimconxregex{1}_{\Sigma, \varset} = \xregex_{\Sigma, \varset}$. If the terminal alphabet $\Sigma$ or set $\varset$ of variables are negligible or clear from the context, we also drop the corresponding subscripts. We also write $\bar{\alpha}[i]$ to denote the $i^{\text{th}}$ element of some $\bar{\alpha} \in \dimconxregex{m}$. \par
Before defining next the semantics of conjunctive xregex, we first give some intuitive explanations. The central idea of a conjunctive xregex $\bar{\alpha} = (\alpha_1$, $\alpha_2$, $\ldots$, $\alpha_m)$ is that the definition of some $\varsx$ in some $\alpha_i$ also serves as definition for possible references of $\varsx$ in some $\alpha_j$ with $i \neq j$ (which, due to the requirement that $\alpha_1\alpha_2 \ldots \alpha_m$ is an xregex, cannot contain a definition of $\varsx$). Let us illustrate the situation with the concrete example $\bar{\gamma} = (\gamma_1, \gamma_2)$, with $\gamma_1 = (\varsx\{\ta^*\} \altop \tb^*) \varsy$ and $\gamma_2 = \varsy\{\varsx \ta \varsx \tb\} \tb \varsy^*$.\par
As an intermediate step, we first add dummy definitions for the undefined variables, i.\,e., we define $\interxregex{\gamma_1} = \varsy\{\Sigma^*\} \# \gamma_1$ and $\interxregex{\gamma_2} = \varsx\{\Sigma^*\} \# \gamma_2$, which are both xregex in which all variables have a definition. Now, we can treat $(\interxregex{\gamma_1}, \interxregex{\gamma_2})$ as a generator for pairs of ref-words (and therefore as pairs of words over $\Sigma$), but we only consider those pairs of ref-words that have the same variable mapping. For example, $u_1 = \open{\varsy} \ta^5 \tb \close{\varsy} \# \open{\varsx} \ta \ta\close{\varsx} \varsy \in \reflang(\interxregex{\gamma_1})$ and $u_2 = \open{\varsx} \ta \ta \close{\varsx} \# \open{\varsy} \varsx \ta \varsx \tb \close{\varsy} \tb \varsy \varsy \in \reflang(\interxregex{\gamma_2})$ and, furthermore, since $u_1$ and $u_2$ have the same variable mapping $(\ta \ta, \ta^5 \tb)$, they are witnesses for the conjunctive match $(w_1, w_2) = (\ta \ta \ta^5 \tb, \ta^5 \tb \tb (\ta^5 \tb)^2)$ of $\bar{\gamma}$, since $\deref(u_1) = \ta^5 \tb \# w_1$ and $\deref(u_2) = \ta \ta \# w_2$. On the other hand, $v_1 = \open{\varsy} \ta \close{\varsy} \# \open{\varsx} \ta \close{\varsx} \varsy$ and $v_2 = \open{\varsx} \ta \close{\varsx} \# \open{\varsy} \varsx \ta \varsx \tb \close{\varsy} \tb \varsy$ are also ref-words from $\reflang(\interxregex{\alpha_1})$ and $\reflang(\interxregex{\alpha_2})$, respectively, but, even though $\deref(v_1) = \ta  \# \ta \ta$ and $\deref(v_2) =  \ta \# \ta^3 \tb \tb \ta^3 \tb$, the tuple $(\ta \ta, \ta^3 \tb \tb \ta^3 \tb)$ is not a conjunctive match for $\bar{\alpha}$, since $\varmap{v_1}(\varsy) = \ta \neq \ta^3 \tb = \varmap{v_2}(\varsy)$. We shall now generalise this idea to obtain a sound definition of the semantics of conjunctive xregex.\par
For any xregex $\gamma \in \xregex_{\Sigma, \varset}$, let $\varprefix{\gamma} = \Pi_{x \in A} x\{\Sigma^*\}$, where $A \subseteq \varset$ is the set of variables that have no definition in $\gamma$ (the order of the definitions $\varsx\{\Sigma^*\}$ with $\varsx \in A$ in $\varprefix{\gamma}$ shall be irrelevant), and we also define $\interxregex{\gamma} = \varprefix{\gamma} \# \gamma$, where $\#$ is a new symbol with $\# \notin \Sigma$. Finally, a tuple $\bar{w} = (w_1, w_2, \ldots, w_m) \in (\Sigma^*)^m$ is a (\emph{conjunctive}) \emph{match} for $\bar{\alpha} = (\alpha_1, \alpha_2, \ldots, \alpha_m) \in (\xregex_{\Sigma, \varset})^m$ with variable mapping $\psi$ if, for every $i \in [m]$, there is a ref-word $v_i = v'_i \# v''_i \in \reflang(\interxregex{\alpha_i})$ with $\deref(v_i) = \deref(v'_i) \# w_i$ and $\varmap{v_i} = \psi$.\par 
Conjunctive xregex behave similar to xregex in terms of how $\eword$ can be allocated to a variable.
It might happen that $\alpha_i$ contains the definition of some variable $\varsx$ (and therefore $\varprefix{\alpha_i}$ does not contain $\varsx\{\Sigma^*\}$), but the corresponding ref-word $v_i \in \reflang(\interxregex{\alpha_i})$ does not contain a definition of $\varsx$. This means that $\varmap{v_i}(\varsx) = \emptyword$ and therefore all other $v_{j}$ with $i \neq j$ must also allocate $\eword$ to $\varsx$ in their definitions $\varsx\{\Sigma^*\}$ contained in $\varprefix{\alpha_j}$. On the other hand, it is possible that $\varmap{v_i}(\varsx) = \emptyword$ even though $v_i$ contains a definition of $\varsx$. \par
%
%
For an $\bar{\alpha} \in \conxregex$, $\lang(\bar{\alpha})$ is the \emph{set of conjunctive matches} for $\bar{\alpha}$. We say that conjunctive xregex $\bar{\alpha}$ and $\bar{\beta}$ are \emph{equivalent} if $\lang(\bar{\alpha}) = \lang(\bar{\beta})$ (note that $\bar{\alpha}$ and $\bar{\beta}$ having the same dimension is necessary for this). We stress the fact that the order of the elements $\alpha_i$ of a conjunctive xregex $\bar{\alpha} = (\alpha_1, \alpha_2, \ldots, \alpha_m)$ is vital. \par
A special class of conjunctive xregex is $\dimconxregex{m}_{\Sigma, \emptyset}$, i.\,e., the class of all $m$-dimensional tuples of classical regular expressions. For every $\bar{\alpha} \in \dimconxregex{m}_{\Sigma, \emptyset}$, $\lang(\bar{\alpha}) = \lang(\bar{\alpha}[1]) \times \lang(\bar{\alpha}[2]) \times \ldots \times \lang(\bar{\alpha}[m])$.

\begin{example}
Consider the xregex 
\begin{align*}
&\alpha_1 = \:\varsx_2\{\varsx_1 \altop \ta^*\} \tb\,,& &\alpha_2 = \:\varsx_1\{(\ta \altop \tb)^*\} \varsx_3\{\tc^*\} \tb \varsx_3\,,&\\
&\alpha_3 = \:\varsx_2^* \ta^* \varsx_1\,,& &\alpha_4 = \:\varsx_4\{\ta^*\} \tb \varsx_4 \varsx_1\{\varsx_2 \ta\}\,.&
\end{align*}
Then $(\alpha_2, \alpha_4)$ is not a conjunctive xregex since $\alpha_2 \alpha_4$ is not sequential, but both $(\alpha_3, \alpha_4)$ and $(\alpha_1, \alpha_2, \alpha_3)$ are conjunctive xregex.\par
Obviously, $w_1 = \ta \ta \tb$, $w_2 = \tb \tb \ta \tc \tb \tc$ and $w_3 = \ta \ta$ are matches for $\alpha_1$, $\alpha_2$ and $\alpha_3$, respectively, with variable mappings $(\eword, \ta \ta, \eword)$, $(\tb \tb \ta, \eword, \tc)$ and $(\eword, \eword, \eword)$, respectively. However, 
for $(w_1, w_2, w_3)$ to be a conjunctive match for $(\alpha_1, \alpha_2, \alpha_3)$, there must be words $u_1, u_2, u_3$ such that the following $w'_i$ match the $\interxregex{\alpha_i}$ all with the same variable mapping $(u_1, u_2, u_3)$, which can be easily seen to be impossible.
\begin{align*}
&w'_1 = u_1 u_3 \# w_1& &\interxregex{\alpha_1} = \:\varsx_1\{\Sigma^*\} \varsx_3\{\Sigma^*\} \# \alpha_1& \\
&w'_2 = u_2 \# w_2& &\interxregex{\alpha_2} = \:\varsx_2\{\Sigma^*\} \# \alpha_2& \\
&w'_3 = u_1 u_2 u_3 \# w_3& &\interxregex{\alpha_3} = \:\varsx_1\{\Sigma^*\} \varsx_2\{\Sigma^*\} \varsx_3\{\Sigma^*\} \# \alpha_3&
\end{align*}
However, it can be verified that $(\ta \tb \tb, \ta \tb \tc \tc \tb \tc \tc, \ta \tb \ta \tb \ta \ta \ta \tb)$ is a conjunctive match for $(\alpha_1, \alpha_2, \alpha_3)$ with variable mapping $(\ta \tb, \ta \tb, \tc \tc)$.
\end{example}

The following lemma is straightforward, but nevertheless helpful in the following proofs. It also points out how the semantics of a conjunctive xregex $\bar{\alpha}$ depends on the ref-languages of the individual components $\bar{\alpha}[i]$.

\begin{lemma}\label{refLangStableChangesLemma}
Let $\bar{\alpha} = (\alpha_1, \alpha_2, \ldots, \alpha_m) \in \dimconxregex{m}_{\Sigma, \varset}$ and let $\beta_1, \beta_2, \ldots, \beta_m \in \xregex_{\Sigma, \varset}$ such that, for every $i \in [m]$, $\reflang(\alpha_i) = \reflang(\beta_i)$. Then $\bar{\beta} = (\beta_1$, $\beta_2$, $\ldots$, $\beta_m)$ is a conjunctive xregex with $\lang(\bar{\alpha}) = \lang(\bar{\beta})$.
\end{lemma}

\begin{proof}
We first note that $\reflang(\alpha_i) = \reflang(\beta_i)$ for every $i \in [m]$ also implies that $\reflang(\alpha_1 \alpha_2 \ldots \alpha_m) = \reflang(\beta_1 \beta_2 \ldots \beta_m)$. If $\bar{\beta}$ is not a conjunctive xregex, then $\beta_1 \beta_2 \ldots \beta_m$ is not sequential and, since $\reflang(\alpha_1 \alpha_2 \ldots \alpha_m) = \reflang(\beta_1 \beta_2 \ldots \beta_m)$, this directly implies that $\alpha_1 \alpha_2 \ldots \alpha_m$ is not sequential, which is a contradiction.\par
It remains to show that $\lang(\bar{\alpha}) = \lang(\bar{\beta})$. Since, for every $i \in [m]$, $\reflang(\alpha_i) = \reflang(\beta_i)$, we can conclude that the $\varprefix{\beta_i} = \varprefix{\alpha_i}$. Moreover, this also implies that, for every $i \in [m]$, $\reflang(\varprefix{\beta_i} \# \beta_i) = \reflang(\varprefix{\alpha_i} \# \alpha_i)$ and therefore $\reflang(\interxregex{\alpha_i}) = \reflang(\interxregex{\beta_i})$. This directly implies that any $(w_1, w_2, \ldots, w_m)$ is a conjunctive match for $\bar{\alpha}$ if and only if it is a conjunctive match for $\bar{\beta}$.
\end{proof}

\section{Conjunctive Xregex Path Queries}\label{sec:CXRPQ}

Finally, we can define 
conjunctive xregex path queries:

\begin{definition}[$\CXRPQ$]
A \emph{conjunctive xregex path query} ($\CXRPQ$ for short) is an $\CLPQ$ $q = \bar{z} \gets G_q$, where $\Re$ is the class of xregex, $G_q = (V_q, E_q)$ with $E_q = \{(x_i, \alpha_i, y_i) \mid i \in [m]\}$  and $\bar{\alpha} = (\alpha_1, \alpha_2, \ldots, \alpha_m)$ is a conjunctive xregex (which is also called \emph{the conjunctive xregex of $q$}).
\end{definition}

The semantics are defined by combining the semantics of conjunctive path queries and conjunctive xregex in the obvious way. More precisely, let $q = \bar{z} \gets G_q$ be a $\CXRPQ$, where $G_q = (V_q, E_q)$ with $E_q = \{(x_i, \alpha_i, y_i) \mid i \in [m]\}$, and let $\DBD$ be a graph database. Then $h : V_q \to V_{\DBD}$ is a matching morphism for $q$ and $\DBD$ if there is a conjuctive match $(w_1, w_2, \ldots, w_m)$ of $\bar{\alpha}$, such that, for every $i \in [m]$, $\DBD$ contains a path from $h(x_i)$ to $h(y_i)$ labelled with $w_i$. Note that this definition of a matching morphism $h$ for $\CXRPQ$ also implies definitions for $q_h(\DBD)$, $q(\DBD)$, $\DBD \models q$, $\llbracket q \rrbracket$ and $\llbracket \CXRPQ \rrbracket$. \par
Since, for a $q \in \CXRPQ$ and a fixed graph database $\DBD$, the set $q(\DBD)$ is completely determined by the corresponding graph pattern and the set of conjunctive matches of the corresponding conjunctive xregex, the following can directly be concluded from the definitions.

\begin{proposition}\label{xregexDetermineGraphQueryProposition}
Let $q = \bar{z} \gets G_q$ be a $\CXRPQ$ with conjunctive xregex $\bar{\alpha} \in \dimconxregex{m}$. Let $\bar{\beta} \in \dimconxregex{m}$ and let $q' = \bar{z} \gets G_{q'}$ be a $\CXRPQ$ where $G_{q'}$ is obtained from $G_q$ by replacing each edge label $\bar{\alpha}[i]$ by $\bar{\beta}[i]$, for every $i \in [m]$. If $\lang(\bar{\alpha}) = \lang(\bar{\beta})$, then $q \equiv q'$.
\end{proposition}

Next, we observe that evaluating $\CXRPQ$ is surprisingly difficult. Let $\Delta = \{\ta, \tb, \#\}$ and let $\NFAintQuery = \# \varsz\{(\ta \altop \tb)^*\} \, (\#\# \, \varsz)^* \#\#\# \in \xregex_{\Delta, \varset}$. 

\begin{theorem}\label{CXRPQDataComplexityHardnessTheorem}
Deciding whether a given graph-database over $\Sigma \supset \Delta$ contains a path labelled with some $w \in \lang(\NFAintQuery)$ is $\pspaceclass$-hard.
\end{theorem}

\begin{proof}
We will prove the result by a reduction from the $\pspaceclass$-complete $\NFA$-intersection problem over binary alphabet $\{\ta, \tb\}$, which is defined as follows: Given $\NFA$ $M_1, \ldots, M_k$ over alphabet $\{\ta, \tb\}$, decide whether or not $\bigcap^k_{i = 1} \lang(M_i) \neq \emptyset$. \par
Without loss of generality, we assume that, for every $i \in [k]$, $M_i$ has state set $Q_i$, transition function $\delta_i$, initial state $q_{0, i}$, only one accepting state $q_{f, i}$, and we also assume that $\bigcap^{k}_{i = 1} Q_i = \emptyset$. We transform the $\NFA$ $M_1, \ldots, M_k$ into a graph-database $\DBD = (V_{\DBD}, E_{\DBD})$ over alphabet $\Delta = \{\ta, \tb, \#\}$ as follows. For the sake of convenience, we assume that arcs can also be labelled by the words $\#\#$ and $\#\#\#$ (technically, these would be paths of length $2$ and $3$, respectively, instead of single arcs). We set $V_{\DBD} = (\bigcup^k_{i = 1} Q_i) \cup \{s, t\}$, where $\{s, t\} \cap (\bigcup^k_{i = 1} Q_i) = \emptyset$, and we set $E_{\DBD} = (\bigcup^k_{i = 1} \delta_i) \cup \{(q_{f, i}, \#\#, q_{0, i+1}) \mid 1 \leq i \leq k-1\} \cup \{(s, \#, q_{0, 1}), (q_{f, k}, \#\#\#, t)\}$. \par
We first note that in $\DBD$ there is a path labelled with a word from $\lang(\NFAintQuery)$ if and only if there is such a path from node $s$ to node $t$. Let us now assume that there is a path in $\DBD$ from $s$ to $t$ labelled with a word $w \in \lang(\NFAintQuery)$. Since $w \in \lang(\NFAintQuery)$, we can conclude that $w = \# w' (\#\# w' )^\ell \#\#\#$ for some $w' \in \{\ta, \tb\}^*$ and $\ell \geq 0$. By the structure of $\DBD$, this directly implies that $\ell = k-1$ and, for every $i \in [k]$, there is a path labelled with $w'$ from $q_{0, i}$ to $q_{f, i}$. Consequently, $w' \in \bigcap^k_{i = 1} \lang(M_i)$. On the other hand, if there is some word $w' \in \bigcup^k_{i = 1} \lang(M_i)$, then, for every $i \in [k]$, there is a path labelled with $w'$ from $q_{0, i}$ to $q_{f, i}$, which directly implies that there is a path from $s$ to $t$ labelled with $\# w' (\#\# w' )^{k-1} \#\#\# \in \lang(\NFAintQuery)$. 
\end{proof}

Theorem~\ref{CXRPQDataComplexityHardnessTheorem} constitutes a rather strong negative result: $\CXRPQ$-$\booleProb$ is $\pspaceclass$-hard in data-complexity, even for a single-edge query with a conjunctive xregex $\alpha \in \xregex_{\Sigma, \varset}$, where $|\Sigma| = 3$ and $|\varset| = 1$. However, as explained in the introduction, $\NFAintQuery$ gives a hint how to restrict $\CXRPQ$ to obtain more tractable fragments. Such fragments are investigated in the following Sections~\ref{sec:vsfCXRPQ}~and~\ref{sec:boundedImageSize}.


\section{Variable-Star Free $\CXRPQ$}\label{sec:vsfCXRPQ}

An xregex $\alpha$ is \emph{variable-star free} (\emph{vstar-free}, for short), if no variable reference is subject to the $+$-operator.
\footnote{Since we use the Kleene-star $r^*$ only as short hand form for $r^+ \altop \eword$, the term ``variable-\emph{plus} free'' seems more appropriate. We nevertheless use the term ``star free'', since it is much more common in the literature on regular expressions and languages.} See Example~\ref{restrictionsExample} for an illustration. 
A conjunctive xregex $\bar{\alpha}$ of dimension $m$ is vstar-free if, for every $i \in [m]$, $\bar{\alpha}[i]$ is vstar-free, and a $q \in \CXRPQ$ is vstar-free, if its conjunctive xregex is vstar-free. Finally, let $\vsfCXRPQ$ be the class of vstar-free $\CXRPQ$. \par
The main result of this section is as follows.

\begin{theorem}\label{vsfCXRPQEvalUpperBoundTheorem}
$\vsfCXRPQ$-$\booleProb$ is in $\expspaceclass$ with respect to com\-bined-complexity and in $\nlclass$ with respect to data-complexity.
\end{theorem}




Before discussing at greater length how to prove the upper bounds of Theorem~\ref{vsfCXRPQEvalUpperBoundTheorem}, we first observe the following lower bound. 

%

\begin{theorem}\label{vsfCXRPQEvalLowerBoundTheorem}
$\vsfCXRPQ$-$\booleProb$ 
\begin{itemize}
\item is $\pspaceclass$-hard in combined-complexity, even for single-edge queries with xregex $\alpha \in \xregex_{\Sigma, \varset}$, where $|\Sigma| = 3$ and $|\varset| = 1$, and
\item is $\nlclass$-hard in data-complexity, even for single-edge queries with xregex $\alpha$, where $\alpha \in \RE_{\Sigma}$ with $|\Sigma| = 2$.
\end{itemize}
\end{theorem}


\begin{proof}
For the $\pspaceclass$-hardness of $\CXRPQ$-$\booleProb$ in com\-bined-complexity, we can proceed similarly to the proof of Theorem~\ref{CXRPQDataComplexityHardnessTheorem}. An instance $M_1, \ldots, M_k$ of the $\NFA$-intersection problem over alphabet $\{\ta, \tb\}$ is transformed into a graph-database $\DBD = (V_{\DBD}, E_{\DBD})$ over alphabet $\Sigma = \{\ta, \tb, \#\}$ in the same way as in the proof of Theorem~\ref{CXRPQDataComplexityHardnessTheorem} and into a graph pattern $(x, \NFAintQueryk, y)$ with 
\begin{equation*}
\NFAintQueryk = \# \varsz\{(\ta \altop \tb)^*\} \, (\#\# \, \varsz)^{k-1} \#\#\# 
\end{equation*}
(i.\,e., $\NFAintQueryk$ is obtained from $\NFAintQuery$ by replacing $(\#\# \, \varsz)^*$ by $k$ copies of $\#\# \, \varsz$). We note that $\NFAintQueryk$ is vstar-free, but, due to the explicit $(k-1)$-times repetition of $(\#\# \, \varsz)$, it does not have constant size. It follows analogously as in the proof of Theorem~\ref{CXRPQDataComplexityHardnessTheorem} that in $\DBD$ there is a path labelled with some $w \in \lang(\NFAintQueryk)$ if and only if $\bigcap^k_{i = 1} \lang(M_i) \neq \emptyset$. \par
The $\nlclass$-hardness of $\CXRPQ$-$\booleProb$ in data-complexity can directly be concluded from the fact that $\CRPQ$-$\booleProb$ is already $\nlclass$-hard in data-complexity. In order to show the stronger statement of the result, we give a full proof. We consider the Boolean $\CRPQ$ $q$ represented by the graph pattern $(x, \alpha, z)$ with $\alpha = \ta \tb^* \ta \ta$. From an arbitrary directed (and unlabelled) graph $G = (V, E)$ and two vertices $s, t \in V$, we can construct a graph-database $\DBD = (V_{\DBD}, E_{\DBD})$, where $V_{\DBD} = V \cup \{s', t'\}$ and $E_{\DBD} = \{(u, \tb, v) \mid (u, v) \in E\} \cup \{(s', \ta, s), (t, \ta, t'), (t', \ta, t'')\}$. It can be easily verified that in $G$ there is a path from $s$ to $t$ if and only if in $\DBD$ there is a path from $s'$ to $t''$ labelled with $\ta \tb^k \ta \ta$ if and only if $\DBD \models q$ is true. 
\end{proof}



The upper bounds require more work. Before we can give an intuitive idea of our procedure, we have to introduce some more restrictions of xregex. An xregex $\alpha$ is

\begin{itemize}
\item \emph{variable-alternation free} (\emph{valt-free}) if, for every subexpression $(\beta_1 \altop \beta_2)$ of $\alpha$, neither $\beta_1$ nor $\beta_2$ contain any variable definition or variable reference. 
\item \emph{variable-simple} if it is vstar-free and valt-free. 
\item \emph{simple} if it is variable-simple and every variable definition $\varsx\{\gamma\}$ is \emph{basic}, i.\,e., $\gamma$ is a classical regular expression or a single variable reference.
\item in \emph{normal form} if $\alpha = \alpha_1 \altop \alpha_2 \altop \ldots \altop \alpha_m$ and, for every $i \in [m]$, $\alpha_i$ is simple. 
\end{itemize}

Let us clarify these definitions with some intuitive explanations and examples. The condition of being vstar-free or valt-free can also be interpreted as follows: an xregex is vstar-free (or valt-free), if every subtree of its syntax tree rooted by a node for a $+$-operation (for a $\altop$-operation, respectively) does not contain any nodes for variable definitions or references.
If this is true for all $+$-operation nodes \emph{and} all $\altop$-operation nodes, then the xregex is variable-simple. Equivalently, $\alpha$ is variable-simple if $\alpha = \beta_1 \beta_2 \ldots \beta_k$, where each $\beta_i$ is a classical regular expression, a variable reference or a variable definition $\varsx\{\gamma\}$, where $\gamma$ is also variable-simple. If additionally each variable definition $\varsx\{\gamma\}$ is such that $\gamma$ is a classical regular expressions or a single variable reference, then $\alpha$ is simple. Finally, an xregex is in normal form, if it is an alternation of simple xregex.

\begin{example}\label{restrictionsExample}
The xregex $\varsx\{\ta^*\} (\tb \varsx (\tc \altop \ta))^* \tb$ is not vstar-free, but valt-free. The xregex $\varsx\{\ta^*\} \varsy{((\tb \varsx) \altop (\tc \ta))} \tb^* \varsy$ is vstar-free, but not valt-free. Furthermore, the xregex $\ta \varsx\{(\tb \altop \tc)^* \tb \varsy\{\td \varsx \ta^*\} \} \tb \varsx \ta^* \varsz\{\td^*\} \varsz \varsy$ is variable-simple, but not simple, 
and $\ta \varsx\{(\tb \altop \tc)^*\td \ta\} \tb \varsx \ta^* \varsy\{\varsz\} \varsx \varsy$ is simple.
\end{example}

We extend these restrictions to conjunctive xregex and $\CXRPQ$ in the obvious way, i.\,e., as also done above for vstar-free xregex. \par
%

We can now give a high-level ``road map'' for the proof of Theorem~\ref{vsfCXRPQEvalUpperBoundTheorem}. As an important building block, we first prove that $\CXRPQ$ that are simple (in the sense defined above) can be evaluated in nondeterministics space $\bigO(|q| \log(|\DBD|) + |q|\log(|q|))$ (see Lemma~\ref{simpleCXRPQLemma} below). Then, in Subsection~\ref{sec:normalForm}, we show that every variable star-free conjunctive xregex can be transformed into an equivalent one in normal form (Lemmas~\ref{conjunctiveNormalFormLemmaStepOne},~\ref{uniqueDefinitionsLemma}~and~\ref{removeNonBasicLemma}), which also means that $\vsfCXRPQ$ can be transformed into equivalent $\CXRPQ$ in normal form. A $\CXRPQ$ in normal form has a conjunctive xregex whose elements are alternations of simple xregex; thus, they can be evaluated by using Lemma~\ref{simpleCXRPQLemma}. This directly yields a nondeterministic log-space algorithm with respect to data-complexity. For combined-complexity, we have to deal with the problem that the normal form causes a double exponential size blow-up. In order to get the $\expspaceclass$ upper bound, we observe that the first exponential blow-up can be handled by nondeterminism. \par
After this, in Subsection~\ref{sec:pseudobasic}, we discuss more thoroughly what exactly causes the exponential blow-up in our normal form and how it can be avoided. Our observations can be transformed into the fragment $\vsfpbCXRPQ$ (already mentioned in the introduction) that avoids the exponential blow-up, and therefore has a $\pspaceclass$-complete evaluation problem in combined-complexity. \par

We proceed with Lemma~\ref{simpleCXRPQLemma} (note that its proof is similar as the corresponding result for $\CRPQ$ (see, e.\,g.,~\cite{BarceloEtAl2012}).

\begin{lemma}\label{simpleCXRPQLemma}
Given a Boolean $q \in \CXRPQ$ that is simple and a graph-database $\DBD$, we can nondeterministically decide whether or not $\DBD \models q$ in space $\bigO(|q| \log(|\DBD|))$.
\end{lemma}

\begin{proof}
Let $q$ be defined by a graph pattern $G_q$ with edge set $E_q = \{(x_i, \alpha_i, y_i) \mid i \in [m]\}$ and simple conjunctive xregex $\bar{\alpha} = (\alpha_1, \alpha_2, \ldots, \alpha_m) \in \conxregex_{\Sigma, \varset}$. By assumption, for every $i \in [m]$, $\alpha_i = \alpha_{i, 1} \alpha_{i, 2} \ldots \alpha_{i, t_i}$, where, for every $j \in [t_i]$, $\alpha_{i, j}$ is a classical regular expression, a variable reference, or a variable definition $\varsx\{\gamma\}$, where $\gamma$ is either a classical regular expression or a single variable reference. \par
We first note that since $\bar{\alpha}$ is simple, every variable $\varsx \in \varset$ has exactly one definition in every $\bar{v} \in \reflang(\alpha_1) \times \reflang(\alpha_2) \times \reflang(\alpha_m)$. Consequently, if some variable $\varsx$ has a definition $\varsx\{\varsy\}$, then this definition and all references of $\varsx$ can be replaced by references of $\varsy$ without changing the set of conjunctive matches. In the following, we can therefore assume that, for every $i \in [m]$ and $j \in [t_i]$, $\alpha_{i, j}$ is a classical regular expression, a variable reference, or a variable definition $\varsx\{\gamma\}$, where $\gamma$ is a classical regular expression. \par
We next modify $G_q$ such that every $(x_i, \alpha_i, y_i)$ is replaced by edges $$(z_{i, 0}, \alpha_{i, 1}, z_{i, 1}), (z_{i, 1}, \alpha_{i, 2}, z_{i, 2}), \ldots, (z_{i, t_i - 1}, \alpha_{i, t_i}, z_{i, t_i})\,,$$ where $z_{i, 0} = x_i$ and $z_{i, t_i} = y_i$. We denote this modified $\CXRPQ$ by $q'$ and let $\bar{\beta} = (\beta_{1}, \ldots, \beta_{m'})$ be its conjunctive xregex, i.\,e., $m' = \sum^m_{i = 1} t_i$ and $(\beta_{1}, \ldots, \beta_{m'})$ equals 
\begin{equation*}
(\alpha_{1, 1}, \ldots, \alpha_{1, t_1}, \alpha_{2, 1}, \ldots, \alpha_{2, t_2}, \ldots, \alpha_{m, 1}, \ldots, \alpha_{m, t_m})\,.
\end{equation*} 
It can be easily seen that $q \equiv q'$.\par
Next, for every $i \in [m']$, we define an $\NFA$ $M_{i}$ with state-set $Q_{i}$ and transition function $\delta_{i}$ as follows:
\begin{itemize}
\item If $\beta_{i} \in \{\gamma, \varsx\{\gamma\}\}$ for some classical regular expression $\gamma$, then $M_{i}$ accepts $\lang(\gamma)$. 
\item If $\beta_{i} = \varsx$ for some $\varsx \in \varset$, then $M_{i}$ accepts $\Sigma^*$. 
\end{itemize}
We now show how to decide $\DBD \models q'$ for a given graph database $\DBD$. \par
For technical reasons, we first add an $\emptyword$-labelled self-loop to every node in $\DBD$ and to every state in every $M_i$. Next, we define a graph $G_{{q'}, \DBD} = (V_{{q'}, \DBD}, E_{{q'}, \DBD})$, where $V_{{q'}, \DBD} = (V_{\DBD})^{m'} \times Q_{1} \times Q_{2} \times \ldots \times Q_{m'}$ and $E_{{q'}, \DBD}$ contains an edge 
\begin{equation*}
((u_1, \ldots, u_{m'}, p_1, \ldots, p_{m'}), (v_1, \ldots, v_{m'}, r_1, \ldots, r_{m'}))
\end{equation*}
if and only if there are $b_1, b_2, \ldots, b_{m'} \in \Sigma \cup \{\eword\}$ such that,
\begin{itemize}
\item for every $i \in [{m'}]$, $r_i \in \delta_i(p_i, b_i)$ and $(u_i, b_i, v_i) \in E_{\DBD}$, and
\item for every $i, i' \in [m']$, if $\beta_i, \beta_{i'} \in \{\varsx, \varsx\{\gamma\}\}$, for some $\varsx \in \varset$ and a classical regular expression $\gamma$, then $b_{i} = b_{i'}$.
\end{itemize}
Let $\bar{u} = (u_1, \ldots, u_{m'}, p_1, \ldots, p_{m'})$ and $\bar{v} = (v_1, \ldots, v_{m'}, r_1, \ldots, r_{m'})$ be two vertices from $G_{{q'}, \DBD}$. We observe that in $G_{{q'}, \DBD}$ there is a path from $\bar{u}$ to $\bar{v}$ if and only if there is a tuple $(w_1, w_2, \ldots, w_{m'}) \in (\Sigma^*)^{m'}$ such that, for every $i \in [{m'}]$, there is a $w_i$-labelled path from $p_i$ to $r_i$ in $M_i$, a $w_i$-labelled path from $u_i$ to $v_i$ in $\DBD$, and, for every $i, i' \in [m']$, if $\beta_i, \beta_{i'} \in \{\varsx, \varsx\{\gamma\}\}$, for some $\varsx \in \varset$ and a classical regular expression $\gamma$, then $w_{i} = w_{i'}$. \par
We say that a vertex $\bar{u}$ of $G_{{q'}, \DBD}$ is \emph{initial}, if, for every $i \in [{m'}]$, $p_i$ is the initial state of $M_i$, and, for every $i, i' \in [{m'}]$, $x_i = x_{i'}$ implies $u_i = u_{i'}$. Analogously, we say that a vertex $\bar{v}$ of $G_{{q'}, \DBD}$ is \emph{final}, if, for every $i \in [{m'}]$, $r_i$ is the final state of $M_i$, and, for every $i, i' \in [{m'}]$, $y_i = y_{i'}$ implies $v_i = v_{i'}$. With the observation from above, we can conclude that $\DBD \models {q'}$ if and only if in $G_{{q'}, \DBD}$ there is a path from some initial to some final vertex. For simplicity, we add two additional vertices $s$ and $t$ to $G_{{q'}, \DBD}$ with an edge $(s, \bar{u})$ and an edge $(\bar{v}, t)$ for every initial vertex $\bar{u}$ and final vertex $\bar{u}$. Checking $\DBD \models {q'}$ can now be done by checking whether there is a path from $s$ to $t$ in $G_{{q'}, \DBD}$.\par
It remains to describe how this can be done in the claimed time and space bound. To this end, we first observe that the initial replacement of variable definition $\varsx\{\varsy\}$ by $\varsy$ can obviously be carried out in space $\bigO(|q|)$, and it does not increase the size of $q$. Then, transforming $\bar{\alpha}$ into $\bar{\beta}$ can also be done in space $\bigO(|q'|)$, where $|q'| = \bigO(|q|)$. Moreover, we can transform each $\beta_i$ into an $\NFA$ $M_i$ of size $\bigO(|\beta_i|)$ in space $\bigO(|\beta_i|)$. Consequently, we can obtain all $M_i$ in space $\bigO(|{q'}|)$. In particular, we note that removing $\eword$-transitions may result in $\NFA$ of size $\bigO(|\beta_i|^2)$, which justifies our construction of $G_{{q'}, \DBD}$ from above that can also handle possible $\eword$-transitions of the $M_i$. \par
Now a vertex $\bar{u}$ from $G_{{q'}, \DBD}$ can be represented by $\bigO({m'})$ pointers to $\DBD$ and $\bigO(1)$ pointers to each of the $M_i$. Moreover, evaluating the edge relation of $G_{{q'}, \DBD}$, i.\,e., checking for fixed vertices $\bar{u}, \bar{v}$ of $G_{{q'}, \DBD}$ whether there is an edge $(\bar{u}, \bar{v})$, can be done as follows. Checking whether there are $b_1, b_2, \ldots, b_{m'} \in \Sigma \cup \{\eword\}$ such that, for every $i \in [{m'}]$, $r_i \in \delta_i(p_i, b_i)$ and $(u_i, b_i, v_i) \in E_{\DBD}$ can be done by consulting the pointers that represent $\bar{u}$ and $\bar{v}$, the transition functions of the $M_i$ and the edge-relation of $\DBD$. Moreover, checking, for every $i, i' \in [m']$, whether $\beta_i, \beta_{i'} \in \{\varsx, \varsx\{\gamma\}\}$, for some $\varsx \in \varset$ and a classical regular expression $\gamma$, and $b_{i} = b_{i'}$, can be done by using two pointers to $q'$. Consequently, we can nondeterministically check whether there is a path from $s$ to $t$ in $G_{{q'}, \DBD}$ in space $\bigO(|q'| (\log(|\DBD|) + \log(|q'|)) = \bigO(|q| (\log(|\DBD|) + \log(|q|)) = \bigO(|q| \log(|\DBD|) + |q|\log(|q|)) = \bigO(|q| \log(|\DBD|))$.
\end{proof}

\subsection{Construction of the Normal Form}\label{sec:normalForm}

The goal of this section is to show the following result  (which also implies that $\vsfCXRPQ$ can be transformed in normal form).

\begin{theorem}
Let $\bar{\alpha} \in \dimconxregex{m}_{\Sigma, \varset}$ be vstar-free. Then there is an equivalent $\bar{\beta} \in \dimconxregex{m}_{\Sigma, \varset'}$ in normal form with $|\bar{\beta}| = \bigO\left(2^{2^{\poly{|\bar{\alpha}|}}}\right)$. 
\end{theorem}

Our construction to transform vstar-free $\bar{\alpha} \in \conxregex$ in normal form has three main steps, represented by Lemmas~\ref{conjunctiveNormalFormLemmaStepOne},~\ref{uniqueDefinitionsLemma}~and~\ref{removeNonBasicLemma}. Before stating these lemmas, we will first explain the three main steps on an intuitive level with an example.


\begin{figure}
\begin{center}
\scalebox{1.5}{\includegraphics{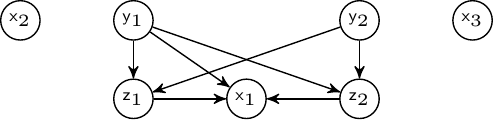}}
\end{center}
\caption{The DAG $G_{\bar{\gamma}}$ at Step $3$.}
\label{DagFigure}
\end{figure}

Let $\bar{\gamma} = (\gamma_1, \gamma_2)$ be a variable star free conjunctive xregex with
\begin{align*}
\gamma_1 &= \varsx\{\ta^* \varsy\{\tb^*\} \ta \varsz\} \altop (\varsx\{\tb^*\} \cdot (\varsz \altop \varsy\{\tc^*\})\\
\gamma_2 &= (\ta^* \altop \varsx) \cdot \varsz\{\varsy \cdot (\ta \altop \tb)\}
\end{align*}
\noindent\textbf{Step 1:} 
We will ``multiply-out'' alternations with definitions or variables, which transforms $\gamma_1$ and $\gamma_2$ into alternations of variable-simple xregex. 
\begin{align*}
\gamma_1 &= [\varsx\{\ta^* \varsy\{\tb^*\} \ta \varsz\}] \altop [\varsx\{\tb^*\} \cdot \varsz] \altop [\varsx\{\tb^*\} \cdot \varsy\{\tc^*\}]\\
\gamma_2 &= [\ta^* \cdot \varsz\{\varsy \cdot (\ta \altop \tb)\}] \altop [\varsx \cdot \varsz\{\varsy \cdot (\ta \altop \tb)\}]
\end{align*}
This only results in an equivalent xregex because $\bar{\gamma}$ is vstar-free. This transformation may cause an exponential size blow-up.\par
\noindent\textbf{Step 2:} Next, we relabel the variables in such a way, that every variable has at most one definition in $\bar{\gamma}$. This also requires that variable references are substituted by several variable references (e.\,g., references of $\varsx$ by $\varsx_{1}\varsx_{2}\varsx_{3}$). The size blow-up is at most quadratic.
\begin{align*}
\gamma_1 =\:&[\varsx_{1}\{\ta^* \varsy_{1}\{\tb^*\} \ta \varsz_{1} \varsz_2\}] \altop [\varsx_{2}\{\tb^*\} \cdot \varsz_{1} \varsz_{2}] \altop [\varsx_{3}\{\tb^*\} \cdot \varsy_{2}\{\tc^*\}]\\
\gamma_2 =\:&[\ta^* \cdot \varsz_{1}\{\varsy_{1} \varsy_{2} \cdot (\ta \altop \tb)\}] \altop [\varsx_{1}\varsx_{2}\varsx_{3} \cdot \varsz_{2}\{\varsy_{1} \varsy_{2} \cdot (\ta \altop \tb)\}]
\end{align*}
\noindent\textbf{Step 3:} We are now able to remove non-basic variable definitions, which remains to be done in order to obtain the normal form. The main idea is to split $\gamma$ of a variable definition $\varsx\{\gamma\}$ into smaller parts and (if they are not already definitions) store them in new variables. Instead of $\varsx$, we can then reference the new variables, which allows to remove $\varsx$.
This, however, has to be done in the order given by the DAG induced by the relation $\preceq_{\bar{\alpha}}$ (see Figure~\ref{DagFigure}). More precisely, we first consider the roots of the DAG representing variables $\varsx_2, \varsx_3, \varsy_1, \varsy_2$, which have basic variable definitions and therefore can be left unchanged. After deleting them from the DAG, the nodes for $\varsz_1, \varsz_2$ are roots. The non-basic definition $\varsz_{1}\{\varsy_{1} \varsy_{2} \cdot (\ta \altop \tb)\}$ is replaced by $\varsu_{1}\{\varsy_{1}\} \varsu_{2}\{\varsy_{2}\} \varsu_{3}\{\ta \altop \tb\}$ and $\varsz_{2}\{\varsy_{1} \varsy_{2} \cdot (\ta \altop \tb)\}$ is replaced by $\varsu_{4}\{\varsy_{1}\} \varsu_{5}\{\varsy_{2}\} \varsu_{6}\{\ta \altop \tb\}$.
All references for $\varsz_1$ and $\varsz_2$ are replaced by $\bar{\varsz}_1 = \varsu_{1}\varsu_{2}\varsu_{3}$ and $\bar{\varsz}_2 = \varsu_{4}\varsu_{5}\varsu_{6}$, respectively.
\begin{align*}
\gamma_1 =\:&[\varsx_{1}\{\ta^* \varsy_{1}\{\tb^*\} \ta \bar{\varsz}_{1} \bar{\varsz}_2\}] \altop [\varsx_{2}\{\tb^*\} \cdot \bar{\varsz}_{1} \bar{\varsz}_{2}] \altop [\varsx_{3}\{\tb^*\} \cdot \varsy_{2}\{\tc^*\}]\\
\gamma_2 =\:&[\ta^* \cdot \varsu_{1}\{\varsy_{1}\} \varsu_{2}\{\varsy_{2}\} \varsu_{3}\{\ta \altop \tb\}] \altop\\
&[\varsx_{1}\varsx_{2}\varsx_{3} \cdot \varsu_{4}\{\varsy_{1}\} \varsu_{5}\{\varsy_{2}\} \varsu_{6}\{\ta \altop \tb\}]
\end{align*}
Finally, after deleting the nodes for $\varsz_1$ and $\varsz_2$ from the DAG, the node for $\varsx_1$ is the only root left. Therefore, we can now replace 
\begin{equation*}
\varsx_{1}\{\ta^* \varsy_{1}\{\tb^*\} \ta \bar{\varsz}_{1} \bar{\varsz}_2\} = \varsx_{1}\{\ta^* \varsy_{1}\{\tb^*\} \ta \varsu_{1}\varsu_{2}\varsu_{3} \varsu_{4}\varsu_{5}\varsu_{6}\}
\end{equation*} 
by the following concatenation of variable definitions:
\begin{align*}
&\varsu_7\{\ta^*\} \varsy_{1}\{\tb^*\} \varsu_8\{\ta\} \varsu_9\{\varsu_{1}\} \varsu_{10}\{\varsu_{2}\}\varsu_{11}\{\varsu_{3}\}\\ 
&\varsu_{12}\{\varsu_{4}\}\varsu_{13}\{\varsu_{5}\}\varsu_{14}\{\varsu_{6}\}\,.
\end{align*} 
Moreover, all references of $\varsx_1$ are replaced by $\bar{\varsx}_1 = \varsu_7 \varsy_{1} \varsu_8 \varsu_9 \ldots \varsu_{14}$. The thus obtained conjunctive xregex is in fact in normal form.\par
We shall discuss a few particularities. While Steps $1$ and $2$ are more or less straightfoward, Step $3$ is more complicated. In particular, it is not easy to see why it is necessary to use new variables $\varsu$ with definition $\varsu\{\varsx\}$ instead of just using the existing $\varsx$. For example, if a variable definition $\varsz\{\varsy_1 \ta^* \varsy_2\}$ would be replaced by $\varsy_1 \varsu\{\ta^*\} \varsy_2$ and references of $\varsz$ by $\varsy_1 \varsu \varsy_2$, then $\varsy_1 \varsu \varsy_2$ must refer to $\eword$ in the case that $\varsz$ is not instantiated, which is not the case if $\varsy_1$ is instantiated elsewhere. So we use $\varsu_1\{\varsy_1\} \varsu_2\{\ta^*\} \varsu_3\{\varsy_2\}$ and replace references of $\varsz$ by $\varsu_1 \varsu_2 \varsu_3$, which means that if $\varsz$ is not instantiated, then also none of the $\varsu_1$, $\varsu_2$ and $\varsu_3$ are instantiated in the modified xregex. Why it is necessary to apply these modifications according to the order given by $\preceq_{\bar{\gamma}}$ will be discussed in detail in the proof of the corresponding lemma further below. The possible exponential size blow-up of this step is also discussed 
in Subsection~\ref{sec:pseudobasic}.\par
We now give the lemmas for the three steps of the construction (it will be helpful to keep in mind that the constructions of these lemmas are illustrated by the example from above). The proofs of the first two lemmas are more or less straightforward proofs of correctness for the constructions of Step $1$ and $2$ illustrated above. The proof of the third lemma is technically more involved.

\begin{lemma}\label{conjunctiveNormalFormLemmaStepOne}
(\textbf{Step 1}) Let $\bar{\alpha} \in \dimconxregex{m}_{\Sigma, \varset}$ be vstar-free. Then there is an equivalent $\bar{\beta} \in \dimconxregex{m}_{\Sigma, \varset}$ with $|\bar{\beta}| = \bigO(2^{|\bar{\alpha}|})$, such that each $\bar{\beta}[i]$ is an alternation of variable-simple xregex.
\end{lemma}

\begin{proof}
Let $\bar{\alpha} = (\alpha_1, \alpha_2, \ldots, \alpha_m)$. For every $i \in [m]$, we transform $\alpha_i$ into $\beta_i = \beta_{i, 1} \altop \beta_{i, 2} \altop \ldots \altop \beta_{i, t_i}$, such that, for every $j \in [t_i]$, $\beta_{i, j}$ is variable-simple. \par
If $\alpha_i$ is not already in the form described above (otherwise we set $\beta_i = \alpha_i$ and are done), i.\,e., an alternation of variable-simple xregex, then it has the form $\alpha_i = \pi_1 \altop \pi_2 \altop \ldots \altop \pi_q$ (note that $q = 1$ is possible), where at least one $\pi_j$ is not valt-free (since $\alpha_i$ is vstar-free, each $\pi_j$ is vstar-free as well). This means that $\pi_j$ has some subexpression $\gamma = (\gamma_1 \altop \gamma_2)$ and $\gamma$ contains a variable definition or a variable reference. Let $\pi_{j, 1}$ be obtained from $\pi_j$ by replacing $\gamma$ by $\gamma_1$ and let $\pi_{j, 2}$ be obtained from $\pi_j$ by replacing $\gamma$ by $\gamma_2$. Finally, let $\alpha'_i = \pi_1 \altop \ldots \altop \pi_{j, 1} \altop \pi_{j, 2} \altop \ldots \altop \pi_q$ and let $\bar{\alpha}'$ be obtained from $\bar{\alpha}$ by replacing $\alpha_i$ by $\alpha'_i$. \par
We first note that $\alpha'_i$ is sequential, variable-acyclic and vstar-free, which means that $\bar{\alpha}'$ is a vstar-free conjunctive xregex. Moreover, since $\alpha_i$ is vstar-free, $\gamma$ is not under a $+$-operator, which implies that $\reflang(\pi_j) = \reflang(\pi_{j, 1} \altop \pi_{j, 2})$ (which would not be the case if $\gamma$ was under a $+$-operator), and therefore also 
\begin{equation*}
\reflang(\pi_1 \altop \pi_2 \altop \ldots \altop \pi_q) = \reflang(\pi_1 \altop \ldots \altop \pi_{j, 1} \altop \pi_{j, 2} \altop \ldots \altop \pi_q)\,, 
\end{equation*}
Consequently, we have $\reflang(\alpha_i) = \reflang(\alpha'_i)$, which, by Lemma~\ref{refLangStableChangesLemma}, means that $\lang(\bar{\alpha}) = \lang(\bar{\alpha}')$. \par
By repeating this step, we can therefore transform $\bar{\alpha}$ into an equivalent $\bar{\beta} \in \dimconxregex{m}_{\Sigma, \varset}$ such that each $\bar{\beta}[i]$ is an alternation of variable-simple xregex. Finally, we note that in each step of this construction, the size of the xregex can double in the worst case. Thus, $|\bar{\beta}| = \bigO(2^{|\bar{\alpha}|})$.
\end{proof}

\begin{lemma}\label{uniqueDefinitionsLemma}
(\textbf{Step 2}) Let $\bar{\alpha} \in \dimconxregex{m}_{\Sigma, \varset}$ such that, for every $i \in [m]$, $\bar{\alpha}[i]$ is an alternation of variable-simple xregex. Then there is an equivalent $\bar{\beta} \in \dimconxregex{m}_{\Sigma, \varset'}$ such that, for every $i \in [m]$, $\bar{\beta}[i]$ is an alternation of variable-simple xregex, and every $\varsx \in \varset'$ has at most one definition in $\bar{\beta}$. Moreover, $|\bar{\beta}| = \bigO(|\bar{\alpha}|^2)$. 
\end{lemma}

\begin{proof}
Let $\bar{\alpha} = (\alpha_1, \alpha_2, \ldots, \alpha_m)$, and, for every $i \in [m]$, let 
\begin{equation*}
\alpha_i = \alpha_{i, 1} \altop \alpha_{i, 2} \altop \ldots \altop \alpha_{i, t_i}\,, 
\end{equation*}
where, for every $j \in [t_i]$, $\alpha_{i, j}$ is variable-simple. 
We transform $\bar{\alpha}$ into an equivalent $\bar{\beta}$ with the property described in the statement of the lemma as follows.\par
Let $\varsx \in \varset$ and let $i \in [m]$ be such that $\alpha_i$ contains a definition of $\varsx$ (since $\bar{\alpha}$ is a conjunctive xregex, there is at most one such $i \in [m]$). Moreover, if, for some $j \in [t_i]$, there are at least two definitions for $\varsx$ in $\alpha_{i, j}$, then $\alpha_{i, j}$ is not sequential or not valt-free, which is a contradiction. Thus, for every $j \in [t_i]$, there is at most one definition for $\varsx$ in $\alpha_{i, j}$, and, for the sake of convenience, we shall assume that, for every $j \in [t_i]$, $\alpha_{i, j}$ contains a definition for $\varsx$ (it will be obvious how to adapt the construction when this is not the case). Now, for every $j \in [t_i]$, we replace in $\alpha_{i, j}$ the variable definition $\varsx\{\gamma\}$ by $\varsx^{(j)}\{\gamma\}$, where $\varsx^{(j)}$ is a new variable. Let $\bar{\alpha}' = (\alpha'_1, \alpha'_2, \ldots, \alpha'_m)$ be the thus obtained conjunctive xregex. Obviously, every variable $\varsx^{(j)}$ with $j \in [t_i]$ has exactly one definition in $\bar{\alpha}'$, but $\bar{\alpha}'$ is not necessarily equivalent to $\bar{\alpha}$, since we did not change the references for variable $\varsx$. Before we do this, we note that, for any ref-word $v \in \reflang(\alpha'_i)$, there is exactly one $j \in [t_i]$ such that $v$ contains a definition of $\varsx^{(j)}$ (or, in other words, for exactly one $j \in [t_i]$ the definition for $\varsx^{(j)}$ is instantiated), while for all other $j' \in [t_i] \setminus \{j\}$, there is no definition of $\varsx^{(j')}$, which means that the image for $\varsx^{(j')}$ with respect to $v$'s variable mapping is necessarily empty. This is due to the fact that $\alpha_i$ is sequential, variable alternation-free and we assume that, for every $j \in [t_i]$, $\alpha_{i, j}$ has a definition of $\varsx$. Hence, we can simply substitute \emph{all} variable references for $\varsx$ in $\bar{\alpha}'$ (not only those in $\alpha_i$) by $\bar{\varsx} = \prod^{t_i}_{j = 1} \varsx^{(j)}$. We denote the thus modified variant of $\bar{\alpha}$ by $\bar{\alpha}'' = (\alpha''_1, \alpha''_2, \ldots, \alpha''_m)$. \par
It is straightforward to see that $\bar{\alpha}''$ is sequential and variable-acyclic; thus, it is a conjunctive xregex, and also that, for every $i \in [m]$, $\alpha''_i$ is an alternation of variable-simple xregex. With the observations from above, it is also easy to see that $\bar{\alpha}$ and $\bar{\alpha}''$ are equivalent. More precisely, every tuple $\bar{v} \in \reflang(\alpha_1) \times \reflang(\alpha_2) \times \ldots \times \reflang(\alpha_m)$ corresponds to a tuple $\bar{v}'' \in \reflang(\alpha''_1) \times \reflang(\alpha''_2) \times \ldots \times \reflang(\alpha''_m)$, obtained from $\bar{v}$ by replacing each variable reference $\varsx$ by $\bar{\varsx}$, and the variable definition $\open{\varsx} u \close{\varsx}$ by $\open{\varsx^{(j)}} u \close{\varsx^{(j)}}$, where $j \in [t_i]$ is such that the variable definition of $\varsx$ in $\bar{v}$ was generated by $\alpha_{i, j}$.\par
By repeating this modification step with respect to every variable from $\varset$, we can obtain a conjunctive xregex $\bar{\beta} = (\beta_1, \beta_2, \ldots, \beta_m)$ with the desired property that is equivalent to $\bar{\alpha}$. 
Moreover, this modification replaces each variable reference by a sequence of $\bigO(|\bar{\alpha}|)$ variable references; thus, $|\bar{\beta}| = \bigO(|\bar{\alpha}|^2)$. 
\end{proof}

\begin{lemma}\label{removeNonBasicLemma}
(\textbf{Step 3}) Let $\bar{\alpha}\in \dimconxregex{m}_{\Sigma, \varset}$ such that, for every $i \in [m]$, $\bar{\alpha}[i]$ is an alternation of variable-simple xregex, and every $\varsx \in \varset$ has at most one definition in $\bar{\alpha}$. Then there is an equivalent $\bar{\beta} \in \dimconxregex{m}_{\Sigma, \varset'}$ in normal form with $|\bar{\beta}| = \bigO(|\bar{\alpha}|^{|\varset| + 1})$. 
\end{lemma}

\begin{proof}
Let $\bar{\alpha} = (\alpha_1, \alpha_2, \ldots, \alpha_m)$, and, for every $i \in [m]$, let 
\begin{equation*}
\alpha_i = \alpha_{i, 1} \altop \alpha_{i, 2} \altop \ldots \altop \alpha_{i, t_i}\,, 
\end{equation*}
where, for every $j \in [t_i]$, $\alpha_{i, j}$ is variable-simple, and, for every $\varsx \in \varset$, there is at most one definition for $\varsx$ in $\bar{\alpha}$. First, we define as general modification step, which can be applied to any variable definition in $\bar{\alpha}$. This modification step will then be used in order to transform $\bar{\alpha}$ into normal form.\medskip\par
\noindent\textbf{Main modification step:} Let $\varsz\{\gamma\}$ be a variable definition of $\bar{\alpha}$. Since $\gamma$ is variable-simple, $\gamma = \gamma_1 \gamma_2 \ldots \gamma_p$ such that each $\gamma_{\ell}$ with $\ell \in [p]$ is either a classical regular expression, a variable definition, or a variable reference.\par
For every $\ell \in [p]$, we define a $\gamma'_\ell$ as follows. If $\gamma_\ell$ is a variable definition $\varsy_{\ell}\{\ldots\}$, then we set $\gamma'_\ell = \gamma_\ell$. If $\gamma_\ell$ is a classical regular expression or a variable reference, then we set $\gamma'_\ell = \varsy_{\ell} \{\gamma_\ell\}$ for a \emph{new} variable $\varsy_{\ell} \notin \varset$. Note that $\gamma'_1 \gamma'_2 \ldots \gamma'_p = \varsy_1\{\ldots\} \varsy_2\{\ldots\} \ldots \varsy_p\{\ldots\}$ is a concatenation of variable definitions. Next, we replace the variable definition $\varsz\{\gamma\}$ in $\bar{\alpha}$ by $\gamma'_1 \gamma'_2 \ldots \gamma'_p$. Then, we replace all variable references of $\varsz$ in $\bar{\alpha}$ by $\varsy_1 \varsy_2 \ldots \varsy_p$.\par
Let the thus modified version of $\bar{\alpha}$ be denoted by $\bar{\beta} = (\beta_1$, $\beta_2$, $\ldots$, $\beta_m)$. It can be easily seen that $\bar{\beta}$ is sequential and variable acyclic and therefore a conjunctive xregex. Moreover, for every $i \in [m]$, $\beta_i$ is still an alternation of variable-simple xregex. \medskip\\
\noindent \textbf{Correctness of the main modification step}: Next, we show that $\lang(\bar{\alpha}) = \lang(\bar{\beta})$. 
The difficulty in doing so is due to the fact that we cannot assume, for every $i \in [m]$, that $\reflang(\alpha_i) = \reflang(\beta_i)$; i.\,e., we cannot conveniently apply Lemma~\ref{refLangStableChangesLemma}, as it has been done in the proof of Lemma~\ref{conjunctiveNormalFormLemmaStepOne}.\par
We need a few more notational preliminaries. We assume that $r \in [m]$ and $\widehat{j} \in [t_r]$ is such that the modified variable definition $\varsz\{\gamma\}$ is in $\alpha_{r, \widehat{j}}$.
Moreover, let $\beta_r = \beta_{r, 1} \altop \beta_{r, 2} \altop \ldots \altop \beta_{r, t_r}$, where, for every $j \in [t_r]$, $\beta_{r, j}$ is obtained from $\alpha_{r, j}$ by the construction from above (i.\,e., $\beta_{r, \widehat{j}}$ is obtained from $\alpha_{r, \widehat{j}}$ by replacing the definition $\varsz\{\gamma\}$ with $\gamma'_1 \gamma'_2 \ldots \gamma'_p$ and all references $\varsz$ with $\varsy_1 \varsy_2 \ldots \varsy_p$, and all other $\beta_{r, j}$, with $j \in [t_r] \setminus \{\widehat{j}\}$, are obtained from $\alpha_{r, j}$ by only replacing all references $\varsz$ with $\varsy_1 \varsy_2 \ldots \varsy_p$). Furthermore, all other $\beta_{r'}$ with $r' \in [m] \setminus \{r\}$ are obtained from $\alpha_{r'}$ by only replacing each variable reference for $\varsz$ by $\varsy_1 \varsy_2 \ldots \varsy_p$.\par
Let $\bar{w} = (w_1, w_2, \ldots, w_m)$ be a conjunctive match for $\bar{\alpha}$, with variable mapping $\psi_{\bar{\alpha}}$, i.\,e., for every $i \in [m]$, there is a ref-word $v_{i, \alpha} = v'_{i, \alpha} \# v''_{i, \alpha} \in \reflang(\interxregex{\alpha_{i}})$ with variable mapping $\psi_{\bar{\alpha}}$ and with $\deref(v_{i, \alpha}) = \deref(v'_{i, \alpha}) \# w_i$. In order to show that $\bar{w}$ is a conjunctive match for $\bar{\beta}$, we have to show that, for every $i \in [m]$, there are ref-words $v_{i, \beta} = v'_{i, \beta} \# v''_{i, \beta} \in \reflang(\interxregex{\beta_i})$ with $\deref(v_{i, \beta}) = \deref(v'_{i, \beta}) \# w_i$, and $\varmap{v_{1, \beta}} = \varmap{v_{2, \beta}} = \ldots = \varmap{v_{m, \beta}}$. \par
There are two cases, which we shall first discuss on an intuitive level. The simple case is when the ref-word $v''_{r, \alpha} \in \reflang(\alpha_r)$ (i.\,e., the ref-word that yields $w_r$) can be constructed by some $\alpha_{r, j}$ with $j \in [t_r] \setminus \{\widehat{j}\}$; more precisely, we have $v''_{r, \alpha} \in \reflang(\alpha_{r, j})$ for some $j \in [t_r] \setminus \{\widehat{j}\}$. In this case, the modified variable definition $\varsz\{\gamma\}$ is not instantiated. This means that suitable ref-words $v_{i, \beta} = v'_{i, \beta} \# v''_{i, \beta}$ that witness that $\bar{w}$ is a conjunctive match for $\bar{\beta}$ can be easily obtained from the witnesses $v'_{i, \alpha} \# v''_{i, \alpha}$. We only have to take care of the variables $\varsy_\ell$, for which references may occur in $v''_{i, \alpha}$. However, these variables, regardless of whether they have been newly introduced by the modification step or whether definitions for them have already been present before, have only definitions in $\alpha_{r, \widehat{j}}$ (note that for the variables that are not newly introduced, this is only true since we assume that every variable has at most one definition in $\bar{\alpha}$), and therefore their images must be empty.\par
The more complicated case is when, for every $j \in [t_r] \setminus \{\widehat{j}\}$, $v''_{r, \alpha}$ is not in $\reflang(\alpha_{r, j})$, which means that the only way for $\alpha_r$ to produce $v''_{r, \alpha}$ is to use the subexpression $\alpha_{r, \widehat{j}}$, which is the one that contains the variable definition $\varsz\{\gamma\} = \varsz\{\gamma_1 \gamma_2 \ldots \gamma_p\}$ that has been modified by the construction. Moreover, since $\alpha_{r, \widehat{j}}$ is variable-simple, we also know that it must be instantiated by $v''_{r, \alpha}$, i.\,e., $v''_{r, \alpha} \in \reflang(\alpha_{r, \widehat{j}})$ and $v''_{r, \alpha}$ does contain a factor $\open{\varsz} g_1 g_2 \ldots g_p \close{\varsz}$, where, for every $\ell \in [p]$, $g_{\ell} \in \reflang(\gamma_i)$. In this case, we have to obtain the witnesses $v_{i, \beta} = v'_{i, \beta} \# v''_{i, \beta}$ from the witnesses $v_{i, \alpha} = v'_{i, \alpha} \# v''_{i, \beta}$ in such a way that the $g_\ell$ parts are handled by the variables $\varsy_{\ell}$. We now discuss these two cases more formally.
\begin{enumerate}
\item[(C1)] [$v''_{r, \alpha} \in \reflang(\alpha_{r, j})$ for some $j \in [t_r] \setminus \{\widehat{j}\}$] \par 
In this case, for every $i \in [m]$, we let $v''_{i, \beta}$ be the ref-word obtained from $v''_{i, \alpha}$ by replacing each variable reference $\varsz$ by $\varsy_1 \varsy_2 \ldots \varsy_p$. From the fact that, for every $i \in [m]$, $v''_{i, \alpha} \in \reflang(\alpha_{r, i})$, and from the assumption made in this case, it follows that also $v''_{i, \beta} \in \reflang(\beta_{r, i})$.\par
By construction, we have $\varprefix{\beta_{r}} = \varprefix{\alpha_{r}}$. Consequently, we can set $v'_{r, \beta} = v'_{r, \alpha}$ and observe that $v_{r, \beta} = v'_{r, \beta} \# v''_{r, \beta} \in \reflang(\interxregex{\beta_r})$. It remains to show that $\deref(v_{r, \beta})$ equals $\deref(v'_{r, \beta}) \# w_r$. To this end, we first note that, for every variable $\varsx$ different from $\varsz$ or any $\varsy_{\ell}$, we have $\varmap{v_{r, \beta}}(\varsx) = \varmap{v_{r, \alpha}}(\varsx)$. Moreover, the variables $\varsy_{\ell}$ have no definition in $v''_{r, \beta}$, which means that $\varmap{v_{r, \beta}}(\varsy_\ell) = \emptyword$, and since $v''_{r, \alpha} \in \reflang(\alpha_{r, j})$ for some $j \in [t_r] \setminus \{\widehat{j}\}$, we also have $\varmap{v_{r, \alpha}}(\varsz) = \emptyword$. Finally, since $v''_{i, \beta}$ is obtained from $v''_{i, \alpha}$ by replacing each variable-reference $\varsz$ by $\varsy_1 \varsy_2 \ldots \varsy_p$ and since $\deref(v_{r, \alpha}) = \deref(v'_{r, \alpha}) \# w_r$, it follows that $\deref(v_{r, \beta}) = \deref(v'_{r, \beta}) \# w_r$.\par
%
It remains to define $v_{i, \beta}$ for every $i \in [m] \setminus \{r\}$. To this end, for every $i \in [m] \setminus \{r\}$, we define $v''_{i, \beta}$ analogously to $v''_{r, \beta}$ done above, i.\,e., we let $v''_{i, \beta}$ be obtained from $v''_{i, \alpha}$ by replacing each variable reference $\varsz$ by $\varsy_1 \varsy_2 \ldots \varsy_p$. As above, we note that, for every $i \in [m]$, $v''_{i, \beta} \in \reflang(\beta_{r, i})$.\par
For every $i \in [m] \setminus \{r\}$, due to the new variables $\varsy_\ell$ with $\ell \in [p]$, it is not the case that $\varprefix{\beta_{i}} = \varprefix{\alpha_{i}}$. Therefore, for every $i \in [m] \setminus \{r\}$, let $v'_{i, \beta}$ be such that $\varmap{v'_{i, \beta}}$ is equal to $\varmap{v'_{i, \alpha}}$ extended by $\varmap{v'_{i, \beta}}(\varsy_1) = \varmap{v'_{i, \beta}}(\varsy_2) = \ldots = \varmap{v'_{i, \beta}}(\varsy_p) = \eword$. We observe that $v_{i, \beta} = v'_{i, \beta} \# v''_{i, \beta} \in \reflang(\interxregex{\beta_i})$ and $\deref(v_{i, \beta}) = \deref(v'_{i, \beta}) \# w_i$, since clearly $\varmap{v_{i, \beta}}(\varsz) = \emptyword$. \par
Finally, we note that all $\varmap{v_{i, \beta}}$ with $i \in [m]$ are identical to $\varmap{v_{i, \alpha}}$ for all variables $\varsx \notin \{\varsy_1, \varsy_2, \ldots, \varsy_p\}$, and that $\varmap{v_{i, \beta}}(\varsy_1) = \varmap{v_{i, \beta}}(\varsy_2) = \ldots = \varmap{v_{i, \beta}}(\varsy_p) = \eword$. 
\item[(C2)] [$v''_{r, \alpha} \in \reflang(\alpha_{r, \widehat{j}}) \setminus (\bigcup_{i \in [t_r] \setminus \{\widehat{j}\}} \reflang(\alpha_{r, i}))$] \par
Since $\alpha_{r, \widehat{j}}$ is valt-free, the assumption made in this case implies that $v''_{r, \alpha}$ must contain a definition $\open{\varsz} g_1 g_2 \ldots g_p \close{\varsz}$, where, for every $\ell \in [p]$, $g_\ell \in \reflang(\gamma_\ell)$. We let $v''_{r, \beta}$ be obtained from $v''_{r, \alpha}$ by replacing $\open{\varsz} g_1 g_2 \ldots g_p \close{\varsz}$ by $\open{\varsy_1} g_1 \close{\varsy_2} \open{\varsy_2} g_2 \close{\varsy_2} \ldots \open{\varsy_p} g_p \close{\varsy_p}$ and all occurrences of $\varsz$ by $\varsy_1 \varsy_2 \ldots \varsy_p$. By construction, $v''_{r, \beta} \in \reflang(\beta_{r, \widehat{j}})$ is satisfied. As in Case~$1$, we have $\varprefix{\beta_{r}} = \varprefix{\alpha_{r}}$ and can therefore set $v'_{r, \beta} = v'_{r, \alpha}$ to get that $v_{r, \beta} = v'_{r, \beta} \# v''_{r, \beta} \in \reflang(\interxregex{\beta_r})$. In order to see that $\deref(v_{r, \beta}) = \deref(v'_{r, \beta}) \# w_r$, we again note that for every variable $\varsx$ different from $\varsz$ or any $\varsy_{\ell}$, we have that $\varmap{v_{r, \beta}}(\varsx) = \varmap{v_{r, \alpha}}(\varsx)$. Moreover, for the variables $\varsy_{\ell}$, we have that $\varmap{v_{r, \beta}}(\varsy_\ell) = u_\ell$, such that $u_1 u_2 \ldots u_p = \varmap{v_{r, \alpha}}(\varsz)$. This follows from the fact that $v''_{r, \beta}$ is obtained from $v''_{r, \alpha}$ by replacing $\open{\varsz} g_1 g_2 \ldots g_p \close{\varsz}$ by $\open{\varsy_1} g_1 \close{\varsy_2} \open{\varsy_2} g_2 \close{\varsy_2} \ldots \open{\varsy_p} g_p \close{\varsy_p}$. This directly implies that $\deref(v_{r, \beta}) = \deref(v'_{r, \beta}) \# w_r$.\par
For every $i \in [m] \setminus \{r\}$, we let $v''_{i, \beta}$ be obtained from $v''_{i, \alpha}$ by replacing each variable-reference $\varsz$ by $\varsy_1 \varsy_2 \ldots \varsy_p$. By construction, it follows that $v''_{i, \beta} \in \reflang(\beta_{i})$. For every $i \in [m] \setminus \{r\}$, $\beta_i$ does not contain a definition for $\varsy_\ell$, which means that $\varprefix{\beta_{i}}$ contains the definition $\varsy_\ell\{\Sigma^*\}$. Thus, for every $i \in [m] \setminus \{r\}$, we can define $v'_{i, \beta}$ such that $\varmap{v'_{i, \beta}}$ is equal to $\varmap{v'_{i, \alpha}}$ up to the following exceptions: for every $\ell \in [p]$, $\varmap{v'_{i, \beta}}(\varsy_\ell) = \varmap{v_{r, \beta}}(\varsy_\ell)$ (recall that $v_{r, \beta}$ contains a definition for each $\varsy_{\ell}$). By construction, it can be easily seen that in fact $v_{i, \beta} = v'_{i, \beta} \# v''_{i, \beta} \in \reflang(\interxregex{\beta_r})$ and $\deref(v_{i, \beta}) = \deref(v'_{i, \beta}) \# w_i$ holds for every $i \in [m] \setminus \{r\}$. Furthermore, we observe that all $v_{i, \beta}$ with $i \in [m]$ are defined in such a way that they have the same variable mapping. 
\end{enumerate}
This shows that every conjunctive match for $\bar{\alpha}$ is also a conjunctive match for $\bar{\beta}$. Moreover, on close inspection it can be noted that the way how we obtained appropriate ref-words $v_{i, \beta}$ from the ref-words $v_{i, \alpha}$ can be reversed in order to prove that every conjunctive match for $\bar{\beta}$ is also a conjunctive match for $\bar{\alpha}$. More precisely, we replace $\varsy_1 \varsy_2 \ldots \varsy_p$ by $\varsz$ and $\open{\varsy_1} g_1 \close{\varsy_2} \open{\varsy_2} g_2 \close{\varsy_2} \ldots \open{\varsy_p} g_p \close{\varsy_p}$ by $\open{\varsz} g_1 g_2 \ldots g_p \close{\varsz}$ instead, and instead of changing the ref-words $v'_{i, \alpha}$ such that the variable mappings are extended for variables $\varsy_\ell$, we change the ref-words $v'_{i, \beta}$ such that the variable mappings are restricted accordingly (in particular, in Case $2$, $v'_{i, \alpha}$ must be chosen such that the image of $\varsz$ equals the concatenation of the images of $\varsy_{\ell}$ as determined by $v'_{r, \alpha}$). \par
Consequently, $\lang(\bar{\alpha}) = \lang(\bar{\beta})$, which conclude the proof of the correctness of the main modification step.\medskip\par
\noindent\textbf{Transforming $\bar{\alpha}$ in normal form:} It remains to show how $\bar{\alpha}$ can be transformed into normal form by applications of the main modification step. \par
The idea is to apply the main modification step to each non-basic variable definition of $\bar{\alpha}$, in order to transform them into basic ones. However, this does not necessarily work if we do not apply the modification steps in a specific order. In order to illustrate the potential problem, let $\varsz\{\gamma\}$ be a variable definition and let there be some other variable definition $\varsx\{\gamma'\}$, such that $\gamma'$ contains a reference of $\varsz$. Then the main modification step will replace $\varsz\{\gamma\}$ by a concatenation of variable definitions, but also the reference of $\varsz$ in $\gamma'$ by a concatenation $\varsy_1 \varsy_2 \ldots \varsy_p$ of variable references. In the case that $\gamma'$ is a single variable references and $p \geq 2$, this will turn a basic variable definition, namely $\varsx\{\varsz\}$, into a non-basic one, namely $\varsx\{\varsy_1 \varsy_2 \ldots \varsy_p\}$. \par
As mentioned above, we avoid this by applying the main modifiation steps to the non-basic variable definitions in a particular order. To this end, we recall the binary relation $\preceq_{\bar{\alpha}}$ over $\varset$ with $\varsx \preceq_{\bar{\alpha}} \varsy$ if $\varsx$ has a reference or a definition in the definition of $\varsy$. In particular, we consider the directed graph $G_{\bar{\alpha}} = (\varset, \{(\varsx, \varsy) \mid \varsx \preceq_{\bar{\alpha}} \varsy\})$. \par
Since $\preceq_{\bar{\alpha}}$ is acyclic, we know that $G_{\bar{\alpha}}$ is a directed acyclic graph (DAG). Moreover, the roots of $G_{\bar{\alpha}}$ (i.\,e., nodes without incoming arcs) correspond to the minimal elements of $\preceq_{\bar{\alpha}}$, which are exactly the variables $\varsx$ whose variable definitions do not contain any references or definitions.
We now apply the main modification step to the non-basic variable definition governed by the structure of $G_{\bar{\alpha}}$.\par
We repeatedly choose some root $\varsx$ of (the current version of) $G_{\bar{\alpha}}$, perform a modification on $\bar{\alpha}$ (which might be the identity) and then we delete this root $\varsx$. Since $G_{\bar{\alpha}}$ is a DAG, this step either deletes the last node of $G_{\bar{\alpha}}$, which terminates the procedure, or after its application there is at least one other root in $G_{\bar{\alpha}}$. Consequently, this procedure terminates after $|\varset|$ steps. A single step works as follows. \par
If the definition $\varsx\{\gamma\}$ of the root $\varsx$ is basic, i.\,e., $\gamma$ is a classical regular expression or a single variable reference, then we will just remove the root $\varsx$ from $G_{\bar{\alpha}}$ without modifying $\bar{\alpha}$. If, on the other hand, $\varsx\{\gamma\}$ is not basic, then we apply the main modification step to $\varsx\{\gamma\}$ (which modifies $\bar{\alpha}$), and then we remove the root $\varsx$ from $G_{\bar{\alpha}}$. Note that in this procedure, $\varsx\{\gamma\}$ always refers to the definition of $\varsx$ in the current version of $\bar{\alpha}$, which may have been changed by applications of the main modification steps, i.\,e., variable references in the original definition of $\varsx$ might have been replaced by concatenations of new variable references.\par
Let $\bar{\beta} = (\beta_1, \beta_2, \ldots, \beta_m)$ be the conjunctive xregex obtained by this procedure. Due to the correctness of the main modification step, we know that $\lang(\bar{\alpha}) = \lang(\bar{\beta})$, and that, for every $i \in [m]$, $\beta_i$ is still an alternation of variable-simple xregex. In order to conclude that $\bar{\beta}$ is in normal form, it only remains to show that $\overline{\beta}$ does not contain non-basic variable definitions. To this end, we first state some observations about the procedure described above:
\begin{enumerate}
\item\label{modStepApplCorrPone} References of new variables that are introduced by the procedure will never be replaced anymore. References of variables from $\varset$ that have already been considered by the procedure, but are not removed by an application of the main modification step, will also never be replaced anymore
\item\label{modStepApplCorrPtwo} If the definition of a variable $\varsx \in \varset$ is initially non-basic, then this variable will necessarily be removed by the procedure. This is due to the fact that the procedure considers each variable from $\varset$, and that other applications of the main modification step with respect to some variables different from $\varsx$ cannot transform the definition of $\varsx$ into a basic one.
\item\label{modStepApplCorrPthree} The definition of a variable $\varsx\{\gamma\}$, where $\gamma$ is a classical regular expression, is never changed by the procedure. This holds for the case that $\varsx \in \varset$ and $\varsx\{\gamma\}$ is its initial definition in $\bar{\alpha}$ before the procedure, and for the case that $\varsx$ is a new variable and $\varsx\{\gamma\}$ is created by the application of the main modification step. 
\item\label{modStepApplCorrPfour} Whenever a variable definition $\varsx\{\gamma\}$ is changed in such a way that references in $\gamma$ are replaced by concatenations of references, then $\varsx \in \varset$ and $\varsx$ has not yet been considered by the procedure. This can be seen as follows. If $\varsx \notin \varset$, i.\,e., $\varsx$ is a new variable, then there was some application of the main modification step with respect to some $\varsz$ that has created variable $\varsx$ with an initial definition $\varsx\{\gamma'\}$. By definition of the main modification step, this means that $\gamma'$ is either a classical regular expression, or a single variable reference. In the first case, due to Point~\ref{modStepApplCorrPthree}, $\varsx\{\gamma'\}$ cannot be changed anymore, which contradicts our assumption. Thus, $\gamma' = \varsy$. If $\varsy$ is a new variable, then $\varsy$ is not replaced, due to Point~\ref{modStepApplCorrPone}. If $\varsy \in \varset$, then this means that $\varsy$ has a reference in the definition of the variable $\varsz$, which means that it must already have been considered when the modification step is carried out with respect to $\varsz$. Hence, due to Point~\ref{modStepApplCorrPone}, it will not be replaced anymore. Therefore, we can assume that $\varsx \in \varset$.\par
If $\varsx$ has already been considered by the procedure, then, since it has not been removed, it initially had a basic definition $\varsx\{\gamma'\}$. If $\gamma'$ is a regular expression, then, due to Point~\ref{modStepApplCorrPthree}, it cannot be changed, which is a contradiction. Therefore, $\gamma' = \varsy$, and in the same way as above, we can conclude that $\varsy$ cannot be a new variable. This means that $\varsy$ has already been considered by the procedure and has not been removed, so, due to Point~\ref{modStepApplCorrPone}, this reference will not be changed anymore. This is a contradiction.
%
%
\end{enumerate}
Now let us assume that $\bar{\beta}$ contains a non-basic variable definition $\varsx\{\gamma\}$, and let us first consider the case that $\varsx \in \varset$. Due to Point~\ref{modStepApplCorrPtwo} from above, the original definition of $\varsx$ in $\bar{\alpha}$ must be basic, and due to Point~\ref{modStepApplCorrPthree}, it must have the form $\varsx\{\varsz\}$. If, in the process of the procedure, the definition of $\varsx$ is changed in such a way that references are replaced by concatenations of references (i.\,e., it is changed into a non-basic definition), then, due to Point~\ref{modStepApplCorrPfour}, it will necessarily be considered by the procedure and the main modification step is applied to it, which removes it. This is a contradiction.\par
Next, let us assume that $\varsx \notin \varset$, so $\varsx$ is introduced in some application of the main modification step. This means that initially its definition is basic and has the form $\varsx\{\varsz\}$. By assumption, its definition is not basic after termination of the procedure, thus, it must be changes such that $\varsz$ is replaced by a concatenation of references. Due to Point~\ref{modStepApplCorrPthree}, this means that $\varsx \in \varset$, which is a contradiction.\par
This shows that after the procedure terminates, there are no non-basic variable definitions and therefore $\bar{\beta}$ is in normal form. We next estimate the size of $\bar{\beta}$, and the time and space required to construct it. \par
In the procedure that transforms $\bar{\alpha}$ into $\bar{\beta}$, there are $k \leq |\varset|$ applications of the main modification step. For every $i \in [k] \cup \{0\}$, let $\bar{\alpha}^{(i)}$ be the version of $\bar{\alpha}$ after the $i^{\text{th}}$ application of the main modification step in the procedure that transforms $\bar{\alpha}$ into $\bar{\beta}$. In particular, $\bar{\alpha}^{(0)} = \bar{\alpha}$ and $\bar{\alpha}^{(k)} = \bar{\beta}$. For every $\varsx \in \varset$, let $k_{\varsx}$ be the number of references of $\varsx$ in $\bar{\alpha}$. We claim that, for every $i \in [k] \cup \{0\}$, $|\bar{\alpha}^{(i)}| = \bigO(|\bar{\alpha}|^{i+1})$. For $i = 0$, this obviously holds. Now let $i \in [k-1] \cup \{0\}$ and assume that $|\bar{\alpha}^{(i)}| = \bigO(|\bar{\alpha}|^{i+1})$. Moreover, let the $(i+1)^{\text{th}}$ application of the main modification step apply to variable $\varsx$. This means that $\bar{\alpha}^{(i + 1)}$ is obtained in the $(i + 1)^{\text{th}}$ application of the main modification step by replacing $k_{\varsx}$ symbols in $\bar{\alpha}^{(i)}$ by at most $|\bar{\alpha}^{(i)}|$ symbols. Hence, $|\bar{\alpha}^{(i + 1)}| = \bigO(k_{\varsx} |\bar{\alpha}^{(i)}|) = \bigO(|\bar{\alpha}| |\bar{\alpha}^{(i)}| = \bigO(|\bar{\alpha}| |\bar{\alpha}|^{i+1}) = \bigO(|\bar{\alpha}|^{i+2})$. Consequently, $|\bar{\beta}| = \bigO(|\bar{\alpha}|^{k + 1}) = \bigO(|\bar{\alpha}|^{|\varset| + 1})$. \par
It can be easily verified that $\bar{\beta}$ can also be computed in space $\bigO(|\bar{\beta}|)$ and in time $\bigOstar(|\bar{\beta}|)$

\end{proof}

We are now sufficiently prepared to give a proof of the upper bound claimed in Theorem~\ref{vsfCXRPQEvalUpperBoundTheorem} in the next subsection.

\subsection{Proof of Theorem~\ref{vsfCXRPQEvalUpperBoundTheorem}}\label{sec:sketchUpperBound}

%
%

Theorem~\ref{vsfCXRPQEvalUpperBoundTheorem} is a consequence from the following lemma.

\begin{lemma}\label{vsfCXRPQEvalUpperBoundTheoremTwo}
Given a Boolean $q \in \vsfCXRPQ$ and a graph database $\DBD$, we can nondeterministically check whether $\DBD \models q$ in space $\bigO(2^{\poly(|q|)} \log(|\DBD|))$.
\end{lemma}

\begin{proof} 
Let $q \in \vsfCXRPQ$ be Boolean and represented by the graph pattern $G_q$ with $E_{\DBD} = \{(x_i, \alpha_i, y_i) \mid i \in [m]\}$, and let $\DBD$ be a graph-database. We define a nondeterministic procedure to check whether or not $\DBD \models q$. First, for every $i \in [m]$, we transform $\alpha_i$ as follows. As long as $\alpha_i$ is not valt-free, we choose some subexpression $(\gamma_1 \altop \gamma_2)$ where $\gamma_1$ or $\gamma_2$ contains a variable definition or a variable reference, and we nondeterministically replace it by $\gamma_1$ or $\gamma_2$. In this way, we obtain a variable-simple conjunctive xregex $\bar{\alpha}' = (\alpha'_1, \alpha'_2, \ldots, \alpha'_m)$. Moreover, this can clearly be done in space $\bigO(|q|)$, and, for every $\bar{w} \in (\Sigma^*)^m$, if $\bar{w} \in \lang(\bar{\alpha})$, then it is possible to perform the nondeterministic guesses in such a way that also $\bar{w} \in \lang(\bar{\alpha}')$. In order to see this, it is sufficient to observe that this nondeterministic construction is similar to how vstar-free conjunctive xregex are made variable-simple in the proof of Lemma~\ref{conjunctiveNormalFormLemmaStepOne}; the difference is that we cannot afford to explicitly store all possibilities of resolving alternations with variable references or definitions, so they need to be nondeterministically guessed. In particular, this implies that if $\DBD \models q$, then it is possible to perform the nondeterministic choices such that $\DBD \models q'$, where $q'$ is the $\vsfCXRPQ$ obtained from $q$ by replacing each $\alpha_i$ by $\alpha'_i$ (note that this uses Proposition~\ref{xregexDetermineGraphQueryProposition}). \par
Next, we transform $\bar{\alpha}'$ into $\bar{\alpha}''$ according to the constructions used in the proof of Lemma~\ref{uniqueDefinitionsLemma} and we note that $|\bar{\alpha}''| = \bigO(|\bar{\alpha}'|^2)$. Then, we transform $\bar{\alpha}''$ into $\bar{\beta} = (\beta_1, \beta_2, \ldots, \beta_m)$ according to the construction used in the proof of Lemmas~\ref{removeNonBasicLemma} and we note that $|\bar{\beta}| = \bigO(|\bar{\alpha}''|^{|\bar{\alpha}''| + 1}) = \bigO((|\bar{\alpha}'|^2)^{|\bar{\alpha}'|^2 + 1}) = \bigO(2^{\poly(|\bar{\alpha}|)})$. By $q''$ we denote the query obtained from $q'$ by replacing each $\alpha'_i$ by $\beta_i$. We note that $|q''| = \bigO(2^{\poly(|q|)})$ and, according to Lemmas~\ref{uniqueDefinitionsLemma}~and~\ref{removeNonBasicLemma}, $q' \equiv q''$. Thus, if $\DBD \models q$, then it is also possible to perform the initial nondeterministic guesses in such a way that also $\DBD \models q''$. \par
Next, we use Lemma~\ref{simpleCXRPQLemma} to nondeterministically decide whether $\DBD \models q''$ in space $$\bigO(|q''| \log(|\DBD|) + |q''|\log(|q''|)) = 
\bigO(2^{\poly(|q|)} \log(|\DBD|))\,.$$
%
%
Thus, this whole procedure decides nondeterministically whether $\DBD \models q$ in space $\bigO(2^{\poly(|q|)} \log(|\DBD|))$.
\end{proof}

\subsection{$\pspaceclass$ Combined-Complexity}\label{sec:pseudobasic}

Step $3$ of the normal form construction (Lemma~\ref{removeNonBasicLemma}) works for $\bar{\alpha}$, where, for every $i \in [m]$, $\bar{\alpha}[i]$ is an alternation of variable-simple xregex and every $\varsx \in \varset$ has at most one definition. However, in the proof of Lemma~\ref{vsfCXRPQEvalUpperBoundTheoremTwo}, when we apply Step $3$, we can make the much stronger assumption that $\bar{\alpha}$ is even variable-simple and every $\varsx \in \varset$ has at most one definition in $\bar{\alpha}$.
Unfortunately, as illustrated by the following example, this does not make any difference with respect to the possible exponential size blow-up caused by Step $3$.
\begin{equation*}
\alpha = \varsx_1\{\ta\} \varsx_2\{\varsx_1\varsx_1\} \varsx_3\{\varsx_2\varsx_2\} \varsx_4\{\varsx_3\varsx_3\} \ldots \varsx_n\{\varsx_{n-1}\varsx_{n-1}\}
\end{equation*}
is a conjunctive xregex that is variable-simple and every variable has at most one definition. However, the procedure of Lemma~\ref{removeNonBasicLemma} will apply the main modification step with respect to all variables $\varsx_2, \varsx_3, \ldots, \varsx_n$ in this order (note that $G_{\bar{\alpha}}$ is a path), which replaces each references of $x_2$ by $2$ variables, each references of $x_4$ by $4$ variables, each references of $x_5$ by $8$ variables, and so on. 
%
For example, the first application of the modification step will produce 
\begin{equation*}
\varsx_1\{\ta\} \varsu_1\{\varsx_1\}\varsu_2\{\varsx_1\} \varsx_3\{(\varsu_1\varsu_2)^2\} \varsx_4\{\varsx_3\varsx_3\} \ldots \varsx_n\{\varsx_{n-1}\varsx_{n-1}\}
\end{equation*}
The crucial point seems to be non-basic definitions for variables with references in other non-basic definitions. If we restrict vstar-free conjunctive xregex accordingly, then the exponential size blow-up does not occur in Step $3$ of the normal form construction.\par
A variable $\varsx \in \varset$ is \emph{flat} (in some $\bar{\alpha} \in \conxregex_{\Sigma, \varset}$) if its definition is basic, or it has no reference in any other definition. For example, let $\alpha_1 = \varsu \tb^* \varsx\{ \varsy\{\ta^*\} (\ta \altop \tb)^* \varsz \varsy\}$ and $\alpha_2 = \varsu\{\tc \tb \varsz\{\ta^* (\tb \altop \tc \ta)\}\} \ta \varsx$, then in $\bar{\alpha} = (\alpha_1, \alpha_2)$ every variable is flat. Finally, let $\vsfpbCXRPQ$ be the class of vstar-free $\CXRPQ$ with only flat variables.

We next show that if a conjunctive xregex satisfies the preconditions of Lemma~\ref{removeNonBasicLemma} \emph{and} only has flat variables, then Step $3$ of the normal form construction does not cause an exponential size blow-up.

\begin{lemma}
Let $\bar{\alpha}\in \dimconxregex{m}_{\Sigma, \varset}$ such that, for every $i \in [m]$, $\bar{\alpha}[i]$ is an alternation of variable-simple xregex, every $\varsx \in \varset$ has at most one definition in $\bar{\alpha}$, and all variables are flat. Then there is an equivalent $\bar{\beta} \in \dimconxregex{m}_{\Sigma, \varset'}$ in normal form with $|\bar{\beta}| = \bigO(|\bar{\alpha}|^{2})$. 
\end{lemma}

\begin{proof}
Since the conjunctive xregex $\bar{\alpha}$ satisfies the conditions of the statement of Lemma~\ref{removeNonBasicLemma}, we can transform $\bar{\alpha}$ to $\bar{\beta}$ in exactly the same way as in the proof of Lemma~\ref{removeNonBasicLemma}. In order to conclude the proof, we only have to show that if all variables of $\bar{\alpha}$ are flat, then $|\bar{\beta}| = \bigO(|\bar{\alpha}|^{2})$. \par
Let $\varsx \in \varset$. If the definition of $\varsx$ is basic, then it is not changed by the procedure that changes $\bar{\alpha}$ to $\bar{\beta}$. If the definition of $\varsx$ is not basic, then, by assumption, it has no reference in any other variable definition. Thus, when the main modification step is applied to $\varsx$, then the size of all other variable definitions does not increase. Consequently, in the procedure that changes $\bar{\alpha}$ to $\bar{\beta}$, every reference of $\varsx$ is not changed if the definition of $\varsx$ is basic, and it is replaced by at most $|\bar{\alpha}|$ other variable references, if it is not basic. This implies that, in the worst case, every symbol of $\bar{\alpha}$ can only be replaced by at most $|\bar{\alpha}|$ symbols, which means that $|\bar{\beta}| = \bigO(|\bar{\alpha}|^2)$.
\end{proof}

With this lemma, and the observation that Steps~$1$~and~$2$ of the normal form construction preserve flat variables, we can now show that the space upper bound for Boolean evaluation of vstar-free conjunctive xregex path queries is polynomial.

\begin{lemma}\label{pseudoBasicLemma}
Given a Boolean $q \in \vsfCXRPQ$ with only flat variables and a graph database $\DBD$, we can nondeterministically check whether $\DBD \models q$ in space $\bigO(\poly(|q|) \log(|\DBD|))$.
\end{lemma}

\begin{proof}
We can proceed analogously to the proof of Lemma~\ref{vsfCXRPQEvalUpperBoundTheoremTwo}. We only have to note that the initial nondeterministic transformation as well as the procedure of Lemma~\ref{uniqueDefinitionsLemma} preserves the property of having only flat variables. Consequently, the application of the procedure Lemma~\ref{removeNonBasicLemma} yields a $\CXRPQ$ in normal form that is of size polynomial in the initial query. 
\end{proof}

Finally, let us summarise all these observations more concisely as follows.

\begin{theorem}\label{pseudoBasicTheorem}
$\vsfpbCXRPQ$-$\booleProb$ is $\pspaceclass$-complete in combined complexity.
\end{theorem}

\section{$\CXRPQ$ with Bounded Image Size}\label{sec:boundedImageSize}

The fragment to be defined in this section directly follows from the definition of the subsets $\langimagebound{k}(\alpha) \subseteq \lang(\alpha)$ that contain the words that match xregex $\alpha$ with a witness $v$ such that $|\varmap{v}(\varsx)| \leq k$ for every $\varsx \in \varset$. This can be easily done as follows.  \par
Let 
$\alpha \in \xregex_{\Sigma, \varset}$ with $\varset = \{\varsx_1, \varsx_2, \ldots, \varsx_n\}$ and $\bar{v} \in (\Sigma^*)^n$. 
We define $\reflangvarmap{\bar{v}}(\alpha) = \{u \in \reflang(\alpha) \mid \varmap{u, \varset} = \bar{v}\}$ and
$\langvarmap{\bar{v}}(\alpha) = \deref(\reflangvarmap{\bar{v}}(\alpha))$.
For every $k \geq 1$, let $\reflangimagebound{k}(\alpha) = \bigcup_{\bar{v} \in (\Sigma^{\leq k})^n} \reflangvarmap{\bar{v}}(\alpha)$.
Finally, we define $\langimagebound{k}(\alpha) = \deref(\reflangimagebound{k}(\alpha))$.\par
The notions $\reflangvarmap{\bar{v}}$ and $\reflangimagebound{k}$ also extend to conjunctive xregex in the obvious way, and therefore to $\CXRPQ$ as follows. For a $q \in \CXRPQ$ with conjunctive xregex $\bar{\alpha} \in \conxregex_{\Sigma, \varset}$, a $\bar{v} \in (\Sigma^*)^n$, and a graph database $\DBD$, by $q^{\bar{v}}(\DBD)$ we denote the subset of $q(\DBD)$ that contains those $q_h(\DBD) \in q(\DBD)$, where $h$ is a matching morphism with respect to some matching words $\bar{w} = (w_1, w_2, \ldots, w_m) \in \langvarmap{\bar{v}}(\bar{\alpha})$; $q^{\leq k}(\DBD)$ is defined analogously by restricting the matching words to be from $\langimagebound{k}(\bar{\alpha})$. In the Boolean case, we also set $\DBD \models^{\bar{v}} q$ and $\DBD \models^{\leq k} q$ to denote that $q^{\bar{v}}(\DBD)$ or $q^{\leq k}(\DBD)$, respectively, contains the empty tuple.

The above defined restrictions do not restrict the class $\CXRPQ$ (as it is the case for $\vsfCXRPQ$ considered in Section~\ref{sec:vsfCXRPQ}), but rather how the results of queries from $\CXRPQ$ are defined. However, it shall be convenient to allow a slight abuse of notation and define, for every $k \in \mathbb{N}$, the class $\bisCXRPQ{k}$, which is the same as $\CXRPQ$, but for every $q \in \bisCXRPQ{k}$ it is understood that $q(\DBD)$ actually means $q^{\leq k}(\DBD)$ (and $\DBD \models q$ actually means $\DBD \models^{\leq k} q$).
In this way, for every $k \in \mathbb{N}$, also the problems $\bisCXRPQ{k}$-$\booleProb$ and $\bisCXRPQ{k}$-$\checkProb$ are defined.  \par
For the classes $\bisCXRPQ{k}$, we can show the following upper complexity bounds (which shall be proven in Subsection~\ref{sec:imageSizeUpperBound}). 

\begin{theorem}\label{upperBoundsBISCXRPQTheorem}
For every $k \in \mathbb{N}$, $\bisCXRPQ{k}$-$\booleProb$ is in $\npclass$ with respect to combined-complexity and in $\nlclass$ with respect to data-complexity.
\end{theorem}

Since matching lower bounds
directly carry over from $\CRPQ$s, the evaluation complexity of $\bisCXRPQ{k}$ and $\CRPQ$ seems to be the same. However, we can state a much stronger hardness result, which points out an important difference between $\bisCXRPQ{k}$ and $\CRPQ$ in terms of evaluation complexity: for $\bisCXRPQ{1}$, we have $\npclass$-hardness in combined-complexity even for single-edge graph patterns (with even simple xregex). This is not the case for $\CRPQ$, which can be evaluated in polynomial-time if the structure of the underlying graph pattern is acyclic (see~\cite{BarceloEtAl2016, BarceloEtAl2012, Barcelo2013}).

\begin{theorem}\label{HittingSetReductionTheorem}
$\bisCXRPQ{k}$-$\booleProb$ is 
\begin{itemize}
\item $\npclass$-hard in combined-complexity, even for $k = 1$ and single-edge queries with simple xregex $\alpha \in \xregex_{\Sigma, \varset}$ and $|\Sigma| = 3$, 
\item $\nlclass$-hard in data-complexity, even for $k = 0$ and for single-edge queries with xregex $\alpha$, where $\alpha \in \RE_{\Sigma}$ with $|\Sigma| = 2$.
\end{itemize}
\end{theorem}

\begin{figure}
\begin{center}
\scalebox{1.5}{\includegraphics{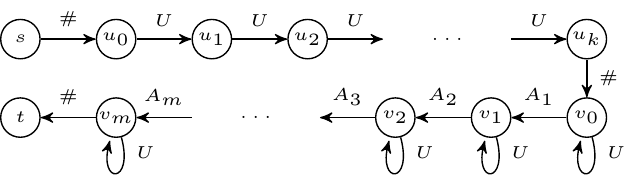}}
\end{center}
\caption{Sketch of the \textsc{Hitting Set} Reduction from Theorem~\ref{HittingSetReductionTheorem}. An arc labelled with $U$ or $A_i$ stands for arcs labelled by $\langle z \rangle$ for every $z \in U$ or $z \in A_i$, respectively.
}
\label{HSReductionAppendix}
\end{figure}

\begin{proof}
Since $\CRPQ \subseteq \bisCXRPQ{1}$, the $\nlclass$-hardness for data-complexity follows in the same way as in the proof for Theorem~\ref{vsfCXRPQEvalLowerBoundTheorem}.\par
In order to prove the $\npclass$-hardness for combined-complexity, we devise a reduction from the problem \textsc{Hitting Set}, which is defined as follows. Given subsets $A_1, A_2, \ldots, A_m$ of some universe $U$
and $k \in \mathbb{N}$, decide whether there is a \emph{hitting set} of size at most $k$, i.\,e., a set $B \subseteq U$ with $|B| \leq k$ and $B \cap A_i \neq \emptyset$ for every $i \in [m]$.\par
Now let $A_1, A_2, \ldots, A_m \subseteq U = \{z_1, z_2, \ldots, z_n\}$ and $k \in \mathbb{N}$ be an instance of \textsc{Hitting Set}. 
Let $\Sigma = \{\ta, \tb, \#\}$ and, for every $z_i \in U$, we define $\langle z_i \rangle = \tb \ta^i \tb$.\par
Next, we define a graph database $\DBD = (V_{\DBD}, E_{\DBD})$ over $\Sigma$ as follows. For the sake of convenience, we also allow words as edge labels in graph databases, which stand for paths in the obvious way, i.\,e., for some $w \in \Sigma^*$, an arc $(u, w, v)$ in a graph-database represents a path from $u$ to $v$ labelled with $w$ (using some distinct intermediate nodes). 
We set 
\begin{equation*}
V_{\DBD} = \{s, u_0, u_1, \ldots, u_k, v_0, v_1, \ldots, v_m, t\}
\end{equation*} 
and the set $E_{\DBD}$ is defined as follows. 
\begin{itemize}
\item There are arcs $(s, \#, u_0)$, $(u_k, \#, v_0)$ and $(v_m, \#, t)$,
\item For every $i \in [k]$ and $z \in U$ there is an arc $(u_{i-1}, \langle z \rangle, u_{i})$, 
\item for every $i \in [m]$ and $z \in A_i$ there is an arc $(v_{i-1}, \langle z \rangle, v_{i})$, 
\item for every $i \in [m] \cup \{0\}$ and every $z \in U$ there is an arc $(v_i, \langle z \rangle, v_i)$.
\end{itemize}
See Figure~\ref{HSReductionAppendix} for an illustration of $\DBD$. Next, we define the xregex 
\begin{equation*}
\alpha = \# \, \prod^{(n+2)k}_{i = 1} \varsx_i \{\ta \altop \tb \altop \eword\} \, \# \, \left(\prod^{(n+2)k}_{i = 1} \varsx_i\right)^m \, \#\,,
\end{equation*}
and the Boolean $q \in \CXRPQ$ is defined by the single edge graph pattern $\{\{x, y\}, (x, \alpha, y)\}$. It is crucial to note that, for every $k \geq 1$, $\langimagebound{k}(\alpha) = \lang(\alpha)$, which is due to the fact that for every possible ref-word in $\reflang(\alpha)$, each variable image is either $\ta$, $\tb$ or $\eword$. Consequently, it does not matter for what $k$ we interpret $q$ to be a $\bisCXRPQ{k}$.\par
In order to prove the correctness, we first make some observations:
\begin{enumerate}
\item\label{ObsPointOne} Due to the occurrences of $\#$, there is a path in $\DBD$ from a node $y$ to a node $z$ labelled with a word from $\lang(\alpha)$ if and only if there is such a path in $\DBD$ from $s$ to $t$.
\item\label{ObsPointTwo} The language $\lang(\alpha)$ contains exactly the words $w = \# w_1 \# w_2 \#$, where $w_1 \in \{\ta, \tb\}^*$ with $|w_1| \leq (n+2)k$, and $w_2 = (w_1)^m$.
\item\label{ObsPointThree} Every path in $\DBD$ from $s$ to $t$ is labelled by $w = \# w_1 \# w_2 \#$, where $w_1 = \langle z_{j_1} \rangle \langle z_{j_2} \rangle \ldots \langle z_{j_k} \rangle$, for some $\{j_1, j_2, \ldots, j_k\} \subseteq [n]$, and $$w_2 = u_1 \langle z_{r_1} \rangle u_2 \langle z_{r_2} \rangle u_3 \ldots u_m \langle z_{r_m} \rangle u_{m+1}\,,$$ such that $\{r_1, r_2, \ldots, r_m\} \subseteq [n]$ and, for every $i \in [m]$, $z_{r_i} \in A_i$; in particular, this means that $\{z_{r_1}, z_{r_2}, \ldots, z_{r_m}\}$ is a hitting set (with respect to the considered problem-instance). 
\end{enumerate}
Next, we assume that in $\DBD$ there is a path $p$ labelled with a word $w \in \lang(\alpha)$. With Point~\ref{ObsPointOne}, we conclude that $p$ is a path from $s$ to $t$. Moreover, with Point~\ref{ObsPointThree}, we have that $w = \# w_1 \# w_2 \#$, where $w_1 = \langle z_{j_1} \rangle \langle z_{j_2} \rangle \ldots \langle z_{j_k} \rangle$
and $w_2 = u_1 \langle z_{r_1} \rangle u_2 \langle z_{r_2} \rangle u_3 \ldots u_m \langle z_{r_m} \rangle u_{m+1}$, such that $\{z_{r_1}, z_{r_2}, \ldots, z_{r_m}\}$ is a hitting set. Finally, Point~\ref{ObsPointTwo} means that $w_2 = (w_1)^m$, which directly implies that $\{z_{r_1}, z_{r_2}, \ldots, z_{r_m}\} \subseteq \{z_{j_1}, z_{j_2}, \ldots, z_{j_k}\}$ and therefore $|\{z_{r_1}, z_{r_2}, \ldots, z_{r_m}\}| \leq k$. \par
On the other hand, if $\{z_{j_1}, z_{j_2}, \ldots, z_{j_k}\}$ is a hitting set of size $k$, then we can construct a word $w \in \lang(\alpha)$ and a path from $s$ to $t$ that is labelled with $w$ as follows. We set $w = \# w_1 \# w_2 \#$ with $w_1 = \langle z_{j_1} \rangle \langle z_{j_2} \rangle \ldots \langle z_{j_k} \rangle$ and $w_2 = (w_1)^m$. We note that, according to Point~\ref{ObsPointTwo}, $w \in \lang(\alpha)$ and we have to show that there is a path from $s$ to $t$ labelled with $w$. There is obviously a path from $s$ to $v_0$ labelled with $\# w_1 \#$. The set $\{z_{j_1}, z_{j_2}, \ldots, z_{j_k}\}$ is a hitting set and each of the $m$ occurrences of factor $w_1$ in $w_2$ contains a factor $\langle z_{j_i} \rangle$ for every $i \in [k]$. Thus, $w_2$ can be factorised into $w_2 = u_1 \langle z_{r_1} \rangle u_2 \langle z_{r_2} \rangle u_3 \ldots u_m \langle z_{r_m} \rangle u_{m+1}$, such that, for every $i \in [m]$, $\langle z_{r_i} \rangle \in A_i$. This means that there is a path from $v_0$ to $t$ labelled with $w_2 \#$ and therefore a path from $s$ to $t$ labelled with $w$.
\end{proof}

\subsection{Proof of Theorem~\ref{upperBoundsBISCXRPQTheorem}}\label{sec:imageSizeUpperBound}

The main building block of our algorithm is that if we have fixed some variable mapping $\bar{v} = (v_1, v_2, \ldots, v_m)$, then the subset $\langvarmap{\bar{v}}(\bar{\alpha})$ of $\lang(\bar{\alpha})$ can be represented by a conjunctive xregex without variable definitions, i.\,e., a tuple of classical regular expressions (see Lemma~\ref{checkConXregexVariableMappingLemma}). However, the corresponding procedure is not as simple as ``replace each $\varsx_i\{\ldots\}$ and $\varsx_i$ by $v_i$''. We shall now illustrate this with an example and some intuitive explanations. Let $\alpha = (\alpha_1, \alpha_2) \in \conxregex_{\Sigma, \varset}$ be defined by
\begin{align*}
\alpha_1 =\:&\varsx_3\{\varsx_1\{\tc \ta^* \tc\} \varsx_2^*\} \altop [\,(\varsx_1\{\tc \tb^*\} \altop \varsx_1\{\varsx_4 \tc^*\}) (\tb \altop \varsx_2^*) \varsx_3\{\varsx_1 \varsx_2 \varsx_1^*\}\,]\,,\\
\alpha_2 =\:&(\varsx_1 \altop \varsx_2)^* \varsx_4\{(\tb \altop \tc)^* \varsx_2^*\} \varsx_2\{(\ta \altop \tb)^* \ta\}\,,
\end{align*}
and let $\bar{v} = (v_1, \ldots, v_4) = (\tc \ta, \ta, \tc \ta \ta \tc \ta, \tc \ta)$. \par 
For computing $\beta = (\beta_1, \beta_2) \in \conxregex_{\Sigma, \emptyset}$ with $\lang(\bar{\beta}) = \langvarmap{\bar{v}}(\bar{\alpha})$, replacing $\varsx_i\{\gamma\}$ by $v_i$ can only be correct, if $\gamma$ can produce $v_i$ (in a match with variable mapping $\bar{v}$). So we should treat the variable definitions and references of $\gamma$ as their intended images and then check whether the thus obtained classical regular expression can produce $v_i$. For example, we transform $\varsx_3\{\varsx_1\{\tc \ta^* \tc\} \varsx_2^*\}$ in $\alpha_1$ to $\varsx_3\{v_1 v_2^*\} = \varsx_3\{\tc \ta \ta^*\}$ and check $v_3 = \tc \ta \in \lang(\tc \ta \ta^*)$. However, this assumes that $\varsx_1\{\tc \ta^* \tc\}$ can really produce $v_1$, which is not the case, so we should rather remove $\varsx_3\{\varsx_1\{\tc \ta^* \tc\} \varsx_2^*\}$ altogether, since it can never be involved in a conjunctive match from $\langvarmap{\bar{v}}(\bar{\alpha})$. Hence, we also need to cut whole alternation branches in $\bar{\alpha}$, and, moreover, we have to make sure that if $\varsx_i \neq \emptyword$, then at least one definition of $\varsx_i$ is necessarily instantiated, while this is not required if $\varsx_i = \emptyword$. Let us illustrate the correct transformations with this example.\par
For every variable definition $\varsx_i\{\gamma\}$, where $\gamma$ is a classical regular expression, we check whether $v_i \in \lang(\gamma)$, and we mark the definition accordingly with $\tone$ or $\tzero$, respectively. Then, we cut all branches that necesssarily produce a definition marked with $\tzero$. This transform $\alpha_1$ as follows
\begin{align*}
\alpha_1 &\leadsto \varsx_3\{\varsx_1\{\tzero\} \varsx_2^*\} \altop [\,(\varsx_1\{\tzero\} \altop \varsx_1\{\tone\}) (\tb \altop \varsx_2^*) \varsx_3\{\tone\}\,] \\
&\leadsto \varsx_1\{\tone\}(\tb \altop \varsx_2^*) \varsx_3\{\tone\}\,.
\end{align*}
If there were variable definitions left, we would repeat this step, but for checking whether $v_i$ can be generated by $\gamma$, we would treat variable definitions marked with $\tone$ as their intended images. With respect to $\alpha_2$, we get 
\begin{align*}
(\varsx_1 \altop \varsx_2)^* \varsx_4\{(\tb \altop \tc)^* \varsx_2^*\} \varsx_2\{(\ta \altop \tb)^* \ta\} \: \leadsto \: (\varsx_1 \altop \varsx_2)^* \varsx_4\{\tone\} \varsx_2\{\tone\}\,.
\end{align*}
Note that for $\varsx_4\{(\tb \altop \tc)^* \varsx_2^*\}$, we check whether $(\tb \altop \tc)^* (v_2)^* = (\tb \altop \tc)^* \ta^*$ can produce $v_4 = \tc \ta$, which is the case.\par
Finally, we replace all definitions and references by the intended images to obtain $(\beta_1, \beta_2) = (\tc \ta(\tb \altop \ta^*) \tc \ta \ta \tc \ta, ((\tc \ta) \altop \ta)^* \tc \ta \ta)$. \par
In the general case, the situation can be slightly more complicated, since we also have to make sure that every ref-word necessarily instantiates a definition for $\varsx_i$ if $v_i \neq \emptyset$, while for $v_i = \emptyset$ ref-words without definition for $\varsx_i$ should still be possible. It can also happen that this procedure reduces an xregex to $\varnothing$, which means that $\langvarmap{\bar{v}}(\bar{\alpha}) = \emptyset$. \par
%
%
%
%
These exemplary observations can be turned into a general procedure, which yields the following results.

\begin{lemma}\label{checkConXregexVariableMappingLemma}
For every $\bar{\alpha} \in \dimconxregex{m}_{\Sigma, \varset}$ with $\varset = \{\varsx_1$, $\varsx_2$, $\ldots, \varsx_n\}$ and every $\bar{v} \in (\Sigma^*)^n$, there is a $\bar{\beta} \in \dimconxregex{m}_{\Sigma, \emptyset}$ such that $\lang(\bar{\beta}) = \langvarmap{\bar{v}}(\bar{\alpha})$. Furthermore, $|\bar{\beta}| = \bigO(|\bar{\alpha}|k)$, where $k = \max\{|\bar{v}[i]| \mid i \in [n]\}$, and $\bar{\beta}$ can be computed in time polynomial in $|\bar{\alpha}|$ and $|\bar{v}|$.
\end{lemma}

\begin{proof}
Let $\bar{\alpha} = (\alpha_1, \alpha_2, \ldots, \alpha_m)$ and let $\bar{v} = (v_1, v_2, \ldots, v_n) \in \Sigma^*$. We give an algorithm that transforms $\bar{\alpha}$ into $\bar{\beta}$ with the desired property. \par
The general idea of the procedure is as follows. For every considered variable definition $\varsx_i\{\gamma\}$ in some $\alpha_j$, we want to determine whether $\gamma$ can produce $v_i$ under the assumption that all variable references and variable definitions in $\gamma$ can produce their corresponding images (according to $\bar{v}$). If this is not the case, then we delete the alternation branch that instantiates the variable definition $\varsx_i\{\gamma\}$ to make sure that it cannot be instantiated. When this procedure terminates, we can replace all remaining variable definitions and variable references by their corresponding images in order to obtain the components $\beta_i$ of $\bar{\beta}$, which are then classical regular expressions. We now describe this procedure in more detail.\medskip \par
\noindent \textbf{Step 1:} We assume that the variable definitions in $\bar{\alpha}$ can be either \emph{marked} or \emph{unmarked} and that they are initially all \emph{unmarked}. The algorithm considers each variable definition separately in such an order that every considered variable definition does only contain other variable definitions (if any) that are already marked. We repeat the following step until all variable definitions are \emph{marked} (recall that initially all variable definitions are \emph{unmarked}). Let $\varsx_i\{\gamma\}$ be some \emph{unmarked} variable definition in some $\alpha_j$ such that $\gamma$ does not contain any \emph{unmarked} variable definition. Let $\gamma'$ be the classical regular expression obtained from $\gamma$ by replacing each reference and definition for a variable $\varsx_{i'}$ by $v_{i'}$.
If $v_i \in \lang(\gamma')$, then we \emph{mark} $\varsx_i\{\gamma\}$. If $v_i \notin \lang(\gamma')$, then we have to modify $\alpha_j$ in such a way that $\varsx_i\{\gamma\}$ is never instantiated. This is achieved as follows. We start at the node of the syntax-tree of $\alpha_j$ that represents $\varsx_i\{\gamma\}$, we move up in the syntax tree and simply delete every node that we encounter (including the node that represents $\varsx_i\{\gamma\}$ where we started, which also means that the whole subtree rooted by this node is deleted), and we stop as soon as we encounter an alternation node, which is then replaced by its other child (i.\,e., the sibling of the node from which we entered the alternation node). In particular, if no alternation node is encountered, then we can conclude that the definition $\varsx_i\{\gamma\}$ under consideration will necessarily be instantiated by every ref-word of $\alpha_j$, which immediately implies that there is no conjunctive match of $\bar{\alpha}$ with variable mapping $(v_1, v_2, \ldots, v_n)$. In this case, the procedure will replace $\alpha_j$ by $\varnothing$ (the unique regular expressions with $\lang(\varnothing) = \emptyset$). 
Let $\bar{\alpha}' = (\alpha'_1, \alpha'_2, \ldots, \alpha'_m)$ be the conjunctive xregex obtained when the procedure of this step terminates.\medskip \par
\noindent \textbf{Step 2:} 
Next, we have to modify $\bar{\alpha}'$ with respect to every $i \in [n]$ with $v_i \neq \emptyword$ as follows. If, for some $j \in [m]$, $\alpha'_j$ contains a definition of $\varsx_i$ (note that this is possible for at most one $\alpha'_j$), then we have to modify it such that it necessarily instantiates a definition for $\varsx_i$, which is done as follows. For every node of the syntax tree for $\alpha'_j$ that corresponds to a definition for $\varsx_i$, we mark it 
and then we move from this node up to the root and mark every visited node along the way. Next, for every marked alternation-node, we remove its unmarked child nodes (note that every such marked alternation-node has either one or two marked child nodes). \par
In particular, we note that this modification is only necessary for $i \in [n]$ with $v_i \neq \emptyset$, since ref-words that have no definition for $\varsx_i$ correspond to variable mappings with image $\eword$ for variable $\varsx_i$, which should not be excluded if $v_i = \eword$. Let $\bar{\alpha}'' = (\alpha''_1, \alpha''_2, \ldots, \alpha''_m)$ be the conjunctive xregex obtained when the procedure of this step terminates.\medskip\par
\noindent After these two modification steps, it is possible that, for some $i \in [m]$ with $v_i \neq \emptyset$, there is no definition of $\varsx_i$ in $\bar{\alpha}''$. If this is the case, we replace each $\alpha''_{j}$ by $\varnothing$.\par
Finally, for every $i \in [n]$, we replace each definition and each occurrence of $\varsx_i$ by $v_i$ in order to obtain a $\bar{\beta} \in \dimconxregex{m}_{\Sigma, \emptyset}$. It can be verified with moderate effort that $\lang(\bar{\beta}) = \langvarmap{\bar{v}}(\bar{\alpha})$.
Moreover, this procedure can obviously be carried out in time polynomial in $|\bar{\alpha}|$ and $|\bar{v}|$, and also $|\bar{\beta}| = \bigO(|\bar{\alpha}|k)$, where $k = \max\{|\bar{v}[i]| \mid i \in [n]\}$.
\end{proof}

The statement of Lemma~\ref{checkConXregexVariableMappingLemma} directly carries over from conjunctive xregex to $\CXRPQ$ as follows.

\begin{lemma}\label{checkCXRPQVariableMappingLemma}
For every $q \in \CXRPQ$ with conjunctive xregex $\bar{\alpha} \in \dimconxregex{m}_{\Sigma, \varset}$ with $\varset = \{\varsx_1, \varsx_2, \ldots, \varsx_n\}$ and every $\bar{v} \in (\Sigma^*)^n$, there is a $q' \in \CRPQ$, such that, for every database $\DBD$, we have that $q^{\bar{v}}(\DBD) = q'(\DBD)$. Furthermore, $|q'| = \bigO(|q|k)$, where $k = \max\{|\bar{v}[i]| \mid i \in [n]\}$, and $q'$ can be computed in time polynomial in $|q|$ and $|\bar{v}|$.
\end{lemma}

\begin{proof}
According to Lemma~\ref{checkConXregexVariableMappingLemma}, we can compute in time polynomial in $|\bar{\alpha}|$ and $|\bar{v}|$ (and therefore in time polynomial in $|q|$ and $|\bar{v}|$) a $\bar{\beta} \in \dimconxregex{m}_{\Sigma, \emptyset}$ such that $\lang(\bar{\beta}) = \langvarmap{\bar{v}}(\bar{\alpha})$.
Thus, $q'$ can be obtained from $q$ by replacing each edge label $\alpha_i$ by $\beta_i$. In particular, we have $|\bar{\beta}| = \bigO(|\bar{\alpha}|k)$, where $k = \max\{|\bar{v}[i]| \mid i \in [n]\}$, and therefore also $|q'| = \bigO(|q|k)$.
\end{proof}

We are now ready to give a formal proof for Theorem~\ref{upperBoundsBISCXRPQTheorem}.

\begin{proof}(\emph{of Theorem~\ref{upperBoundsBISCXRPQTheorem}})
Let $q \in \bisCXRPQ{k}$ be Boolean with a conjunctive xregex $\bar{\alpha} = (\alpha_1, \alpha_2, \ldots, \alpha_m)$ over $\Sigma$ and $\varset = \{\varsx_1, \varsx_2, \ldots, \varsx_n\}$, and let $\DBD$ be a graph-database. We note that $\DBD \models q$ if and only if there is a $\bar{v} = (v_1, v_2, \ldots, v_n) \in (\Sigma^{\leq k})^n$ such that $\DBD \models^{\bar{v}} q$. Consequently, the following is a nondeterministic algorithm that checks whether $\DBD \models q$.\par
\begin{enumerate}
\item\label{NPAlgoPOne} Nondeterministically guess some $\bar{v} = (v_1, v_2, \ldots, v_n) \in (\Sigma^{\leq k})^n$.
\item\label{NPAlgoPTwo} Compute $q' \in \CRPQ$ such that, for every database $\DBD$, we have that $q^{\bar{v}}(\DBD) = q'(\DBD)$. 
\item\label{NPAlgoPThree} Check whether $\DBD \models q'$.
\end{enumerate}
Point~\ref{NPAlgoPTwo} can be done according to Lemma~\ref{checkCXRPQVariableMappingLemma}, and Point~\ref{NPAlgoPThree} can be done according to Lemma~\ref{evalCRPQLemma}. It remains to show that this nondeterministic algorithm requires polynomial time in combined-complexity and logarithmic space in data-complexity.\par
We first note that Point~\ref{NPAlgoPOne} can be done in time $\bigO(nk)$, which is polynomial in combined complexity. Moreover, the required space for this does only depend on $q$ and $k$, which means that it is constant in data-complexity. \par
According to Lemma~\ref{checkCXRPQVariableMappingLemma}, $q'$ can be computed in time that is polynomial in $|q|$ and $|\bar{v}|$ in combined-complexity, which, since $k$ is a constant, is polynomial in $|q|$. Again, the required space for this does only depend on $q$ and $k$, which means that it is constant in data-complexity.\par
According to Lemma~\ref{evalCRPQLemma}, we can nondeterministically check $\DBD \models q'$ in time polynomial in $|q'|$ and $|\DBD|$ in combined complexity. Moreover, according to Lemma~\ref{checkCXRPQVariableMappingLemma}, $|q'| = \bigO(|q|k) = \bigO(|q|)$, which means that we can nondeterministically check $\DBD \models q'$ in time polynomial in $|q|$ and $|\DBD|$, so polynomial in data-complexity. Finally, we also note that Lemma~\ref{evalCRPQLemma} implies that $\DBD \models q'$ can be checked in nondeterministic space that is logarithmic in $\DBD$ with respect to data-complexity. 
\end{proof}

\subsection{Logarithmically Bounded Image Size}

Analogously to $q^k(\DBD)$ for every $k \geq 1$, $q \in \CXRPQ$ and graph database $\DBD$, we can also define $q^{\log}(\DBD) = \bigcup^{\log(|\DBD|)}_{k = 0} q^k(\DBD)$. Just like we derived the classes $\bisCXRPQ{k}$, this gives rise to the class $\logCXRPQ$. Note that a Boolean $q \in \logCXRPQ$ matches a graph database $\DBD$, if and only if there is a matching morphism with image size bounded by $\log(|\DBD|)$. \par
Analogously to the proof of Theorem~\ref{upperBoundsBISCXRPQTheorem}, we can show the following upper bound for $\logCXRPQ$.

\begin{corollary}\label{logBoundCorollary}
$\logCXRPQ$-$\booleProb$ can nondeterministically be solved in polynomial time in combined-complexity and in space $\bigO(\log^2(|\DBD|))$ in data-com\-plexity.
\end{corollary}

\begin{proof}
We can apply the same nondeterministic algorithm from the proof of Theorem~\ref{upperBoundsBISCXRPQTheorem}, with the only difference that we initially guess $\bar{v} = (v_1, v_2, \ldots, v_n) \in (\Sigma^{\leq \log(|\DBD|)})^{n}$. More precisely, we use the following nondeterministic algorithm:
\begin{enumerate}
\item\label{NPAlgoPOnev2} Nondeterministically guess some $\bar{v} = (v_1, v_2, \ldots, v_n) \in (\Sigma^{\leq \log(|\DBD|)})^{n}$.
\item\label{NPAlgoPTwov2} Compute $q' \in \CRPQ$ such that, for every database $\DBD$, we have that $q^{\bar{v}}(\DBD) = q'(\DBD)$. 
\item\label{NPAlgoPThreev2} Check whether $\DBD \models q'$.
\end{enumerate}
Point~\ref{NPAlgoPOnev2} can be done in time and space $\bigO(n\log(|\DBD|))$. According to Lemma~\ref{checkCXRPQVariableMappingLemma}, $q'$ can be computed in time that is polynomial in $|q|$ and $|\bar{v}|$, which, since $|\bar{v}| = n \log(|\DBD|)$, is polynomial in $|q|$ and $|\DBD|$. Moreover, according to Lemma~\ref{checkCXRPQVariableMappingLemma}, $|q'| = \bigO(|q|\log(|\DBD|))$.\par
We can use Lemma~\ref{evalCRPQLemma} in order to conclude that we can nondeterministically check $\DBD \models q'$ in polynomial time. With respect to space complexity, we note that since $q' \in \CRPQ$, we also have that $q'$ is a simple $\CXRPQ$. Therefore, we can conclude with Lemma~\ref{simpleCXRPQLemma} that we can check $\DBD \models q'$ in space 
\begin{align*}
&\bigO(|q'| \log(|\DBD|) + |q'|\log(|q'|))\\
= &\bigO(|q| \log(|\DBD|) \log(|\DBD|) + |q| \log(|\DBD|)\log(|q| \log(|\DBD|)))\\
= &\bigO(|q| \log(|q|) \log^2(|\DBD|))\,.
\end{align*}
%
%
%
%
%
%
\end{proof}

\section{Expressive Power}\label{sec:expressivePower}

In this section, we compare the expressive power of $\CXRPQ$s and their fragments with other established classes of graph queries.\par
The \emph{extended conjunctive regular path queries} ($\ECRPQ$s), already mentioned in the introduction, have been introduced in~\cite{BarceloEtAl2012}. We shall define them now in more detail (a definition in full details can be found in~\cite{BarceloEtAl2012}). Extended conjunctive regular path queries ($\ECRPQ$s) have the form $q = \bar{z} \gets G_q, \bigwedge_{j \in [t]} R_j(\bar{\omega}_j)$, where $\bar{z} \gets G_q$ is a $\CRPQ$ and, for every $j \in [t]$, $\bar{\omega}_j$ is an $s_j$-tuple over $E_q$ and $R_j$ is a regular expression that describes a regular relations over $\Sigma^*$ of arity $s_j$. The semantics of $\ECRPQ$ can be derived from the semantics of $\CRPQ$ as follows. We interpret $q$ as a $\CRPQ$, but we add to the concept of a matching morphism the requirement that there must be a tuple $(w_{e_1}, w_{e_2}, \ldots, w_{e_m})$ of matching words such that, for each $\bar{\omega}_j = (e_{p_1}, e_{p_2}, \ldots, e_{p_{s_j}})$, we have $(w_{e_{p_1}}, w_{e_{p_2}}, \ldots, w_{e_{p_{s_j}}}) \in \lang(R_j)$.\par
By $\ECRPQ$ with \emph{equality relations} ($\eqECRPQ$ for short), we denote the class of $\ECRPQ$ for which each $R_j$ is the \emph{equality relation} (for some arity $s_j$), i.\,e., the relation $\{(u_1, u_2, \ldots, u_{s_j}) \in (\Sigma^*)^{s_j} \mid u_1 = u_2 = \ldots = u_{s_j}\}$. In order to ease our notations, we sometimes represent the equality relations as a partition $\{E_{q, 1}, E_{q, 2}, \ldots, E_{q, t}\}$ of $E_q$ (i.\,e., for every $j \in [t]$, the edges of $E_{q, j}$ are subject to an equality relation), or, for simple queries, we also state which edges are required to be equal without formally stating the equality relations.\par
We recall that for some conjunctive path query $q$, the mapping $\DBD \mapsto q(\DBD)$ from the set of graph-databases to the set of relations over $\Sigma$ that is defined by $q$ is denoted by $\llbracket q \rrbracket$, and for any class $A$ of conjunctive path queries, $\llbracket A \rrbracket = \{\llbracket q \rrbracket \mid q \in A\}$. \par
For a class $Q$ of conjunctive path queries, a \emph{union of $Q$s} (or $\unionsof{Q}$, for short) is a query of the form $q = q_1 \vee q_2 \vee \ldots \vee q_k$, where, for every $i \in [k]$, $q_i \in Q$. For a graph-database $\DBD$, we define $q(\DBD) = \bigcup_{i \in [k]} q_i(\DBD)$.

The results of this section can be summarised as follows.

\begin{theorem}\label{inclusionDigaramTheorem}
The inclusion-diagram of Figure~\ref{expPowerInclusionMap} is correct.
\end{theorem}

Next, we first discuss these results in Subsection~\ref{sec:exprpowerDiscussions}, and then, in the following subsections, we formally prove all the inclusions of Figure~\ref{expPowerInclusionMap}.

\begin{figure}
\begin{center}
\scalebox{1}{\includegraphics{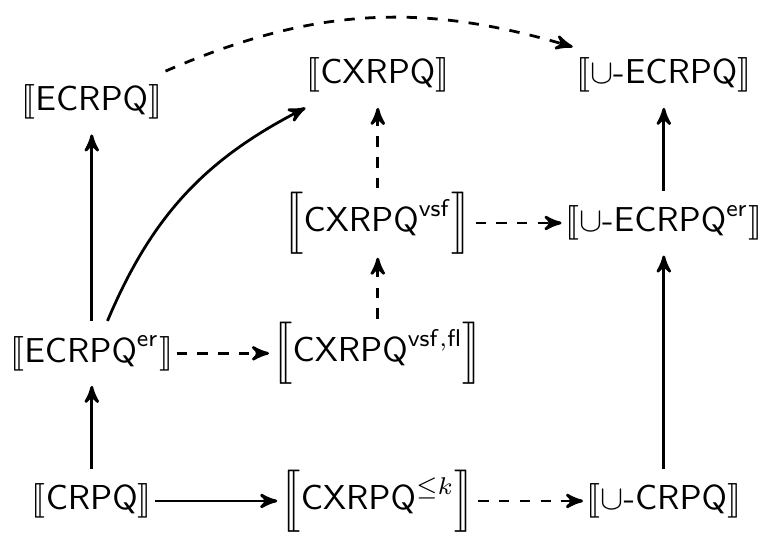}}
\end{center}
\caption{Illustration of the relations between the considered classes of conjunctive path queries. A dashed or solid arrow from $A$ to $B$ denotes $A \subseteq B$ or $A \subsetneq B$, respectively.}
\label{expPowerInclusionMap}
\end{figure}

\subsection{Discussion of Results}\label{sec:exprpowerDiscussions}

The diagram of Figure~\ref{expPowerInclusionMap} shows three vertical layers of query classes of increasing expressive power. More precisely, on the left side we have the \emph{classical} conjunctive path query classes $\CRPQ$, $\eqECRPQ$ and $\ECRPQ$. Then, our new fragments of $\CXRPQ$ follow. Finally, we have the classes of unions of the classical conjunctive path query classes. All these layers are naturally ordered by the subset relation, and these subset relations follow all directly by definition. \par
The more interesting inclusion relations are the vertical ones. The inclusion $\llbracket \eqECRPQ \rrbracket \subseteq \llbracket \vsfpbCXRPQ \rrbracket$ is as expected, but nevertheless points out that also quite restricted classes of $\CXRPQ$ still cover $\CRPQ$ that are extended by equality relations. However, this does not seem to be the case for the classes $\llbracket \bisCXRPQ{k} \rrbracket$, which nevertheless contains the class $\llbracket \CRPQ \rrbracket$.\par
Less obvious are the inclusions of the class $\llbracket \vsfCXRPQ \rrbracket$ in $\llbracket \unionsof{\eqECRPQ} \rrbracket$, and of the class $\llbracket \bisCXRPQ{k} \rrbracket$ in $\llbracket \unionsof{\CRPQ} \rrbracket$. Both of them are more or less a result from the upper bounds for these $\CXRPQ$-fragments shown in Sections~\ref{sec:vsfCXRPQ}~and~\ref{sec:boundedImageSize}. In a sense, this means that the queries of our $\CXRPQ$-fragments can be ``decomposed'' into unions of the more classical $\CRPQ$s and $\eqECRPQ$s. Thus, semantcially, the $\vsfCXRPQ$ and $\bisCXRPQ{k}$ can still be described in the formalisms of $\CRPQ$ and $\eqECRPQ$. However, there is a significant syntactical difference: our conversions of $\vsfCXRPQ$ or $\bisCXRPQ{k}$ into $\unionsof{\CRPQ}$ or $\unionsof{\eqECRPQ}$, respectively, require exponential (or even double exponential) size blow-ups.\par
It is to be expected, that $\vsfCXRPQ$ are strictly more powerful than $\eqECRPQ$. However, we can only show that $\CXRPQ$ are strictly more powerful than $\eqECRPQ$. On the other hand, that also $\bisCXRPQ{k}$s are strictly more powerful than $\CRPQ$s seems less obvious, since the string variables of $\bisCXRPQ{k}$s, which syntactically set them apart from $\CRPQ$s, can only range over finite sets of possible images. So it is more surprising that, in fact, even $\bisCXRPQ{1}$ contains queries that are not expressible as $\CRPQ$. 

\subsection{$\llbracket \CRPQ \rrbracket \subsetneq \llbracket \eqECRPQ \rrbracket \subsetneq \llbracket \ECRPQ \rrbracket$ and \\$\llbracket \unionsof{\CRPQ} \rrbracket \subsetneq \llbracket \unionsof{\eqECRPQ} \rrbracket \subsetneq \llbracket \unionsof{\ECRPQ} \rrbracket$}


\begin{theorem}\label{CRPQECRPQExpressivePowerTheorem}
$\llbracket \CRPQ \rrbracket \subsetneq \llbracket \eqECRPQ \rrbracket \subsetneq \llbracket \ECRPQ \rrbracket$.
\end{theorem}

\begin{proof}
The inclusions follow immediately, since $\CRPQ$ can be interpreted as $\eqECRPQ$ without any equality relations, and $\eqECRPQ \subseteq \ECRPQ$. We shall next prove that they are proper. \par
Let $\Sigma = \{\ta, \tb, \tc, \td\}$ and let $q_{\ta^n\tb^n}$ be the Boolean $\ECRPQ$ over $\Sigma$ defined by the graph-pattern $G_{q_{\ta^n\tb^n}} = (V_{q_{\ta^n\tb^n}}, E_{q_{\ta^n\tb^n}})$ with 
\begin{align*}
V_{q_{\ta^n\tb^n}} = \{&x, y_1, y_2, z, x', y'_1, y'_2, z'\}\,,\\
E_{q_{\ta^n\tb^n}} = \{&(x, \tc, y_1), (y_1, \ta^*, y_2), (y_2, \tc, z),\\
&(x', \td, y'_1), (y'_1, \tb^*, y'_2), (y'_2, \td, z')\}\,,
\end{align*}
and with only one equal-length relation (i.\,e., the relation $\{(u_1, u_2) \in (\Sigma^*)^2 \mid |u_1| = |u_2|\}$) that applies to the edges $(y_1, \ta^*, y_2)$ and $(y'_1, \tb^*, y'_2)$. See Figure~\ref{ECRPQseparatingFigure} for an illustration. We note that $\llbracket {q_{\ta^n\tb^n}} \rrbracket$ is the set of graph-databases $\DBD$ that contain (not necessarily distinct) vertices $u, v, u', v'$ and a path from $u$ to $v$ labelled with $\tc \ta^n \tc$ and a path from $u'$ to $v'$ labelled with $\td \tb^n \td$, respectively, for some $n \geq 0$.\medskip\\
\emph{Claim $1$}: $\llbracket q_{\ta^n \tb^n} \rrbracket \notin \llbracket \eqECRPQ \rrbracket$.\smallskip\\
\emph{Proof of Claim $1$}: For the sake of convenience, we relabel $q_{\ta^n\tb^n}$ to $q$ in the proof of the claim. We assume that there is a Boolean $q' \in \eqECRPQ$, such that $\llbracket q' \rrbracket = \llbracket q \rrbracket$. Moreover, let $q'$ be defined by a graph pattern $G_{q'} = (V_{q'}, E_{q'})$ with $E_{q'} = \{(\widehat{x}_i, \alpha_i, \widehat{y}_i) \mid i \in [m]\}$ and some equality relations. \par
For every $n \in \mathbb{N}$, let $\DBD_{n, n}$ be the graph-database given by two node-disjoint paths $(r_0, r_1, \dots, r_{n+2})$ and $(s_0, s_1, \dots, s_{n+2})$ labelled with $\tc \ta^n \tc$ and $\td \tb^n \td$, respectively. Obviously, for every $n \in \mathbb{N}$, $\DBD_{n, n} \in \llbracket q \rrbracket = \llbracket q' \rrbracket$, which means that there is at least one matching morphism $h$ for $q'$ and $\DBD_{n, n}$. In the following, for every $n \in \mathbb{N}$, let $h_n$ be some fixed matching morphism for $q'$ and $\DBD_{n, n}$. By the structure of $\DBD_{n, n}$, we also know that, for every $i \in [m]$, the arc $(\widehat{x}_i, \alpha_i, \widehat{y}_i)$ is matched to some sub-path of either $(r_0, r_1, \dots, r_{n+2})$ or $(s_0, s_1, \dots, s_{n+2})$; more precisely, there are $\ell_i, \ell'_i \in [n + 2] \cup \{0\}$ with $0 \leq \ell_i \leq \ell'_i \leq n + 2$ such that either $h_n(\widehat{x_i}) = r_{\ell_i}$ and $h_n(\widehat{y_i}) = r_{\ell'_i}$, or $h_n(\widehat{x_i}) = s_{\ell_i}$ and $h_n(\widehat{y_i}) = s_{\ell'_i}$. Now let $\{C_n, D_n\}$ be a partition of $[m]$ such that $i \in C_n$ if $(\widehat{x}_i, \alpha_i, \widehat{y}_i)$ is matched to some sub-path of $(r_0, r_1, \dots, r_{n+2})$ and $i \in D_n$ if $(\widehat{x}_i, \alpha_i, \widehat{y}_i)$ is matched to some sub-path of $(s_0, s_1, \dots, s_{n+2})$. In particular, we note that $(r_0, r_1, \dots, r_{n+2})$ only contains labels $\ta$ and $\tc$, while $(s_0, s_1, \dots, s_{n+2})$ only contains labels $\tb$ and $\td$. This means that for each single equality relation of $q'$ with arity $p$ that applies to a set $\{e_{j_1}, e_{j_2}, \ldots, e_{j_p}\}$ of arcs, there are three possibilities: (1) $\{j_1, j_2, \ldots, j_p\} \subseteq C_n$, (2) $\{j_1, j_2, \ldots, j_p\} \subseteq D_n$, or (3), for every $i \in [p]$, $h_{n}(\widehat{x}_{j_i}) = h_{n}(\widehat{y}_{j_i})$ (i.\,e., they cover paths of length $0$ that are labelled with $\eword$). \par
Since there is only a finite number of partitions of $[m]$ into two sets, there must be some $n_1, n_2 \in \mathbb{N}$ with $n_1 \neq n_2$ such that $C_{n_1} = C_{n_2}$ and $D_{n_1} = D_{n_2}$. We now define a morphism $h : V_{q'} \to \{r_0, r_1, \dots, r_{n_1+2}\} \cup \{s_0, s_1, \dots, s_{n_2+2}\}$ by setting, for every $i \in C_{n_1} = C_{n_2}$, $h(\widehat{x}_i) = h_{n_1}(\widehat{x}_i)$ and $h(\widehat{y}_i) = h_{n_1}(\widehat{y}_i)$, and, for every $i \in D_{n_1} = D_{n_2}$, $h(\widehat{x}_i) = h_{n_2}(\widehat{x}_i)$ and $h(\widehat{y}_i) = h_{n_2}(\widehat{y}_i)$. We note that $h$ is a matching morphism for $q'$ and $\DBD_{n_1, n_2}$. In particular, due to our observation from above, each equality relation is satisfied. Since $n_1 \neq n_2$, we have that $\DBD_{n_1, n_2} \notin \llbracket q' \rrbracket$, which is a contradiction. \hfill \qed \,(\emph{Claim $1$})\medskip\\
Let $q_{\ta^n\ta^n}$ be the $\eqECRPQ$ over $\Sigma$ that is defined in almost the same way as $q_{\ta^n\tb^n}$, with the only differences that we label the arc from $y'_1$ to $y'_2$ of the graph-pattern by $\ta^*$ instead of $\tb^*$ and that the binary equal-length relation on $(y_1, \ta^*, y_2)$ and $(y'_1, \tb^*, y'_2)$ becomes a binary equality relation on $(y_1, \ta^*, y_2)$ and $(y'_1, \ta^*, y'_2)$. \par
More formally, let $q_{\ta^n\ta^n} \in \eqECRPQ$ over $\Sigma = \{\ta, \tb, \tc, \td\}$ be defined by the graph-pattern $G_{q_{\ta^n\ta^n}} = (V_{q_{\ta^n\ta^n}}, E_{q_{\ta^n\ta^n}})$ with 
\begin{align*}
V_{q_{\ta^n\ta^n}} = \{&x, y_1, y_2, z, x', y'_1, y'_2, z'\}\,,\\
E_{q_{\ta^n\ta^n}} = \{&(x, \tc, y_1), (y_1, \ta^*, y_2), (y_2, \tc, z),\\
&(x', \td, y'_1), (y'_1, \ta^*, y'_2), (y'_2, \td, z')\}\,,
\end{align*}
and only one binary equality relation that applies to the edges $(y_1, \ta^*, y_2)$ and $(y'_1, \ta^*, y'_2)$. We note that $\llbracket q_{\ta^n\ta^n} \rrbracket$ is the set of graph-databases $\DBD$ that contain (not necessarily distinct) vertices $u, v, u', v'$ and a path from $u$ to $v$ labelled with $\tc \ta^n \tc$ and a path from $u'$ to $v'$ labelled with $\td \ta^n \td$, respectively, for some $n \geq 0$.\medskip\\
\emph{Claim $2$}: $\llbracket q_{\ta^n \ta^n} \rrbracket \notin \llbracket \CRPQ \rrbracket$.\smallskip\\
\emph{Proof of Claim $2$}: We assume that there is a $q' \in \CRPQ$, such that $\llbracket q' \rrbracket = \llbracket q \rrbracket$. Moreover, let $q'$ be defined by a graph pattern $G_{q'} = (V_{q'}, E_{q'})$ with $E_{q'} = \{(\widehat{x}_i, \alpha_i, \widehat{y}_i) \mid i \in [m]\}$. \par
We can now obtain a contradiction analogously as in to the proof of Claim $1$. In fact, the argument is almost the same, but we argue with graph-databases $\DBD_{n, n}$ given by two node-disjoint paths $(r_0, r_1, \dots, r_{n+2})$ and $(s_0, s_1, \dots, s_{n+2})$ labelled with $\tc \ta^n \tc$ and $\td \ta^n \td$, respectively. In general, the argument is simpler, because we do not have to take care of possible equality relations. \\\hfill \qed \,(\emph{Claim $2$})\medskip\\
This concludes the proof. 
\end{proof}

\begin{figure}[h]
\begin{center}
\scalebox{1.5}{\includegraphics{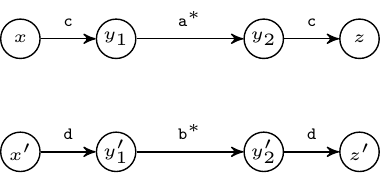}}
\end{center}
      	      	\caption{Illustration of the graph-pattern for $q_{\ta^n\tb^n}$.}
\label{ECRPQseparatingFigure}
\end{figure}

\begin{theorem}\label{UnionsOfCRPQECRPQExpressivePowerTheorem}
$\llbracket \unionsof{\CRPQ} \rrbracket \subsetneq \llbracket \unionsof{\eqECRPQ} \rrbracket \subsetneq \llbracket \unionsof{\ECRPQ} \rrbracket$.
\end{theorem}

\begin{proof}
The Inclusions follow by definition. We next show that the inclusion $\llbracket \unionsof{\eqECRPQ} \rrbracket \subseteq \llbracket \unionsof{\ECRPQ} \rrbracket$ is proper. To this end, we first recall the proof of Claim $1$ in the proof of Theorem~\ref{CRPQECRPQExpressivePowerTheorem}, which showed that for the query $q = q_{\ta^n\tb^n} \in \ECRPQ \subseteq \unionsof{\ECRPQ}$ (see also Figure~\ref{ECRPQseparatingFigure}), we have that $\llbracket q \rrbracket \notin \eqECRPQ$. We have demonstrated that if there is a query $q' \in \eqECRPQ$ with $\DBD_{n, n} \models q'$ for every $n \in \mathbb{N}$, then there are $n_1, n_2 \in \mathbb{N}$ with $n_1 \neq n_2$, such that $q$ can be matched to both $\DBD_{n_1, n_1}$ and $\DBD_{n_2, n_2}$ in such a way that the partition of the edges of $q'$ according to whether they are matched to the $\tc \ta^n \tc$ path or the $\td \tb^n \td$ path are exactly the same. This has lead to the contradiction that $\DBD_{n_1, n_2} \models q'$.\par
If we now instead consider a $q' \in \unionsof{\eqECRPQ}$, then $q'$ contains only a finite number of graph patterns $G^{(1)}_q, G^{(2)}_q, \ldots, G^{(\ell)}_q$, and, for every $n \in \mathbb{N}$, there is at least one $j \in [\ell]$, such that $G^{(j)}_q$ can be matched to $\DBD_{n, n}$. This implies again, that there are $n_1, n_2 \in \mathbb{N}$ with $n_1 \neq n_2$, such that, for some $j \in [\ell]$, $G^{(j)}_q$ can be matched to both $\DBD_{n_1, n_1}$ and $\DBD_{n_2, n_2}$ in such a way that the partition of the edges of $G^{(j)}_q$ according to whether they are matched to the $\tc \ta^n \tc$ path or the $\td \tb^n \td$ are exactly the same. Again, this leads to the contradiction that $G^{(j)}_q$ can be matched to $\DBD_{n_1, n_2}$ and therefore $\DBD_{n_1, n_2} \models q'$.\par
In order to show that the inclusion $\llbracket \unionsof{\eqECRPQ} \rrbracket \subseteq \llbracket \unionsof{\ECRPQ} \rrbracket$ is proper, we argue analogously, but with $q = q_{\ta^n\ta^n} \in \eqECRPQ \subseteq \unionsof{\eqECRPQ}$ (i.\,e., we extend the argument of the proof of Claim $2$ of the proof of Theorem~\ref{CRPQECRPQExpressivePowerTheorem} to the case of unions of queries just as it is done above with respect to Claim $1$ of the proof of Theorem~\ref{CRPQECRPQExpressivePowerTheorem}). 
\end{proof}

\subsection{$\llbracket \eqECRPQ \rrbracket \subseteq \llbracket \vsfpbCXRPQ \rrbracket \subseteq \llbracket \vsfCXRPQ \rrbracket \subseteq \llbracket \CXRPQ \rrbracket$}

All inclusion of $$\llbracket \eqECRPQ \rrbracket \subseteq \llbracket \vsfpbCXRPQ \rrbracket \subseteq \llbracket \vsfCXRPQ \rrbracket \subseteq \llbracket \CXRPQ \rrbracket$$ follow by definition, except the first one, which is due to the following Lemma~\ref{inclusioneqECRPQinvsfpbCXRPQLemma}. Note that this inclusion chain also implies the inclusion $\llbracket \eqECRPQ \rrbracket \subseteq \llbracket \CXRPQ \rrbracket$ depicted in Figure~\ref{expPowerInclusionMap}; its strictness is shown later in Subsection~\ref{sec:strictness}.


\begin{lemma}\label{inclusioneqECRPQinvsfpbCXRPQLemma}
$\llbracket \eqECRPQ \rrbracket \subseteq \llbracket \vsfpbCXRPQ \rrbracket$.
\end{lemma}

\begin{proof}
Let $q \in \eqECRPQ$ be of the form $q = \bar{z} \gets G_q$ with $E_q = \{e_i = (x_i, \alpha_i, y_i) \mid i \in [m]\}$ and let $\{E_{q, 1}, E_{q, 2}, \ldots, E_{q, t}\}$ with $E_{q, j} = \{e_{p_1}, e_{p_2}, \ldots, e_{p_{s_j}}\}$ be the partition of $E_q$ that represents the equality constraints. For every $j \in [t]$, we successively modify $q$ as follows. We replace $(x_{p_1}, \alpha_{p_1}, y_{p_1})$ by $(x_{p_1}, \beta, y_{p_1})$, where $\beta$ is a regular expression for $\bigcap^{s_j}_{i = 1} \lang(\alpha_{p_i})$, and, for every $i$ with $2 \leq i \leq p_{s_j}$, we replace $(x_{p_i}, \alpha_{p_i}, y_{p_i})$ by $(x_{p_i}, \Sigma^*, y_{p_i})$. 
We denote the $\eqECRPQ$ constructed in this way by $q'$ and we note that $q'$ is equivalent to $q$. Moreover, $q'$ is represented by a graph $(V_{q'}, E_{q'})$ and a partition $\{E_{q, 1}, E_{q, 2}, \ldots, E_{q, t}\}$ such that, for every $j \in [t]$, $E_{q, j} = \{(x_{p_1}, \beta_{p_1}, y_{p_1}), (x_{p_2}, \Sigma^*, y_{p_2}), \ldots, (x_{p_{s_j}}, \Sigma^*, y_{p_{s_j}})\}$ for some regular expression $\beta_{j_1}$. We can now translate $q'$ into a $q'' \in \vsfpbCXRPQ$ by replacing, for every $j \in [t]$, edge $(x_{j_1}, \beta_{j_1}, y_{j_1})$ by $(x_{j_1}, \varsz_{j}\{\beta_{j_1}\}, y_{j_1})$ and every edge $(x_{j_{\ell}}, \Sigma^*, y_{j_{\ell}})$, $2 \leq \ell \leq s_j$, by $(x_{j_{\ell}}, \varsz_j, y_{j_{\ell}})$. It can be easily verified that $q''$ is equivalent to $q$.\par
\end{proof}

\subsection{$\llbracket \ECRPQ \rrbracket \subseteq \llbracket \unionsof{\ECRPQ} \rrbracket$ and $\llbracket \vsfCXRPQ \rrbracket \subseteq \llbracket \unionsof{\eqECRPQ} \rrbracket$}

The inclusion $\llbracket \ECRPQ \rrbracket \subseteq \llbracket \unionsof{\ECRPQ} \rrbracket$ follows trivially by definition, whereas $\llbracket \vsfCXRPQ \rrbracket \subseteq \llbracket \unionsof{\eqECRPQ} \rrbracket$ is shown by the following Lemma~\ref{vstarfreeCXRPQToUnionOfECRPQLemma}.

\begin{lemma}\label{vstarfreeCXRPQToUnionOfECRPQLemma}
$\llbracket \vsfCXRPQ \rrbracket \subseteq \llbracket \unionsof{\eqECRPQ} \rrbracket$.
\end{lemma}

\begin{proof}
Let $q = \bar{z} \gets G_q$ with xregex $\bar{\alpha} = (\alpha_1, \alpha_2, \ldots, \alpha_m)$. By Lemmas~\ref{conjunctiveNormalFormLemmaStepOne},~\ref{uniqueDefinitionsLemma}~and~\ref{removeNonBasicLemma}, we can assume that $\bar{\alpha}$ is in normal form, i.\,e., it is such that, for every $i \in [m]$, $\alpha_i = \alpha_{i, 1} \altop \alpha_{i, 2} \altop \ldots \altop \alpha_{i, t_i}$ and every $\alpha_{i, j}$ for $j \in [t_i]$ is simple. \par
We can now transform $q$ into $q_1, q_2, \ldots, q_r \in \vsfCXRPQ$, such that, for every graph-database $\DBD$, $q(\DBD) = \bigcup^r_{i = 1} q_i(\DBD)$, and, for every $i \in [r]$, $q_i$ is simple. More precisely, the $q_j$ with $j \in [r]$ are obtained by considering, for every $i \in [m]$, all possibilities of replacing $\alpha_i$ by exactly one of the $\alpha_{i, 1}, \alpha_{i, 2}, \ldots, \alpha_{i, t_i}$, which obviously results in $\vsfCXRPQ$ that are simple (since $q$ is in normal form). \par
Next, for every $i \in [r]$, we can transform $q_i$ into an equivalent $q''_i \in \eqECRPQ$ as follows. Let $\bar{\alpha}^{(i)} = (\alpha^{(i)}_1, \alpha^{(i)}_2, \ldots, \alpha^{(i)}_m)$ be the conjunctive xregex of $q_i$. We first observe that, analogously to the proof of Lemma~\ref{simpleCXRPQLemma}, since $\bar{\alpha}^{(i)}$ is simple, we can replace definitions $\varsx\{\varsy\}$ and all occurrences of $\varsx$ by references of $\varsy$ without changing the set of conjunctive matches. Let the thus modified version of $\bar{\alpha}^{(i)}$ be denoted by $\bar{\alpha}'^{(i)} = (\alpha'^{(i)}_1, \alpha'^{(i)}_2, \ldots, \alpha'^{(i)}_m)$. For every $j \in [m]$, $\alpha'^{(i)}_{j} = \pi_{1} \pi_{2} \ldots \pi_{t}$, where each $\pi_\ell$ is a classical regular expression $\gamma$, a variable definition $\varsx\{\gamma\}$, where $\gamma$ is a classical regular expression, or a variable reference $\varsx$. Thus, for every $j \in [m]$, we can break up the edge labelled with $\alpha'^{(i)}_{j}$ into a path of size $t$ with the edge labels $\pi_\ell$. If we do this for all $\alpha'^{(i)}_{j}$, then we have turned $q_i$ into a $q'_i \in \vsfCXRPQ$, such that every edge is labelled by a classical regular expression, a variable definition over a classical regular expression, or a variable reference. For every variable $\varsx$, we now do the following. We replace the edge label $\varsx\{\gamma\}$ by $\gamma$ (since $q'_i$ is simple, there can be at most one such edge label) and all edge labels $\varsx$ by $\Sigma^*$ and add an equality constraint that applies to exactly the edges modified by this step. It can be easily seen that the thus obtained $q''_i \in \eqECRPQ$ is equivalent to $q_i$ (note that the tuple of output nodes $\bar{z}$ remains unchanged). 
\end{proof}

%

\subsection{$\llbracket \bisCXRPQ{k} \rrbracket \subseteq \llbracket \unionsof{\CRPQ} \rrbracket$ and $\llbracket \CRPQ \rrbracket \subseteq \llbracket \bisCXRPQ{k} \rrbracket$}

That, for every $k \geq 1$, $\llbracket \bisCXRPQ{k} \rrbracket \subseteq \llbracket \unionsof{\CRPQ} \rrbracket$, is due to the following Lemma~\ref{inclusionbisCXRPQunionsofCRPQLemma}. Moreover, that, for every $k \geq 1$, $\llbracket \CRPQ \rrbracket \subseteq \llbracket \bisCXRPQ{k} \rrbracket$, follows trivially by definition and the strictness of the inclusions is shown in Subsection~\ref{sec:strictness}.

\begin{lemma}\label{inclusionbisCXRPQunionsofCRPQLemma}
For every $k \geq 1$, $\llbracket \bisCXRPQ{k} \rrbracket \subseteq \llbracket \unionsof{\CRPQ} \rrbracket$.
\end{lemma}

\begin{proof}
Let $k \geq 1$ and let $q \in \bisCXRPQ{k}$ be defined by a graph pattern $G_q = (V_q, E_q)$ with $E_q = \{(x_i, \alpha_i, y_i) \mid i \in [m]\}$, where $\bar{\alpha} \in \conxregex_{\Sigma, \varset}$ with $\varset = \{\varsx_1, \varsx_2, \ldots, \varsx_n\}$. Now, for every $\bar{v} \in (\Sigma^{\leq k})^n$, let $q[\bar{v}]$ be a $\CRPQ$ with the property that, for every graph database $\DBD$, $q[\bar{v}](\DBD) = q^{\bar{v}}(\DBD)$. Such $q[\bar{v}]$ exist due to Lemma~\ref{checkCXRPQVariableMappingLemma}. This directly implies that, for every graph database $\DBD$, $q(\DBD) = \bigcup_{\bar{v} \in (\Sigma^{\leq k})^n} q[\bar{v}](\DBD)$. Moreover, we can conclude that $\llbracket q \rrbracket = \llbracket q' \rrbracket$, where $q' \in \unionsof{\CRPQ}$ is defined by $q' = \bigvee_{\bar{v} \in (\Sigma^{\leq k})^n} q[\bar{v}]$.
%
%
\end{proof}

\subsection{$\llbracket \bisCXRPQ{k} \rrbracket \neq \llbracket \CRPQ \rrbracket$ and $\llbracket \eqECRPQ \rrbracket \neq \llbracket \CXRPQ \rrbracket$}\label{sec:strictness}

Interestingly, we can also show that the expressive power of $\CXRPQ$ properly exceeds that of $\eqECRPQ$, and that the expressive power of $\bisCXRPQ{k}$, for every $k \geq 1$, properly exceeds that of $\CRPQ$. 

%

\begin{figure}
\begin{center}
\scalebox{1.4}{\includegraphics{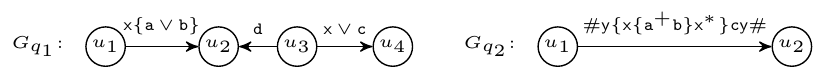}}
\end{center}
\caption{The graph patterns for $q_1$ and $q_2$.}
\label{sepExamplesFigure}
\end{figure}



\begin{lemma}\label{bisCXRPQCRPQstrictnessLemma}
For every $k \geq 1$, there is a Boolean $q \in \bisCXRPQ{k}$ with $\llbracket q \rrbracket \notin \llbracket \CRPQ \rrbracket$.
\end{lemma}

\begin{proof}
Let the Boolean $q_1 \in \CXRPQ$ over $\Sigma = \{\ta, \tb, \tc, \td\}$ be defined by the graph pattern $G_{q_1} = (V_{q_1}, E_{q_1})$ with $V_{q_1} = \{u_1, u_2, u_3, u_4\}$ and 
\begin{equation*}
E_{q_1} = \{(u_1, \alpha_1, u_2), (u_3, \alpha_2, u_2), (u_3, \alpha_3, u_4)\}\,, 
\end{equation*}
where $\alpha_1 = \varsx\{\ta \altop \tb\}$, $\alpha_2 = \td$ and $\alpha_3 = \varsx \altop \tc$ (see Figure~\ref{sepExamplesFigure}). We note that, for every $k \geq 1$ and every graph database $\DBD$, $q_1^k(\DBD) = q_1^1(\DBD) = q_1(\DBD)$. Thus, in order to prove the statement of the lemma for every $k \geq 1$, it is sufficient to show that $\llbracket q_1 \rrbracket \notin \llbracket \CRPQ \rrbracket$, where $q_1$ is interpreted as a $\CXRPQ$ without any restrictions.\par
For every $\sigma_1, \sigma_2 \in \Sigma$, let $\DBD_{\sigma_1, \sigma_2} = (V_{\sigma_1, \sigma_2}, E_{\sigma_1, \sigma_2})$ be a graph-database with $V_{\sigma_1, \sigma_2} = \{v_1, v_2, v_3, v_4\}$ and $$E_{\sigma_1, \sigma_2} = \{(v_1, \sigma_1, v_2), (v_3, \td, v_2), (v_3, \sigma_2, v_4)\}\,.$$ We note that $\DBD_{\sigma_1, \sigma_2} \models q_1$ for every $\sigma_1 \in \{\ta, \tb\}$ and $\sigma_2 \in \{\ta, \tb, \tc\}$ with $\sigma_1 = \sigma_2$ or $\sigma_2 = \tc$. We assume that there is a $q' \in \CRPQ$ with $q' \equiv q_1$, defined by the graph pattern $G_{q'} = (V_{q'}, E_{q'})$ with $E_{q'} = \{(x_i, \beta_i, y_i) \mid i \in [m]\}$. \par
If, for some $i \in [m]$, there is no $w_i \in \lang(\beta_i)$ with $|w_i|_{\ta} = 0$, then $\DBD_{\tb, \tb} \not\models q'$, which, since $\DBD_{\tb, \tb} \models q_1$, is a contradiction. Therefore, we can assume that, for every $i \in [m]$, there is some $w_i \in \lang(\beta_i)$ with $|w_i|_{\ta} = 0$ (note that $w_i = \eword$ is possible). Next, we consider a matching morphism $h$ for $q'$ and $\DBD_{\ta, \ta}$, which, since $\DBD_{\ta, \ta} \models q_1$, must exist. Let $A \subseteq [m]$ be exactly the set of $i \in [m]$ with $h(x_i) = v_1$ and $h(y_i) = v_2$. If $A = \emptyset$, then we can conclude that the edge $(v_1, \ta, v_2)$ is not part of any of the paths between some $h(x_i)$ and $h(y_i)$ that are necessary for $h$ being a matching morphism (this is due to the fact that in $\DBD_{\ta, \ta}$ there are no paths of length strictly greater than $2$). Consequently, we can remove $(v_1, \ta, v_2)$ from $\DBD_{\ta, \ta}$ in order to obtain a graph database $\DBD'$, such that $h$ would still be a matching morphism for $q'$ and $\DBD'$. This, however, is a contradiction, since $\DBD' \not \models q_1$. Thus, $A \neq \emptyset$. Now let $\DBD'_{\ta, \ta}$ be the graph database obtained from $\DBD_{\ta, \ta}$, by deleting the edge $(v_1, \ta, v_2)$ and, for every $i \in A$, adding a path labelled with $w_i$ from $v_1$ to $v_2$. We note that $h$ is still a matching morphism for $q'$ and $\DBD'_{\ta, \ta}$, since, for every $i \in A$, there is a path from $v_1$ to $v_2$ labelled with $w_i \in \lang(\beta_i)$, and, furthermore, as already observed above, the deleted edge $(v_1, \ta, v_2)$ was exclusively covered by edges $(x_i, \beta_i, y_i)$ with $i \in A$. However, it can be easily seen that $\DBD'_{\ta, \ta} \not \models q_1$, which is a contradiction.
%
%
\end{proof}




\begin{lemma}\label{strictnessCXRPQeqECRPQLemma}
There is a Boolean $q \in \CXRPQ$ such that $\llbracket q \rrbracket \notin \llbracket \eqECRPQ \rrbracket$.
\end{lemma}

\begin{proof}
Let $q_2$ be defined by a graph pattern with just a single edge $(u_1, \beta, u_2)$ with $\beta = \#\varsy\{\varsx\{\ta^+\tb\}\varsx^*\}\tc\varsy\#$ (see Figure~\ref{sepExamplesFigure}). We note that $\DBD \models q_2$ if and only if $\DBD$ contains a path labelled with $\#(\ta^{n_1}\tb)^{n_2} \tc (\ta^{n_1}\tb)^{n_2}\#$ for some $n_1, n_2 \geq 1$. Let us assume that there is some $q' \in \eqECRPQ$ with $q_2 \equiv q'$ and $q'$ is defined by a graph pattern $G_{q'} = (V_{q'}, E_{q'})$ with $E_{q'} = \{(x_i, \alpha_i, y_i) \mid i \in [m]\}$ and some equality relations. Moreover, for every $i \in [m]$, let $p_i$ be the pumping lemma constant of $\lang(\alpha_i)$ and let $p = \max\{p_i \mid i \in [m]\}$. We consider the graph database $\DBD = (V_{\DBD}, E_{\DBD})$ with $V_{\DBD} = \{v_0, v_1, \ldots, v_{t}\}$, where $t = 2(p^2m + pm) + 3$ and $(v_0, v_1, \ldots, v_{t})$ is a path labelled with $\#(\ta^{p}\tb)^{pm}\tc(\ta^{p}\tb)^{pm}\#$. \par
Since $\DBD \models q'$, there is a matching morphism $h$ for $q'$ and $\DBD$ with some matching words $(w_1, w_2, \ldots, w_m)$, such that, for every $i \in [m]$, $h(x_i) = v_{j_i}$ and $h(y_i) = v_{j'_i}$ with $0 \leq j_i \leq j'_i \leq t$. We now partition $[m]$ into $S = \{i \mid j'_i - j_i < 2p + 1\}$ and $L = [m] \setminus S$, i.\,e., $(x_i, \alpha_i, y_i)$ is matched to a \emph{long} sub-path of $(v_0, v_1, \ldots, v_{t})$ of length at least $2p + 1$ if $i \in L$ and to a \emph{short} sub-path of $(v_0, v_1, \ldots, v_{t})$ of length strictly less than $2p + 1$ otherwise. \par
If there is an arc $(v_r, \sigma, v_{r + 1})$ of the path $(v_0, v_1, \ldots, v_{t})$ that is not covered by some $(x_i, \alpha_i, y_i)$ (i.\,e., for every $i \in [m]$, it is not the case that $j_i \leq r < j'_i$), then we can contract nodes $v_r$ and $v_{r + 1}$ and $h$ is still a matching morphism for $q'$ and the thus modified graph database $\DBD'$, which is not in $\llbracket q_2 \rrbracket$ anymore. This can be seen by observing that removing a single symbol from a word $w \in \lang(\beta)$ yields a word that is not in $\lang(\beta)$ anymore. Therefore, we can assume that every arc $(v_r, \sigma, v_{r + 1})$ of $\DBD$ is covered by some $(x_i, \alpha_i, y_i)$, i.\,e., $j_i \leq r < j'_i$. Since, for every $i \in S$, at most $2p$ edges can be covered by $(x_i, \alpha_i, y_i)$, we also know that $L \neq \emptyset$ (since otherwise not all arcs are covered).\par
We now modify $\DBD$ as follows. For every $i \in [m]$, we add a \emph{shortcut} from $v_{j_i}$ to $v_{j'_i}$, which is a new path of the same length and with the same label as the path $(v_{j_i}, v_{j_i + 1}, \ldots, v_{j'_i})$. In particular, we note that for every $i \in [m]$, the labels of the shortcuts are identical to the matching words of $h$. \par
We now pump some of the shortcuts depending on the equality relation of $q'$ as follows. Assume that $A \subseteq [m]$ represents an equality relation of $q'$, i.\,e., exactly the edges $\{(x_i, \alpha_i, y_i) \mid i \in A\}$ are subject to the equality relation. This also means that either all $(x_i, \alpha_i, y_i)$ with $i \in A$ cover a short path, i.\,e., $A \subseteq S$, or all $(x_i, \alpha_i, y_i)$ with $i \in A$ cover a long path, i.\,e., $A \subseteq L$. Moreover, as observed above, $L \neq \emptyset$, so there is at least one such equality relation $A$ (note that we assume that the equality relations are represented by a partition of the edge-set, i.\,e., every edge is subject to exactly one equality relation, possibly a unary one). \par
Recall that $(w_1, w_2, \ldots, w_m)$ are the matching words for $h$. Thus, for every $i \in A$, $w_i$ is the label of the sub-path $(v_{j_i}, v_{j_i + 1}, \ldots, v_{j'_i})$. Since $h$ is a matching morphism and since we have the equality relation represented by $A$, we know that, for some $u$, we have $u = w_i$ for every $i \in A$; moreover, the corresponding shortcuts for edges $(x_i, \alpha_i, y_i)$ with $i \in A$ are also all labelled with $u$. Since $A \subseteq L$, we have $|u| \geq 2p + 1$, which means that $u = u' \ta^{p} u''$. Hence, for every $i \in A$, there is a $d_i$ such that $u' \ta^{p + \delta d_i} u'' \in \lang(\alpha_i)$ for every $\delta \geq 0$. This means that, for every $i \in A$, $w' = u' \ta^{p + d} u'' \in \lang(\alpha_i)$, where $d = \Pi_{i \in A} d_{i}$. Consequently, we can pump all shortcuts for the edges $(x_i, \alpha_, y_i)$ with $i \in A$ by the factor $\ta^d$, i.\,e., we replace them by paths of length $|u| + d$ labelled by $w'$. We repeat this pumping-step with respect to all equality relations that refer to edges that cover long paths. After this modification, we have the property that in the tuple $(w_1, w_2, \ldots, w_m)$ of matching words for $h$, we can arbitrarily replace some $w_i$ by the label of the corresponding shortcut for $(x_i, \alpha_i, y_i)$ (regardless of whether it has been pumped or not) and still $h$ is a matching morphism with respect to this modified tuple of matching words. \par
We now choose an arbitrary edge $(v_\ell, \sigma, v_{\ell + 1})$ of the original path $(v_0, \ldots, v_t)$, which is only covered by edges $(x_i, \alpha_i, y_i)$ with $i \in L$, which means that their shortcuts have been pumped. Such an edge must exist, since otherwise all edges are covered by edges from $S$, which is not possible. Then, we delete this edge and we denote the obtained graph database by $\DBD'$. After this modification, due to the shortcuts, the matching morphism $h$ must still be a valid matching morphism for $q'$ and $\DBD$, i.\,e., $\DBD \models q'$. We now conclude the proof by showing that $\DBD' \not \models q_2$, which clearly is a contradiction.\par
For $\DBD' \models q_2$, there must be a path in $\DBD'$ that is labelled by a word of the form $\#(\ta^{n_1}\tb)^{n_2} \tc (\ta^{n_1}\tb)^{n_2}\#$ for some $n_1, n_2 \geq 1$. We note that this is only possible for paths from $v_0$ to $v_t$. Now let us consider an arbitrary path from $v_0$ and $v_{t}$ in $\DBD'$. Since we deleted the edge $(v_\ell, \sigma, v_{\ell + 1})$, this path must use at least one shortcut that corresponds to an edge $(x_i, \alpha_i, y_i)$ with $j_i \leq \ell < j'_i$. However, by our choice of $(v_\ell, \sigma, v_{\ell + 1})$, all such shortcuts have been pumped, which means that the path is labelled with a word $\widehat{w}$ that can be obtained from $\#(\ta^{p}\tb)^{pm}\tc(\ta^{p}\tb)^{pm}\#$ by pumping some unary factors over $\ta$. Furthermore, it is not possible that all maximal unary factors over $\ta$, i.\,e., factors of the form $\# \ta^p \tb$, $\tb \ta^p \tb$ or $\tc \ta^p \tb$, have been pumped, since the considered path can take at most $m$ shortcuts.
\end{proof}

\section{Conclusions and Open Problems}

The fact that evaluation for $\CXRPQ$ is at least $\pspaceclass$-hard even in data-complexity is reason enough to look at fragments of $\CXRPQ$ instead. Nevertheless, an upper bound for evaluating unrestricted $\CXRPQ$s would be interesting from a theoretical point of view; similarly, a lower bound for $\vsfCXRPQ$-evaluation in data-complexity would be interesting. Several of our results implicitly pose conciseness questions. Each $\bisCXRPQ{k}$ can be represented as the union of $\bigO(|\Sigma| + 1)^{nk})$ many $\CRPQ$s, and each $\vsfCXRPQ$ can be represented as the union of exponentially many $\eqECRPQ$s of exponential size. Are these exponential blow-ups necessary? Theorem~\ref{upperBoundsBISCXRPQTheorem} gives some evidence that for $\bisCXRPQ{k}$ this is the case.\par
All our algorithms for $\booleProb$ can be extended to the problem $\checkProb$. With respect to also extracting paths instead of only nodes from the graph-database, we can use the general techniques of~\cite{BarceloEtAl2012} to some extent. More precisely, for a $q \in \vsfCXRPQ$ (or $q \in \bisCXRPQ{k}$), a graph database $\DBD$ and a tuple $\bar{t} \in (V_{\DBD})^\ell$, our evaluation algorithms can be adapted to produce an automaton that represents all tuples of paths corresponding to the matching morphisms of $q$ and $\DBD$, but they are rather large. A thorough analysis of $\CXRPQ$-fragments with respect to how they can be used as queries that also extract path-variables is left for future work.\par
Interestingly enough, restrictions that lower the complexity for matching xregex to strings do not seem to help at all if we use them for querying graphs, and vice versa.
In the string case, the $\npclass$-complete matching problem for xregex trivially becomes polynomial-time solvable, if the number of variables is bounded by a constant (see~\cite{Schmid2016}), while this does not help for graphs (see Theorem~\ref{CXRPQDataComplexityHardnessTheorem}). Variable-star freeness (Section~\ref{sec:vsfCXRPQ}) or bounding the image size (Section~\ref{sec:boundedImageSize}) has a positive impact with respect to graphs. However, the matching problem for xregex remains $\npclass$-hard, even if we require variable-star freeness \emph{and} that variables can only range over words of length at most $1$ (see~\cite[Theorem~$3$]{FernauSchmid2015}).\smallskip\\

\subsection{Acknowledgments}

The author thanks Nicole Schweikardt for helpful discussions. The author is supported by DFG-grant SCHM 3485/1-1.

\bibliographystyle{plain}
\bibliography{bibfile}

\end{document}